\documentclass[aos]{imsart}
\RequirePackage{amsthm,amsmath,amsfonts,amssymb,amstext,booktabs,bm,color,enumitem,epsfig,epsf,float,graphics,latexsym,lscape,mathrsfs,multirow,mathtools,psfrag,rotate,subcaption,xcolor,xr,xspace}
\RequirePackage[authoryear]{natbib}%% uncomment this for author-year citations
\RequirePackage[colorlinks,citecolor=blue,urlcolor=blue]{hyperref}
\RequirePackage{graphicx}

\startlocaldefs
\graphicspath{{figs/}}

\theoremstyle{plain}

\newtheorem{theorem}{Theorem}[section]

\theoremstyle{remark}
\newtheorem{assumption}{Assumption}

\definecolor{darkgreen}{rgb}{0, 0.4, 0.06} 

\newtheorem{Lemma}{Lemma}
\newtheorem{Corollary}{Corollary}
\newtheorem{Proposition}{Proposition}
\newcommand{\bea}{\begin{eqnarray*}}
\newcommand{\eea}{\end{eqnarray*}}
\newcommand{\be}{\begin{eqnarray}}
\newcommand{\ee}{\end{eqnarray}}
\newcommand{\ed}{\end{document}}

\newcommand{\btab}{\begin{tabular}}
\newcommand{\etab}{\end{tabular}}

\newcommand{\la}{\label}
\newcommand{\bi}{\begin{itemize}}
\newcommand{\ei}{\end{itemize}}
\newcommand{\bfi}{\begin{figure}}
\newcommand{\efi}{\end{figure}}
\newcommand{\ben}{\begin{enumerate}}
\newcommand{\een}{\end{enumerate}}
\newcommand{\bay}{\begin{array}}
\newcommand{\eay}{\end{array}}

\newcommand{\nn}{\nonumber}

\newcommand{\Var}{\mathrm{Var}}

\def\F{Fr\'{e}chet\xspace}
\def\o{\omega}
\def\O{\Omega}
\def\po{\Omega_P}
\def\bco{\iffalse}

\def\Cov{{\rm Cov}}
\def\var{{\rm var}}

\def\diag{{\rm diag}}
\def\diam{{\rm diam}}
\def\trace{{\rm trace}}

\def\expit{\hbox{expit}}
\def\ci{\cite}
\def\cp{\citep}
\def\fo{F_\omega}
\def\hfo{{\widehat{F}_\omega}}
\def \I {\mathbb{I}}
\def\tps{^\top}
\def \mcF {\mathcal{F}}
\def\dep{distance profile\xspace}
\def\deps{distance profiles\xspace}

\def\mp{\mu_\oplus}
\def\hmp{\hat{\mu}_\oplus}
\def\C{\text{Cov}_{\Omega}}
\def\RR{\rho_{\Omega}}
\def\ao{\argmin_{\om \in \Om}}

\def\sn{\sum_{i=1}^n}
\def\snn{\sum_{i=1}^{n-1}}
\def\VF{\hat{V}_F}
\def\wc{\rightsquigarrow}

\def\M{\mathcal{M}_\oplus}
\def\hM{\widehat{\mathcal{M}}_\oplus}
\def\tM{\widetilde{\mathcal{M}}_\oplus}
\def\op{\omega_\oplus}
\def\hop{\widehat{\omega}_\oplus}
\def\tiop{\widetilde{\omega}_\oplus}

\def\obj{X}
\def\subidx{i}

\def\fobj{F_{\obj}}
\def\fobji{F_{\obj_{\subidx}}}
\def\dveck{Q} % distribution of (d(w_1,X),\dots,d(w_k,X))
\def\hfi{\widehat{F}_{\obj_{\subidx}}}

\def\hfmean{\widehat{F}_{\oplus}}
\def\rank{R} % rank
\def\hrank{\widehat{\rank}}
\def\trank{\widetilde{\rank}}

\newcommand{\E}{\mathbb{E}}
\def\prob{\mathbb{P}}
\def\expect{\mathbb{E}}

\def\inv{^{-1}}
\def\half{^{1/2}}
\def\diffop{\mathrm{d}}

\def\wsp{\mathcal{W}}
\def\dwass{d_W}
\def\dltwo{d_{L^2}}

\renewcommand{\hat}{\widehat}

\def\Om{\Omega}
\def\om{\omega}

\def\clust{S}

\def\qt{q}

\def\dfrob{d_F}

\def\dsphe{d_S}
\def\sphe{\mathcal{S}^2}
\def\splx{\Delta^2}

\def\eps{\varepsilon}

\newcommand{\bc}{\begin{center}}
\newcommand{\ec}{\end{center}}
\newcommand{\bgt}{\begin{equation}\begin{gathered}}
\newcommand{\egt}{\end{gathered}\end{equation}}
\newcommand{\bal}{\begin{equation}\aligned}
\newcommand{\eal}{\endaligned\end{equation}}
\def\ra{\rightarrow}
\def\real{\mathbb{R}}

\DeclareMathOperator*{\argmax}{argmax}
\DeclareMathOperator*{\argmin}{argmin}
\DeclareMathOperator*{\argsup}{argsup}
\DeclareMathOperator*{\arginf}{arginf}

\bibpunct{(}{)}{;}{a}{}{,}

\def\rdim{p}
\def\DP{DP\xspace}

\def\hf[#1#2#3#4]{\hat{F}^{#1}_{#2_{#3}}(#4)}
\def\f[#1#2#3#4]{{F}^{#1}_{#2_{#3}}(#4)}
\def\hFinv[#1#2#3#4]{(\hat{F}^{#1}_{#2_{#3}})^{-1}(#4)}
\def\denom[#1]{\frac{1}{#1}}
\def\add[#1#2]{\sum_{#1=1}^#2}
\def\addneq[#1#2]{\sum_{#1\neq #2}}
\def\ind[#1]{\mathbb{I}\left(#1\right)}
\def\wght[#1#2#3]{\hat{w}_{#1_#2}\left( #3 \right)}
\def\rwght[#1#2#3#4]{\hat{w}_{#1_#2}^{#4}\left( #3 \right)}
\def\wt[#1#2#3]{{w}_{#1_#2}\left( #3 \right)}
\def\hfin[#1#2]{\hat{F}^{in}_{#1}\left( #2 \right)}
\def\hfout[#1#2]{\hat{F}^{out}_{#1}\left( #2 \right)}
\def\fin[#1#2]{{F}^{in}_{#1}\left( #2 \right)}
\def\fout[#1#2]{{F}^{out}_{#1}\left( #2 \right)}
\newcommand{\df}[1]{\,\mathrm{d}{#1}}
\newcommand{\mcM}{\mathcal{M}}
\newcommand{\uijk}{U^{X,X}_{X_i,X_j,X_k}(u)}
\newcommand{\uijj}{U^{X,X}_{X_i,X_j,X_j}(u)}
\newcommand{\uxyy}{U^{Y,Y}_{X_i,Y_j,Y_k}(u)}
\newcommand{\uxxy}{U^{X,Y}_{X_i,X_j,Y_k}(u)}
\newcommand{\uxww}{U^{Y,Y}_{X_i,\omega_1,\omega_2}(u)}

\def\I{\mathrm{(I)}}
\def\IA{\mathrm{(IA)}}
\def\IB{\mathrm{(IB)}}
\def\IC{\mathrm{(IC)}}
\def\II{\mathrm{(II)}}
\def\IIA{\mathrm{(IIA)}}
\def\IIB{\mathrm{(IIB)}}

\def\III{\mathrm{(III)}}
\def\IIIA{\mathrm{(IIIA)}}
\def\IIIB{\mathrm{(IIIB)}}

\def\pone{\Delta_{1}} % mean shift of mvn r.v.
\def\ptwo{\Delta_{2}} % scale change of mvn r.v.
\def\pthr{\Delta_{3}} % mean shift of r.v. following  Gaussian mixtures
\def\pfou{\nu} % df of t distn
\def\pfiv{\delta_{1}} % mean shift of mean of bivariate Gaussian distns
\def\psix{\delta_{2}} % scale change of mean of bivariate Gaussian distns
\def\psev{\gamma} % network parameter

\def\references{\bibliography{depth}}
\def\ed{end{document}}
\endlocaldefs

\begin{document}

\begin{frontmatter}
\title{Metric Statistics: Exploration and Inference for Random Objects With Distance Profiles\thanksref{t1}}
%\title{A sample article title with some additional note\thanksref{t1}}
\runtitle{Metric Statistics}
\thankstext{T1}{This article is an expanded version of the Rietz Lecture delivered on June 27, 2022 by H.G.M.  at the IMS Meeting in London and is an updated version of the paper posted at arXiv:2202.06117.}
%\thankstext{T1}{A sample additional note to the title.}
% \newcommand\CoAuthorMark{\footnotemark[\arabic{footnote}]}

\begin{aug}
%%%%%%%%%%%%%%%%%%%%%%%%%%%%%%%%%%%%%%%%%%%%%%%
%% Only one address is permitted per author. %%
%% Only division, organization and e-mail is %%
%% included in the address.   %%
%% Additional information can be included in %%
%% the Acknowledgments section if necessary. %%
%% ORCID can be inserted by command:  %%
%% \orcid{0000-0000-0000-0000}  %%
%%%%%%%%%%%%%%%%%%%%%%%%%%%%%%%%%%%%%%%%%%%%%%%
% \author[A]{\fnms{Paromita}~\snm{Dubey}\ead[label=e1]{paromita@marshall.usc.edu}}\footnote{Contributed equally to the paper.},
% \author[B]{\fnms{Yaqing}~\snm{Chen}\ead[label=e2]{yqchen@stat.rutgers.edu}}\protect\CoAuthorMark
% \and
% \author[C]{\fnms{Hans-Georg}~\snm{M\"uller}\ead[label=e3]{hgmueller@ucdavis.edu}}
\author[A]{\fnms{Paromita}~\snm{Dubey}\ead[label=e1]{paromita@marshall.usc.edu}}\thanksref{t2},
\author[B]{\fnms{Yaqing}~\snm{Chen}\ead[label=e2]{yqchen@stat.rutgers.edu}}\thanksref{t2}
\and
\author[C]{\fnms{Hans-Georg}~\snm{M\"uller}\ead[label=e3]{hgmueller@ucdavis.edu}}
\thankstext{t2}{Contributed equally to the paper.}
%%%%%%%%%%%%%%%%%%%%%%%%%%%%%%%%%%%%%%%%%%%%%%
%% Addresses    %%
%%%%%%%%%%%%%%%%%%%%%%%%%%%%%%%%%%%%%%%%%%%%%%
\address[A]{Department of Data Sciences and Operations, Marshall School of Business, University of Southern California\printead[presep={,\ }]{e1}}
\address[B]{Department of Statistics, Rutgers University\printead[presep={,\ }]{e2}}
\address[C]{Department of Statistics, University of California, Davis\printead[presep={,\ }]{e3}}
\end{aug}

\begin{abstract}
This article provides an overview on the statistical modeling of complex data as increasingly encountered in modern data analysis. It is argued that such data can often be described as elements of a metric space that satisfies certain structural conditions and features a probability measure. We refer to the random elements of such spaces as random objects and to the emerging field that deals with their statistical analysis as metric statistics. Metric statistics provides methodology, theory and visualization tools for the statistical description, quantification of variation, centrality and quantiles, regression and inference for populations of random objects, inferring these quantities from available data and samples. In addition to a brief review of current concepts, we focus on distance profiles as a major tool for object data in conjunction with the pairwise Wasserstein transports of the underlying one-dimensional distance distributions. These pairwise transports lead to the definition of intuitive and interpretable notions of transport ranks and transport quantiles as well as two-sample inference. An associated profile metric complements the original metric of the object space and may reveal important features of the object data in data analysis.  We demonstrate these tools for the analysis of complex data through various examples and visualizations. 
\end{abstract}


\end{frontmatter}
%%%%%%%%%%%%%%%%%%%%%%%%%%%%%%%%%%%%%%%%%%%%%%
%% Please use \tableofcontents for articles %%
%% with 50 pages and more   %%
%%%%%%%%%%%%%%%%%%%%%%%%%%%%%%%%%%%%%%%%%%%%%%
%\tableofcontents

\section{Introduction}\label{sec:intro}
We delineate in this article an emerging field of statistics that provides models, methods and theory for complex data situated in metric spaces $(\Om,d)$ with a metric $d$. We refer to this field as {\it metric statistics}. Throughout it is assumed that the metric spaces where the data are situated are separable and endowed with a probability measure $P$. We refer to random variables taking values in such metric spaces as {\it random objects}, adopting the name from a previous review and perspective \cp{mull:16:9}. 

The motivation to address the challenges posed by non-Euclidean data and to study common features of such data and techniques that are applicable across many metric spaces comes from data analysis, where increasingly complex data objects are encountered. Statistical analysis means that the emphasis is on statistical methods that evolved from and have counterparts in classical Euclidean statistics, are interpretable rather than black-box approaches, and are amenable to uncertainty quantification and inference.  The need for such methodology has not gone unnoticed and over the last two decades various groups of statisticians have come up with interesting and important ideas about the  handling of  such data. This includes object-oriented data analysis with roots in statistics for manifold-valued data, shape analysis and geometric statistics and related  ideas for visualization and modeling \cp{wang:07:1,dryd:09,marr:21,huck:21}, 
and also symbolic data analysis, where various subproblems have been emphasized such as data that consist of intervals \cp{bill:03}.

A distinctive feature of metric statistics that differentiates it from classical as well as geometric statistics is the non-reliance on local or global Euclidean or manifold structure. While for some spaces local linearizations may exist, as exemplified by one-dimensional distributional data with the 2-Wasserstein metric,  where one can use Riemannian structure to define $L^2$ tangent spaces \cp{bigo:17,mull:21:9}, 
these are often only of limited utility; for example inverse maps from the linear spaces back to the metric space usually are not well defined on the entire linear approximation space.  The same holds for linear embeddings into a subset of a Hilbert space obtained through kernel maps \cp{scho:38,sejd:13}, although there exist specific invertible maps to a Hilbert space for special nonlinear spaces, which however induce  metric 
distortions \cp{mull:16:1}. The lack of Euclidean
structure in general metric spaces poses challenging problems for statistical theory, methodology and data analysis of random objects and essentially requires to rethink basic notions of mean, variation, regression, inference and other key statistical techniques. The overall goal is to arrive at a  principled, theory-supported and comprehensive toolkit for the analysis of samples of random objects.

After a brief review, we focus here on distance profiles, a basic tool that assigns a one-dimensional distribution to each element of the underlying metric space $(\Om,d)$. Distance profiles are the distributions of the distances of each element to a random object in the space $\Om$ and are determined by the underlying probability measure $P$ on $\Om$. As we will show, distance profiles not only reflect but indeed characterize $P$ under some regularity assumptions. In all of the following, we will assume that one has a sample of i.i.d. random objects drawn from $P$. Empirical estimates for the distance profiles are then simply obtained using the empirical distribution of the distances of any given element of $\Om$ to all other elements, either to all elements in the population or in the empirical version to the other sample elements. We will illustrate this idea also for the simple and familiar special case of  Euclidean data; in all scenarios, distance profiles always correspond to one-dimensional distributions. 
 
Distance profiles have multiple applications that we explore in this article. First of all, they aid the geometric exploration of random objects in $(\Om, d)$ under the measure $P$. Second, since distance profiles are always one-dimensional distributions, we can define a new dissimilarity measure on $\Om$ 
by adopting a metric on the space of one-dimensional distributions which is then applied to the distributional distance of the distance profiles of the two elements. This dissimilarity measure depends on both the original metric in the space $(\Om,d)$ as well on the metric adopted for one-dimensional distributions, where here we adopt the 2-Wasserstein metric as the metric in the distributional space of distance profiles. 

Third, the pairwise transports that result from adopting the 2-Wasserstein metric for the space of distance profiles make it possible to define  novel notions of transport 
centrality and associated transport ranks. These serve to quantify the centrality of objects and provide the basis for a 
partial ordering of random objects and resulting visualizations. 
Fourth, transport centrality and transport ranks can be harnessed to define transport quantiles as the set of elements of $\Om$ with transport ranks such that the elements with lower ranks have a probability mass bounded by the prespecified quantile level.

Fifth, we demonstrate how distance profiles across two samples can be used to 
test whether the probability measures that generate the samples are identical. This relies on the fact that distance profiles characterize the underlying probability measures; we note that related ideas on inference based on distance profiles  as those presented in Section~\ref{sec:test} below, resulting from seemingly independent work, were recently published \cp{wang:23:3}. 

Distance profiles thus emerge as a powerful tool to characterize random objects. As we show in the following Section \ref{sec:review}, they are natural extensions of some basic ideas of how to quantify the variation of random objects. Section \ref{sec:review} contains a brief review of some of the basic concepts of metric statistics, including \F and transport regression. Distance profiles and how they give rise to transport ranks and quantiles and notions such as most central points and the properties of these concepts will be the theme of Section~\ref{sec:method} and Section~\ref{sec:prop}. 
Further connections to applications in inferences, specifically distance profile based inference,  will be discussed in Section~\ref{sec:test}, followed by 
simulation studies and applications to age-at-death distributions of human mortality, U.S. energy generation data and functional connectivity networks based on fMRI data in Section~\ref{sec:data}. We conclude with a discussion on the choice of metrics and other topics in Section~\ref{sec:disc}.  Auxiliary results and proofs as well as additional simulations and data examples are provided in the Supplement. 

\section{Review of Basic Notions for Samples of Random Objects}\label{sec:review}

Random objects encompass the usual random variables that take values in spaces $\real^p$ as encountered in classical statistics and also random functions in Hilbert spaces $L^2$, which is the realm of functional data analysis and where one still has linear structures,  inner products and linear operators  \cp{hsin:15, mull:16:3}. Other well-studied classes of random objects are data on Riemannian
manifolds, notably  spheres, which also appear in shape analysis \cp{jung:12,dryd:16} and where surprising smeariness results were obtained in recent developments on  the limit theory for \F means 
\cp{eltz:19}. 

A recently emerging subarea of metric statistics is distributional data analysis, where the atoms of a sample are distributions. These may be directly observed or more commonly indirectly through the data that each distribution generates. In earlier approaches samples of distributions were treated 
as functional data \cp{knei:01}, but while density or distribution functions can be considered as elements of the function space $L^2$, this approach is suboptimal since distributional objects lie on a constrained submanifold, for example densities are non-negative and
integrate to 1. Taking these constraints fully into account for statistical analysis motivates  distributional data analysis 
\cp{mata:21,pete:22,ghos:23}. Linearization approaches for distributional data include the Bayes space transformation
\cp{hron:16}, which is based on the Aitchison geometry \cp{aitc:86}, however does not yield a 1:1 map,  and a class of 1:1 transformations to linear spaces that includes the log quantile density (lqd) and log hazard transformations \cp{mull:16:1}. More recent approaches have used local linearizations through the geometry of the Wasserstein manifold \cp{mull:21:9, pego:22} and fully intrinsic optimal transport models that do not rely on any ambient $L^2$ space \cp{mull:23:2,ghod:23}. 
In distributional data analysis, the metric space  $\Om$ is the space of distributions, which are often assumed to have a finite domain and to be continuous, and $d$ is an appropriate metric. For statistical analysis in the case of one-dimensional distributions the 2-Wasserstein metric has become popular, not least due to its practical appeal in data analysis \cp{bols:03}. 
For probability measures $\mu,\nu$ with distribution functions $F_{\mu}, F_{\nu}$, the 2-Wasserstein distance is given simply as the $L^2$ distance of the quantile functions,
\bal \label{eq:dwass}
\dwass(\mu,\nu) = \left(\int_0^1 \left[F_{\mu}\inv(u)-F_{\nu}\inv(u)\right]^2\diffop u\right)\half = \dltwo(F_{\mu}\inv,F_{\nu}\inv). \eal

For the case of multivariate distributions, quantile functions do not exist and the 2-Wasserstein metric is more directly tied to optimal transport in the Monge--Kantorovich transportation problem, where the Kantorovich version \cp{kant:06:1} is 
\be \label{m2w} 
\dwass^2(\mu,\nu) = \inf_{\wp \in \mathscr{P}(\mu,\nu)} \expect_{(X,Y)\sim \wp}\|X-Y\|^2.
\ee
Here $X,Y$ are random variables in $\real^d$, $\mu,\nu$ are probability measures supported on a set $D \subset \real^d$, and $\mathscr{P}(\mu,\nu)$ is the space of joint probability measures on $D\times D$ with marginals $\mu$ and $\nu$. If the probability distributions are absolutely continuous, this is equivalent to finding the optimal transport in Monge's version 
\be \label{monge}
T^{*}(\mu,\nu)= \arginf_{T:\, T_\# \mu = \nu}\expect_{X\sim \mu}\|X-T(X)\|^2.
\ee
Here the infimum is taken over all push-forward Borel maps $D \rightarrow D$ that map $\mu$ to $\nu$. The push-forward map $T$ applied to a probability measure $\nu_1$ on $D$ yields $\nu_2= T_{\#}\nu_1$, defined as the measure with 
$\nu_2(A) = \nu_1(T^{-1}(A))$ for any measurable set $A\subset D$. If it exists, this minimizer is the optimal transport map and the minimizing value is the Wasserstein metric $d_W$, 
which coincides with the definition in \eqref{eq:dwass} for the special case of univariate distributions. 

In the multivariate case, this minimization problem is 
computationally challenging and therefore often replaced by a relaxed version, e.g., the Sinkhorn minimization \cp{cutu:13}, 
but then depends critically on regularization parameters. The statistical motivation to use the Wasserstein metric for multivariate distributional data is also less compelling than for the one-dimensional case. Therefore it makes sense to use a simpler metric for this case. Options include the sliced Wasserstein metric \cp{kolo:19} or the
Fisher--Rao metric, which does not have quite the appeal of the Wasserstein metric with its connection to 
optimal mass transport, but is easy to compute in any dimension when densities exist, as it is the geodesic distance for the square roots of densities. These square roots are situated on the Hilbert sphere, whence for measures $\nu_1, \nu_2$ with densities $f_{\nu_1},f_{\nu_2}$ the metric is 
$$d_{FR}(\nu_1,\nu_2)=\text{arccos} \left\{ \int\sqrt{f_{\nu_1}(x)f_{\nu_2}(x)}\, \diffop x \right\}.$$ 
When adopting the Fisher--Rao metric, multivariate distributional data and spherical data as are commonly encountered in directional data analysis can be viewed in a unified framework of spherical data, e.g., in time series analysis \cp{mull:23}. The spherical framework also encompasses compositional data when using the square root transformation for proportions \cp{scea:14}. 

Other important classes of random objects include covariance matrices and surfaces \cp{pigo:14,zeme:19}, 
networks \cp{seve:22,mull:22:11}
and trees \cp{bard:18,garb:21,lueg:22}, where for the BHV metric \cp{bill:01} the requisite entropy conditions for the asymptotic analysis of M-estimators were recently established \cp{mull:21:5}. 
Analogous to functional analysis and linear operator theory being the basis of functional data analysis, so is metric geometry \cp{bura:01}  the basis for metric statistics; an abbreviated 
introduction for statisticians can be found in Section~2 and Appendix~B in \ci{mull:21:5}.

We aim to find commonalities across metric spaces, unifying theory and methodology, regardless of the specific geometry of the metric space, where entropy conditions that quantify the size of the space have emerged as a key tool.
The utility of entropy conditions and empirical process theory for random objects was 
recognized in recent work on \F means/barycenters and \F regression
\cp{mull:19:3, scho:19, ahid:20, scho:22}. 
A basic and classical notion is the measure of location provided by the \F mean 
\cp{frec:48} or barycenter, which is defined as  minimizer of the \F function $\Om \ra \real$, \be \la{Vf} V(\om)=\expect d^2(\om,X).\ee The population/sample minimizers \be \label{fm} \mp= \ao \expect d^2(\om,X), \quad 
\hmp= \ao \sn d^2(\om,X_i) \ee may not be unique and may correspond to a larger set. 
 
Uniqueness of \F means is guaranteed in Hadamard spaces \citep{stur:03}, and for positively curved spaces depends on both the geometry of the space $\Om$ and the probability measure $P$. 
Analogously to \F means one can also consider \F integrals for $\Om$-valued  functions $X(t)$ \cp{mull:16:2}. A functional scenario with $\Om$-valued stochastic processes widens the scope of functional data analysis \cp{mull:20:10}, where the previous standard has been that the underlying processes are Euclidean, either scalar-, vector- or 
$L^2$-valued \cp{mull:17:4}. The special case of distribution-valued stochastic processes is of particular interest and permits a more in-depth investigation
\cp{mull:23:3}. For general types of $\Om$-valued functions the \F integral provides a direct extension of the Riemann integral for $\real$-valued functions and in analogy to the \F mean is defined as  
$$\int_{\oplus} X(t)\, dt = \ao \int d^2(\om,X(t))\, dt. $$ 
This integral has proved useful in various investigations of object-valued processes \cp{mull:16:2,mull:21:5,mull:20:10}. 

Another important extension of \F means is the notion of a conditional \F mean $\expect_{\oplus}(X|Z)$, where $X \in \Om$ and  $Z \in \real^p$,  or more generally $Z\in \Om'$ for another metric space $(\Om',d')$. The statistical motivation is to model complex regression relationships that involve random objects.  A narrower specification is needed to make this notion useful for statistical modeling and data analysis,  targeting 
$$ m_\oplus(z) = \argmin_{\om \in \Om} \expect(d^2(X, \om) | Z = z).$$
Nadaraya--Watson kernel estimators for the case of manifold-to-manifold regression $\Om' \ra \Om$ that have been previously considered \cp{stei:09,stei:10} are subject to a severe version of the curse of dimensionality,  unless the predictor manifold is low-dimensional,  and they are also subject to substantial boundary effects. Special cases include  the smoothing of covariance matrices or data on Riemannian manifold indexed by time using local linear estimators  \cp{yuan:12, chen:13,corn:17},  including versions for functional and longitudinal data analysis, where functional principal components are a primary target after the smoothing step \cp{mull:20:5}. 
 
In addition to local linear and other desirable smoothers for the case of low-dimensional predictors it is also of interest to include global models that extend the classical linear multiple regression model when responses are random objects and predictors are Euclidean vectors. 
A general approach is \F regression \cp{mull:19:3} for the subproblem where $\Om'=\real^p$. Observing that smoothing or global linear regression methods are weighted averages with known or computable weights, one can take the weights that correspond to the respective regression method and form a weighted \F mean. The key problem is that when estimating at certain predictor levels some weights will be negative. Making use of entropy conditions for the metric space $\Om$ leads to asymptotic convergence across all spaces that satisfy these conditions, for both local and global regression models. 

Denoting the weights for a classical regression method with Euclidean predictors and scalar responses that are assigned to a predictor at level $Z\in \real^p$ when targeting  the estimate at a fixed predictor level 
$z \in \real^p$ by 
$w(Z,z)$, the \F regression estimator is 
 $$\hat{m}_\oplus(z) = \argmin_{\om \in \Om}\expect(d^2(X, \om)w(Z,z)).$$ 
 There are many open problems associated with this class of estimators and object regression is a subarea in rapid development. 
 Recent work includes 
 a novel perspective with extensions to other smoothing methods \cp{scho:22}, dimension reduction \cp{zhan:21:1,zhan:22:1,dong:22,virt:22} and consistent predictor selection \cp{tuck:23}.  
 A recent derivation of uniform convergence over the domain of the Euclidean predictor for local linear estimators made it possible to obtain consistent time warping identification for object-valued 
 functional data through pairwise warping comparisons and also to obtain consistent estimates for the location of extrema of functionals such as a specified eigenvalue for symmetric positive definite matrices as random objects \cp{mull:22:8}. This 
 result also facilitated the development of 
 a single-index version for \F regression, which enhances the flexibility of the global version of the model and includes 
 inference for the the predictors \cp{ghos:21:1,mull:23:5}. However,  much further work is needed on inference for object regression. Another class of object regression models that is of potential interest but not 
 sufficiently explored is transport regression, which was primarily developed for distributional data, where one can introduce  a transport algebra \cp{mull:23:2}. 
 
 Another relevant issue is the modeling of noise contamination in metric statistics. The  additive 
 noise model commonly employed  in Euclidean settings  is no longer feasible,  as there is no addition operation in metric spaces. However, noise can be modeled by random perturbation maps $\mathcal{P}: \Om \rightarrow \Om$ that satisfy \cp{mull:22:8} \bea
\om = \text{argmin}_{x \in \Om} \expect[d^2(\mathcal{P}(\om),x)] \ \text{for all } \om \in \Om.
\eea
 This is the equivalent of the postulate that an additive error $e$ in a Euclidean setting satisfies $\expect e=0$, while $\expect (e^2)=\sigma^2$
 corresponds to $\expect[d^2(\mathcal{P}(\om),\om)]=\sigma^2$ in the general case.

 In addition to location estimation another important thread in statistics is the estimation of spread, which is essential for uncertainty quantification. Plugging the \F mean into the \F variance function \eqref{Vf} gives the \F variance 
\be \la{Fvar} V_F=\E d^2(X,\mp) \quad \VF=\frac{1}{n} \sum_{i=1}^n d^2(X_i, \hmp),\ee
for which under suitable entropy conditions a Central Limit Theorem holds, $${n}^{1/2}(\VF-V_F) \wc N(0,\sigma_F^2).$$
This can be used to obtain an ANOVA-like test to compare populations of random objects as well as inference for 
change-points in a sequence of random objects \cp{mull:19:5,mull:20:2}. 

 It is easy to see that when $\Om=\real$ with the Euclidean metric,  the empirical \F variance \eqref{Fvar} equals the classical sample variance $\sigma^2$. For $X_i \in \real$, it is well known that 
 $\hat{\sigma}^2=\frac{1}{n-1}\snn (X_i-\bar{X})^2{=}\frac{1}{2n(n-1)} \sum_{i,j=1}^n (X_i-X_j)^2$
and this also holds in Hilbert spaces, as 
 $\expect \langle X- \expect X, X- \expect X \rangle
{=} \frac{1}{2} 
\E \langle X-X', X-X' \rangle$
where $X'$ is an independent copy of $X$.
However, the analogous equality does not hold in general metric spaces, and a second option for quantifying spread is then 
{metric variance} \cp{mull:20:10},
\be \la{Mvar} \Var_{\O}(X)=\frac{1}{2}\expect d^2(X,X'),\, \text{where } X'  \text{ is an independent copy of } \, X\in\Om. \ee
This notion can also be extended to {metric covariance} and {metric correlation}, 
	\begin{eqnarray*}
	\C(X,Y)&=&\frac{1}{4} \expect\left(d^2\left(X,Y'\right)+d^2\left(X',Y\right)-2d^2\left(X,Y\right)\right)\\
	\RR(X,Y)&=&\frac{\C(X,Y)}{\sqrt{\C(X,X)\C(Y,Y)}}.
	\end{eqnarray*}

	As already noted, the sample version of $\Var_{\O}(X)$ in the special case $\O=\mathbb{R}$ becomes 
	$\hat\sigma^2=\frac{1}{2n(n-1)} \sum_{i,j=1}^n (X_i-X_j)^2$. An advantage of metric variance/covariance is that these measures do not rely on the potentially arduous task of obtaining the \F mean in a first  step.   If $(\O,d)$ is such that $K(x,y)=d^2(x,y)$ is a kernel of negative type  \cp{kleb:05},  i.e.,  for all $n\ne 1$, $x_1,\dots,x_n\in\Om$, and $a_1,\dots,a_n\in\real$ with $\sum_{i=1}^{n} a_i = 0$ one has  $\sum_{i=1}^{n}\sum_{j=1}^{n} a_i a_j K(x_i,x_j) \le 0$,  results of 
 \cite{scho:37,scho:38} imply that  
metric correlation has the desirable property that 
	$-1 \leq \RR(X,Y) \leq 1$. The main distinction between metric correlation and distance correlation, another measure of dependence between paired metric space data \citep{lyon:13,szek:17}, is that the latter is tailored to measure probabilistic independence rather than to quantify  the strength of ‘positive’ or ‘negative’ association, which is the target of metric correlation. 	
The notion of metric covariance is based on pairwise distances between the random objects in a sample for the empirical version and on expected pairwise distances according to the probability measure $P$ in the population version. This
motivated us to consider these distances as a basic characteristic of the distributional properties of random objects that are otherwise hard to assess. 
To quantify this notion, for any fixed  $\om \in \Om$ the distances to the random objects as determined by the underlying measure $P$  are then of interest. 
They are captured by the distribution of the distances between $\om$ and any random element taking values in $\Om$, which then leads  to distance profiles indexed by $\om$.  
In the following sections, we explore the properties of distance profiles and how optimal transports between the corresponding distributions can be utilized to obtain transport ranks, transport quantiles and inference to compare populations of random objects. 

\section{Distance Profiles, Transport Ranks and Transport Quantiles} \label{sec:method} 
To introduce and motivate these key notions, we assume that data and random objects of interest are situated in a totally bounded separable metric space $(\O,d)$. 
Consider a probability space $(S,\mathcal{S},\prob)$, 
where %$\mathcal{S}$ is the Borel sigma algebra on a domain $S$ 
$S$ is a sample space, $\mathcal{S}$ is a sigma algebra of subsets of $S$, and $\prob$ is a probability measure. 
A random object $\obj$ is an $\Om$-valued random variable, i.e., a measurable map $\obj\colon S\ra\O$ and $P$ is a Borel probability measure that governs the distribution  of $\obj$, $X \sim P$, i.e.,  $P(A) = \prob (\{s\in S: \obj(s) \in A\}) \eqqcolon \prob(\obj\in A) = \prob(\obj\inv(A)) \eqqcolon \prob \obj\inv (A)$, for any Borel measurable $A\subseteq\O$. 
For any $\o \in \O$, let $\fo$ denote the cumulative distribution function (cdf) of the distribution of the distance between $\o$ and a random element $\obj$ that is distributed according to $P$. In our notation, we suppress the dependence of $F_\o$ on $P$ and $d$. 

Formally, for any $t \ge 0$, we define the \emph{distance profile} at $\o$ as 
\begin{equation}
 \fo(t) = \prob \left( d(\o,\obj) \leq t \right), \label{dp} 
\end{equation}
so that $\fo$ is a one-dimensional distribution that captures the probability mass enclosed by a metric ball in $\O$ that has center $\o$ and radius $t$, for all $t \ge 0.$ 
Thus the distance profile at $\o$ is the distribution of the distances that need to be covered to reach other elements of $\O$ when starting out at $\om$, as dictated by the distribution $P$ of the random objects $X$. When $t \rightarrow 0$, the distance profile at $\o$, $F_\omega(t)$, has the form of a small ball probability around $\o$ \citep{dabo:02,vakh:12}. An element $\o$ that is centrally located, i.e., close to most other elements, will have a distance profile with more mass near 0, in contrast to a distantly located or outlying element whose distance profile will assign mass farther away from 0. 
If distance profiles have densities, for a centrally located $\o$ the density will have a mode near 0, while the density near 0 will be small for a distantly located $\o$. 
Thus $\{\fo:\om\in\Om\}$ is a family of one-dimensional distributions indexed by $\Om$ that inform about the location of $\om$ relative to $X \sim P$.

The collection of \deps $\{\fo: \o \in \O\}$ represents the one-dimensional marginals of the stochastic process $\{d(\o,\obj)\}_{\o \in \O}$, which is well-defined in the sense of the Kolmogorov existence theorem (see Proposition \ref{prop1} for details). These simple marginals uniquely characterize the underlying measure $P$, 
if $(\O,d)$ is a metric space such that the kernel $K\colon \O \times \O \rightarrow \mathbb{R}$ given by $K(\o,\o')=d^{\theta}(\o,\o')$ for a $\theta>0$ is of strong negative type \cp{kleb:05,lyon:13}. This means that for all Borel probability measures $P$ on $\O$  and all  measurable functions $h\colon \O \rightarrow\mathbb{R}$ it holds  that    $\int_\O\int_\O K(\o,\o') h(\o)h(\o')\diffop P(\o)\diffop P(\o')\le 0$ with equality if and only if $h=0$ $P$-a.e.
Equivalently, for all Borel probability measures  $P_1, P_2$ on $\O$ one has
\begin{align}  \label{eq:snt} &\int_{\O}\int_{\O} K(\o,\o') \diffop P_1(\o)\diffop P_1(\o')\\ \nonumber 
& \hspace{1cm} + \int_{\O}\int_{\O} K(\o,\o') \diffop P_2(\o)\diffop P_2(\o') -2\int_{\O}\int_{\O} K(\o,\o') \diffop P_1(\o)\diffop P_2(\o') \le 0,  \end{align}  where equality holds if and only if $P_1 = P_2$; for   
further discussion  see Section~\ref{sec:disc}.  This  characterization of the underlying measures motivates  the use of distance profiles to obtain information about the complex distribution of the random objects $\obj$. Empirical estimates of the distance profiles that will be used for statistical inference are introduced below. A basic result concerning  distance profiles is as follows.

\begin{Proposition}
\label{prop1}
The stochastic process $\{d(\o,\obj)\}_{\o \in \O}$, 
for which the distance profiles  $\{\fo: \o \in \O\}$ as defined in \eqref{dp} are the one-dimensional marginals, is well-defined. Suppose that for some $\theta>0$, $(\O,d^\theta)$ is of strong negative type \eqref{eq:snt} and that $P_1, P_2$ are two probability measures on this space.  Then $P_1 = P_2$ if and only if $F_{\o}^{P_1}(u) = F_{\o}^{P_2}(u)$ for all $\o\in\O$ and $u \ge 0$, where $F_{\o}^{P_1}$ and $F_{\o}^{P_2}$ are the \deps of $\o$ with respect to $P_1$ and $P_2$.
\end{Proposition}

Consider the \dep $F_X$ of a random object $X \in \O$, $F_X(u) = \prob_{X'}\left( d(X,X') \leq u\right)$ $=\int_{S} \ind[d(X,X'(s)) \leq u] \diffop \prob(s)$,  where $X'$ is an independent copy of $X$. For each $\o$, the push-forward map of $F_\o$ to $F_X$, given by $F_X^{-1}(F_\o(\cdot))$, determines the optimal transport from the \dep $F_\o$ to the \dep $F_X$. {Here and throughout $F\inv$ denotes the quantile function corresponding to a cdf $F$, $F\inv(u) = \inf\{x\in\real: F(x)\ge u\}$, for $u\in(0,1)$.} We utilize the optimal mass transport map 
\be \label{TM} H_{X,\o} (u) =F_X^{-1}(F_\o(u))-u,\quad u \ge 0,\ee 
see, e.g., \cite{ambr:08}, to assign a measure of centrality to an element $\o \in \O$ with respect to $P$. 
When $F_\o$ is continuous, by a change of variable, the integral 
{\be \label{intH} \int H_{X,\o} (u) \diffop F_\o(u)=\int_0^1 \left\lbrace\fobj\inv(u)-\fo\inv(u)\right\rbrace\diffop u \ee {provides a summary measure of the mass transfer when transporting $F_\o$ to $F_X$.}

The utility of this notion is that if $\o$ is more centrally located than $X$ with regard to the measure $P$, we expect %the sign to be positive as 
the mass transfer to be predominantly from left to right and the magnitude of the integral in \eqref{intH} to reflect the outlyingness differential between $\om$ and a random object $X, \, X \sim P$. 
For example, for a distribution $P$ that is symmetric around a central point $\o_0$ and assigns less mass when moving away from $\o_0$, we expect that the integral \eqref{intH} with $\o=\o_0$ is relatively large and the magnitude of the integral \eqref{intH} is decreasing as the distance from $\o_0$ 
increases. 
This motivates to 
take the expected value of the integral in \eqref{intH}
to quantify the degree of centrality or outlyingness of an element $\o$. 

An illustration is in Figure~\ref{fig:ex_depthprfl} for the simple case where $X$ is a bivariate Gaussian random variable with mean zero and covariance $\diag(2,1)$. 
For the points $x\in\{(0,0),(2,0),(4,0),(6,0)\}$ and $\o = (2,2)$, their corresponding distance profiles are depicted as densities $f_{x}$ and $f_{\o}$ in the left panel, where the distances from $x$ to the rest of the data are seen to increase as $x$ moves away from the origin, which is exactly what one expects. 
For $x\in\{(0,0),(2,0),(4,0),(6,0)\}$, the transport maps $H_{x,\o}$ as per \eqref{TM} that move mass from $F_{\om}$ to $F_x$ for the fixed element $\o=(2,2)$ are in the right panel. 
For $\o=(2,2)$, mass moves to the left when transporting $\fo$ to $F_{x}$ for $x\in\{(0,0),(2,0)\}$, which are closer to the origin, and moves to the right for $x\in\{(4,0),(6,0)\}$, which are farther away from the origin.
Another example based on the U.S. electricity generation compositional data in Section~\ref{sec:energy} is shown in Figure~\ref{fig:energy2000_OT}. 
Transporting mass from the distance profile of $\o = \text{New Jersey (NJ)}$  to the profiles $x$ of Maryland (MD), Massachusetts (MA), Louisiana (LA) to Rhode Island (RI), one moves from the profile of a point in the middle of the ternary plot toward the profiles of points closer to the  boundary of the compositional space.  The mass transport  moves mass mostly to the left when transporting $\fo$ to $F_{x}$ for $x=\{\text{MD}, \text{MA}\}$ and unambiguously to the right for $x = \text{RI}$.

\begin{figure}[!htb]
 \centering
 \includegraphics[width=.75\linewidth]{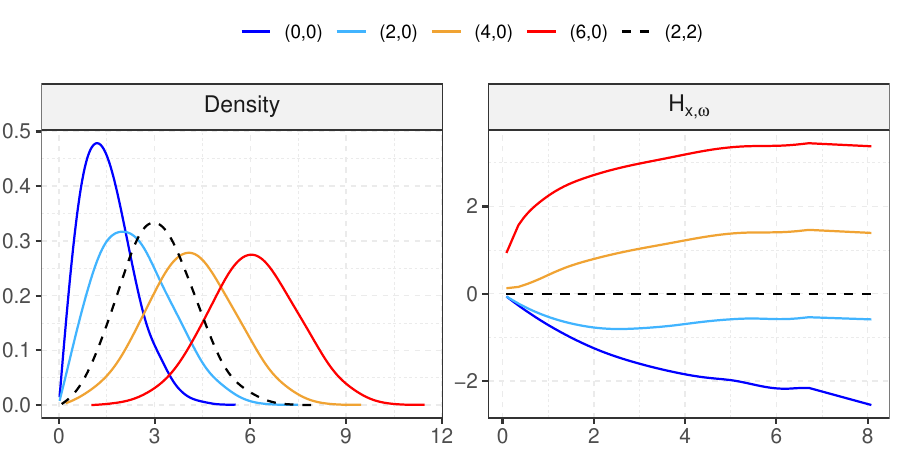}
 \caption{Left: Distance profiles, represented by the corresponding densities, at five points as indicated, with respect to a bivariate Gaussian distribution with mean zero and covariance $\diag(2,1)$. 
 Right: Transport maps subtracted by identity $H_{x,\o}$ as per \eqref{TM} for $x\in\{(0,0),(2,0),(4,0),(6,0)\}$ and $\o = (2,2)$, where negative (positive) values indicate transport to the right (left).}
 \label{fig:ex_depthprfl}
\end{figure}

This motivates the notion of \emph{transport ranks} to measure centrality 
of an element $\o\in \O$ with respect to $P$ as the $\expit$ of the expected integrated mass
transfer when transporting $F_\om$ to $F_X$, where $P = \prob\obj\inv$. We use the expit function $\expit(x)=e^x/(1+e^x)$ in the  definition of these ranks to ensure that the proposed ranks are scaled to lie in $(0,1)$; any strictly monotone invertible  function from $\real$ to $[0,1]$ can be used for this purpose.  Formally, 
\bal\label{eq:rank}
\rank_{\o} = \expit\left[ \expect\left\{\int_0^1[\fobj\inv(u)-\fo\inv(u)]\diffop u \right\}\right]. \eal 
The transport rank of $\o$ quantifies the aggregated preference of $\o$ with respect to the data cloud. The greater the transport rank of $\o$ is, the more centered $\o$ is relative to the sample elements. Equipped with an ordering of the elements of $\O$ by means of their transport ranks,  we define the \emph{transport median} set $\M$ of $P$ as the collection of points in the support $\po\subset\O$ of $P$ which have maximal transport rank and are therefore most central, 
\begin{equation} \label{R}
 \M = \argmax_{\o \in \po} \rank_\o.
\end{equation}

\begin{figure}[!hbt]
 \centering
 \includegraphics[width=.47\linewidth]{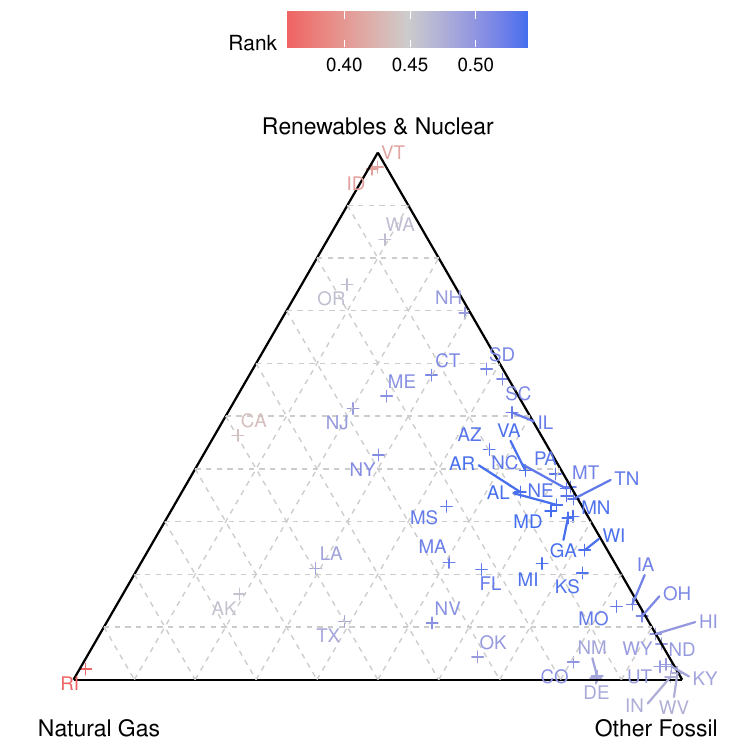}
 \includegraphics[width=.47\linewidth]{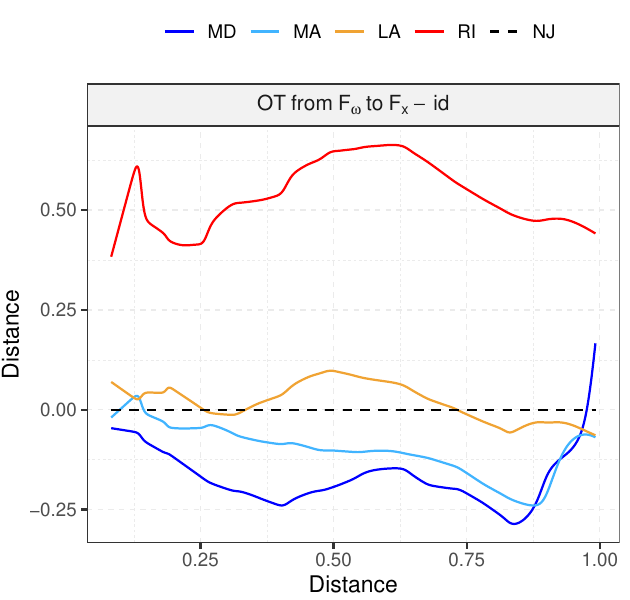}
 \caption{Left: Ternary plot of compositions of electricity generation in the year 2000 for the 50 states in the U.S., where the points are colored according to their transport ranks. Right: Mass transport maps 
 $H_{x,\o}$ as per \eqref{TM} for the transports from $\o = \text{NJ}$ to  $x\in\{\text{MD},\text{MA},\text{LA},\text{RI}\}$.}
 \label{fig:energy2000_OT}
\end{figure}

The \deps of the data objects together with the transport ranks and the transport median set are the key ingredients of the proposed toolkit to quantify centrality. 
These devices lend themselves to devise distance profile based methods for cluster analysis, classification and outlier detection, all of which are challenging when one deals with random objects. The set of maximizers of $\rank_\o$ in $\po$ constitutes the transport median set defined in (\ref{R}). Observing that the function $\rank_\o$ is uniformly continuous in $\o$ by Lemma~\ref{lma:LUC} in the Supplement, the transport median set is guaranteed to be non-empty whenever $\po$ is compact.
If $\O$ is a length space that is complete and locally compact, the Hopf--Rinow theorem \cp{chav:06} implies that  if  $\po$ is any bounded closed subset of $\O$, it is guaranteed to be compact. 

Once a center-outward ordering of the elements of $\O$ has been established through their transport ranks, these ranks can be utilized in numerous ways. One application 
is to define level sets of the form $L_\alpha = \{\o \in \O: \rank_\o = \alpha\}$ and nested superlevel sets $L^{+}_\alpha = \{\o \in \O: \rank_\o \geq \alpha\}$. 
By definition, $L^{+}_{\alpha_1} \subseteq L^{+}_{\alpha_2}$ whenever $\alpha_1 \geq \alpha_2$. Due to the continuity of $\rank_\o$ (see Lemma~\ref{lma:LUC} in the Supplement), the sets $L_\alpha$ and $L^{+}_\alpha$ are closed. Moreover when $(\O,d)$ is a bounded, complete and locally compact length space, again  by the Hopf--Rinow theorem $L_\alpha$ and $L^{+}_\alpha$ are compact as well. Superlevel sets $L^{+}_\alpha$ can be used to define \emph{transport quantile sets}. These can be viewed as a generalization of univariate quantiles to general random objects. Specifically, a $\zeta$-level transport quantile set can be defined as a level set $L_{\alpha}$ where $\alpha$ is such that $P(X \in L^{+}_{\alpha}) = \zeta$, for $\zeta\in(0,1)$.
Complements of superlevel sets can be used to identify potential outliers by highlighting observations with low transport ranks. Data trimming can be achieved by excluding points that have transport ranks lower than a suitably chosen threshold $\alpha_0$; one then might consider maximizers of transport ranks over trimmed versions of $\po$ to obtain trimmed analogues of the transport median set $\M$ and also trimmed \F means. 

\section{Properties of Distance Profiles and Transport Ranks} \label{sec:prop}

We discuss here some desirable properties of \deps, transport ranks and the transport median set that are
appropriately modified versions of analogous properties of classical ranks.

{\it Lipschitz Continuity of Transport Ranks.} By Lemma~\ref{lma:UC} in Section~\ref{sec:proof_est} in the Supplement, 
the distance profiles $F_\o(\cdot)$ and the associated quantile function representations $F_\o\inv(\cdot)$ are uniformly Lipschitz in $\o$ provided that the \deps have uniformly upper bounded densities with respect to the Lebesgue measure. 
This means that $F_{\o_1}$ and $F_{\o_2}$ are uniformly close to each other as long as $\o_1$ and $\o_2$ are close, and the distance between $F_{\o_1}$ and $F_{\o_2}$ is upper bounded by a constant factor of $d(\o_1,\o_2)$. 
Moreover 
transport ranks $\rank_\o$ are uniformly Lipschitz in $\o$, see Lemma~\ref{lma:LUC} in the Supplement.

{\it Invariance of Transport Ranks.} Let $(\tilde{\O},\tilde{d})$ be a metric space. A map $h\colon \O \rightarrow \tilde{\O}$ is isometric if $d(\o_1,\o_2)=\tilde{d}(h(\o_1),h(\o_2))$ for all $\o_1,\o_2 \in \O$. Theorem~\ref{prop:properties}(a) establishes the invariance of distance profiles, and thereby of transport ranks, under isometric transformations. In Euclidean spaces this ensures that the distance profiles are invariant under orthogonal transformations such as rotations. 
 
{\it Transport Modes and Center-Outward Decay of Transport Ranks.} Consider a situation where the distribution $P$ of $X$ concentrates around a point $\o_\oplus \in \O$. Specifically, if there exists an element $\op\in\O$ such that 
\begin{equation}
 \label{center}
 F_{\o_\oplus}(u) \geq F_{\o}(u)
\end{equation}
for any $\o \in \O$ and any $u \ge 0$,  we refer to $\op$ as an $\O$-valued \emph{transport mode} of $P$. 
Condition \eqref{center} states that a $d$-ball of radius $u$ around $\o_\oplus$ contains more mass under $P$ than a similar ball around any other point in $\O$.
According to Theorem~\ref{prop:properties}(b), if $P$ has a transport mode, then the transport rank of the transport mode cannot be smaller than that of any other $\omega \in \Omega$ and therefore, a transport mode is always contained in the transport median set. For distributions that concentrate around their unique \F mean \citep{frec:48}, the \F mean is the transport mode and hence is contained in the transport median set \citep{luna:20}. 
Theorem~\ref{prop:properties}(c) provides a characterization of the radial 
ordering induced by the transport rank for the special case where the data distribution on $\O$ has a transport mode $\o_\oplus$ by considering curves of the form $\gamma\colon [0,1] \rightarrow \O$ that originate from $\o_\oplus$, i.e. $\gamma(0)=\o_\oplus$. According to Theorem~\ref{prop:properties}(c), transport ranks are non-increasing along curves originating from a transport mode $\o_\oplus$, whenever $P$ is such that the \deps decay systematically along the curve $\gamma(t)$ as $t$ is increasing. 

{\it Characterization of the Probability Measure $P$ Through Transport Ranks.} 
Theorem~\ref{prop:properties}(d) shows that when $(\O,d)$ is of strong negative type, the comprehensive set of all transport ranks $\{\rank_\o\}_{\o \in \O}$ uniquely characterizes the underlying measure. 
\begin{theorem}
\label{prop:properties}
For a separable metric space $(\O,d)$ the distance profiles $F_\o$ and the transport ranks $\rank_\o$ satisfy the following properties: 
\begin{enumerate}[label=(\alph*)]
 \item\label{prop1a} Let $h\colon\O \rightarrow \tilde{\O}$ be a bijective isometric measurable map between $(\O,d)$ and $(\tilde{\O},\tilde{d})$ and $P_h(\cdot)=P(h^{-1}(\cdot))$ the push-forward measure on $\tilde{\O}$. 
 Then $F_{h(\o)}^{P_h}(u)=F_\o^P(u)$ for all $u \in \real$, hence $\rank_{h(\o)}^{P_h}=\rank_\o^P$, where $F_\o^P(u)=\prob(d(\o,X)\leq u)$ and $X$ is a $\O$-valued random element such that $P = \prob X\inv$,  $F_{h(\o)}^{P_h}(u)=\prob(\tilde{d}(h(\o),h(X))\leq u)$, 
 $\rank_\o^P$ is the transport rank of $\o$ with respect to $P$ and 
 $\rank_{h(\o)}^{P_h}$ is the transport rank of $h(\o)$ with respect to $P_h$.
 \item\label{prop1b} If $\o_\oplus$ is a transport mode of $P$ as per \eqref{center}, $\rank_{\op} \geq 1/2$. Moreover $\rank_{\op} \geq \rank_\o$ for any $\o \in \O$ and $\o_\oplus \in \M$.
 \item \label{curves} 
 Suppose $\op$ is a transport mode of $P$. Let $\gamma\colon [0,1]\rightarrow\O$ be curve in $(\O,d)$ such that $\gamma(0)=\op$ and $F_{\gamma(s)}(u) \geq F_{\gamma(t)}(u)$ for all $u \in \mathbb{R}$ and $0 \leq s < t \leq 1$. Then $\rank_{\gamma(s)}(u) \geq \rank_{\gamma(t)}(u)$ whenever $0 \leq s < t \leq 1$.
 \item \label{law_ranks} Suppose the metric space $(\O,d)$ is of strong negative type \eqref{eq:snt} and $P_1, P_2$ are two probability measures on the space. Then $P_1 = P_2$ if and only if $\rank_{\o}^{P_1} = \rank_{\o}^{P_2}$ for all $\o\in\O$, where $\rank_{\o}^{P_1}$ and $\rank_{\o}^{P_2}$ are the transport ranks of $\o$ with respect to $P_1$ and $P_2$.
\end{enumerate}
\end{theorem}

\section{Estimation and Large Sample Properties} \label{sec:est}
While so far we have introduced the notions of the distance profiles, transport ranks and transport median sets at the population level, in practice one needs to estimate these quantities from a data sample of random objects $\{\obj_{\subidx}\}_{\subidx=1}^{n}$ consisting of $n$ independent realizations of $\obj$. 
For obtain the distance profiles $\fo$, $\o \in \po$ from a sample, we use empirical estimates 
\begin{equation}
\label{eq:FhatO}
 \hfo(t) = \frac{1}{n} \sum_{\subidx=1}^{n} \ind[ d(\o, \obj_{\subidx}) \leq t ],\quad t \ge 0,
\end{equation} 
where $\ind[A]$ is the indicator function for an event $A$.

Replacing expectations with empirical means and using estimated distance profiles $\hfi$ given by $\hfi(t) = \frac{1}{n-1} \sum_{1\le j \le n,\, j\ne i} \ind[ d(\obj_{j}, \obj_{\subidx}) \leq t ]$ for $t\ge 0$ as surrogates of $\fobji$, we obtain estimates for the transport rank of $\o\in\O$ defined in \eqref{eq:rank} as
\bal\label{eq:hrank}
\hrank_{\o} = \expit\left[ \frac{1}{n}\sum_{i=1}^{n} \left\{\int_0^1[\hfi\inv(u)-\hfo\inv(u)]\diffop u\right\}\right]. \eal
The term 
$\int_0^1[\hfi\inv(u)-\hfo\inv(u)]\diffop u$ provides a comparison between the outlyingness of $\o$ and that of $\obj_{\subidx}$; mass movement predominantly to the right (left) indicates that $\o$ is more central (outlying) compared to $\obj_{\subidx}$, respectively. Finally we define the estimated transport median set 
\begin{equation}\label{eq:tmedian-est}
 \hM = \argmax_{\o \in \{X_1,X_2, \dots, X_n\}} \hrank_\o.
\end{equation}

To obtain asymptotic properties of these estimators and convergence towards their population targets, we require the following assumptions. 

\begin{assumption}
\label{ass:entropy} Let $N(\eps, \Omega, d)$ be the covering number of the space $\O$ with balls of radius $\eps$ and $\log N(\eps, \Omega, d)$ the corresponding metric entropy. Then 
 \begin{equation} \label{entropy} 
 \eps \log N(\eps, \Omega, d) \rightarrow 0\quad \text{as} \quad \eps \rightarrow 0.
 \end{equation}
\end{assumption}

\begin{assumption}
\label{ass:dpfctn} 
 For every $\o \in \O$,  $F_\o$ is absolutely continuous with continuous density $f_\o$. For   
 % \begin{equation*}
 $\underline{\Delta}_{\o} = \inf_{t \in \mathrm{support}(f_\o)} f_\o(t) \ \text{and} \ \overline{\Delta}_{\o} = \sup_{t \in \real} f_\o(t),$ 
% \end{equation*} 
  $\underline{\Delta}_{\o}>0$ for each $\o \in \O$ and 
 there exists 
 $\overline{\Delta} > 0$ such that
 $\sup_{\o \in \O} \overline{\Delta}_\o \leq \overline{\Delta}.$
 \end{assumption}

Assumptions~\ref{ass:entropy} and \ref{ass:dpfctn} are necessary for Theorem~\ref{thm:fhat},  which provides the uniform convergence of $\hfo$ to $\fo$. The entropy condition  Assumption~\ref{ass:entropy} serves to overcome the dependence between the summands in the estimator of the transport rank and to establish uniform convergence to the population transport ranks. 
Assumption~\ref{ass:dpfctn} is a smoothness condition of the probability measure $P$, which is widely satisfied. Examples include  absolutely continuous distributions with compact support in a Euclidean space or  probability distributions on a Riemannian manifold, where  after applying Riemannian log maps 
the transformed distributions on  tangent spaces are absolutely continuous with compact support.

For any $t \ge 0$, $\fo (t)= \expect(\hfo (t))$. For functions $y_{\o,t}\colon \O \ra \real$ with $y_{\o,t}(x) = \mathbb{I}\{d(\o,x) \leq t\}$ and the function class $\mcF=\{ y_{\o,t} : \o \in \O, t \in \real \}$,  the following result establishes that under Assumptions \ref{ass:entropy} and \ref{ass:dpfctn} the function class $\mcF$ is $P$-Donsker. 

\begin{theorem}
\label{thm:fhat}
Under Assumptions \ref{ass:entropy} and \ref{ass:dpfctn}, $\{ \sqrt{n}(\hfo (t)-\fo (t)): \o \in \O,\, t \in \real \}$ converges weakly to a zero-mean Gaussian process $\mathbb{G}_P$ with covariance given by 
\begin{equation*}
 \mathcal{C}_{(\o_1,t_1),(\o_2,t_2)} = \Cov(y_{\o_1,t_1}(\obj), y_{\o_2,t_2}(\obj))
\end{equation*}
for $\o_1,\o_2\in\O$ and $t_1,t_2\in\real$.
\end{theorem}

Assumption~\ref{ass:entropy} is a restriction on the complexity of the metric space $(\O,d)$. It is satisfied for a broad class of spaces. 
In particular, any space $(\O,d)$ such that $\log N(\eps, \O, d) = O\left(\frac{1}{\eps^\alpha}\right)$ for some $\alpha < 1$ satisfies Assumption \ref{ass:entropy}. This is true for any $(\O,d)$ which can be represented as a subset of elements in a finite dimensional Euclidean space, for example the space of graph Laplacians or network adjacency matrices with fixed number of nodes \citep{kola:20,gine:17}, symmetric positive definite matrices of a fixed size \citep{dryd:09}, simplex valued objects in a fixed dimension \citep{jeon:20} and the space of phylogenetic trees with the same number of tips \citep{kim:20,bill:01}. 
It holds that $\log N(\eps, \O, d) = O\left(\eps^{-\alpha}\right)$ for any $\alpha < 1$ when $\O$ is a VC-class of sets or a VC-class of functions \citep[Theorems~2.6.4 and 2.6.7,][]{well:96}. 
Assumption \ref{ass:entropy} also holds for $p$-dimensional smooth function classes $C_1^\alpha(\mathcal{X})$ \citep[page 155,][]{well:96} on bounded convex sets $\mathcal{X}$ in $\real^p$ equipped with the $\|\cdot\|_\infty$-norm \citep[Theorem~2.7.1,][]{well:96} or the $\|\cdot\|_{r,Q}$-norm 
for any probability measure $Q$ on $\real^p$ \citep[Corollary~2.7.2,][]{well:96}, if $\alpha \geq p+1$. 

Of particular interest for many applications is the case when $\O$ is the space of one-dimensional distributions on some compact interval $I \subset \real$ with the 2-Wasserstein metric $d=d_W$ defined in \eqref{eq:dwass} \citep{mull:19:3}. If $\O$ is represented using the quantile function of the distributions then, without any further assumptions, $\log N(\eps, \O, d_W)$ is upper and lower bounded by a factor of $1/\eps$ \citep[Proposition~2.1,][]{blei:07} and does not meet the criterion in Assumption~\ref{ass:entropy}. However, if we assume that the distributions in $\O$ are absolutely continuous with respect to the Lebesgue measure on $I$ with smooth densities uniformly taking values in some interval $[l_0,u_0]$, $0 <l_0<u_0<\infty$, then $\O$ equipped with $d_W$ satisfies Assumption~\ref{ass:entropy}. To see this, observe that with the above characterization of $\O$ the quantile functions corresponding to the distributions in $\O$ have smooth derivatives that are uniformly bounded. 
With  $\mathcal{Q}_{deriv}$ denoting the space of the uniformly bounded derivatives of the quantile functions in $\O$,  $\log N(\eps, \mathcal{Q}_{deriv}, \|\cdot\|_1) = O\left(\eps\inv\right)$, where $\|\cdot\|_1$ is the $L_1$ norm under the Lebesgue measure on $I$ \citep[Corollary~2.7.2,][]{well:96}. 
Using Lemma~1 in \cite{gao:09}, with $\mathcal{F} \equiv \mathcal{Q}_{deriv}$, $\mathcal{G} \equiv \Omega$, $\alpha(x)=x$ and $\phi(\eps)=K/\eps$ for some constant $K$, $\log N(\eps, \O, d_W) = O\left(\eps^{-1/2}\right)$ which meets the requirement of Assumption~\ref{ass:entropy}. 
If $\O$ is the space of $p$-dimensional distributions on a compact convex set $I \subset \real^p$, represented using their distribution functions endowed with the {$L_{r}$} metric with respect to the Lebesgue measure on $I$, then Assumption~\ref{ass:entropy} is satisfied if $\O \subset C_1^\alpha(I)$ for $\alpha \geq p+1$. 

Next we discuss the asymptotic convergence of the estimates $\hrank_\o$ of transport ranks. 
Theorem~\ref{thm:Rhat} establishes a $\sqrt{n}$-rate of convergence uniformly in $\o$. 
\begin{theorem}
\label{thm:Rhat}
Under Assumptions \ref{ass:entropy} and \ref{ass:dpfctn},
\begin{equation*}
 \sqrt{n} \sup_{\o \in \O} |\hrank_\o - \rank_\o| = O_{\prob}(1).
\end{equation*}
\end{theorem}

To conclude this section we consider the convergence of the estimated transport median set $\hM$ to $\M$ in the Hausdorff metric 
\begin{equation}
 \rho_H\left(\hM,\M\right) = \max \left( \sup_{\o \in \hM} d(\o,\M), \sup_{\o \in \M} d(\o,\hM) \right), 
\end{equation}
where for any $\o \in \O$ and any subset $A \subset \O$, $d(\o,A) = \inf_{s \in A} d(\o,s)$. We derive uniform Lipschitz continuity of  transport ranks $\rank_\o$ in $\o$ (Lemma~\ref{lma:LUC} in the Supplement) and require the following  additional assumption.
\begin{assumption}
 \label{ass:separation2} For some $\eta' > 0$, for any $0 < \eps < \eta'$,  $$\alpha(\eps)= \inf_{\tilde\o \in \M}\inf_{d(\o,\tilde\o)>\eps} \left| \rank_{\o}-\rank_{\tilde\o}\right| > 0.$$
\end{assumption}
Assumption~\ref{ass:separation2} deals with the identifiability of transport medians and stipulates  that the transport median set is a union of single point sets which are separated from each other by a minimum fixed distance.  In particular,  unimodal probability measures satisfy Assumption~\ref{ass:separation2}. This assumption is needed to derive our next result 
on the consistency of the estimated transport median set for the true transport median set in the Hausdorff metric. 
\begin{theorem}
\label{thm:Mhat}
Assume that the distribution $P$ is such that $\M$ is non-empty. Under Assumptions \ref{ass:entropy}--\ref{ass:separation2},
\begin{equation*}
 \rho_H\left(\hM,\M\right) = o_{\prob}(1).
\end{equation*}
\end{theorem}

\section{Two-sample Inference With Distance Profiles} \label{sec:test}
\subsection{Construction of a Two-Sample Test} Assume that  $X_1, X_2, \dots, X_n$ is a sample of random objects taking values in $\O$, generated according to a Borel probability measure $P_1$ on $\O$, and that  $Y_1, Y_2, \dots, Y_m$ is another sample of $\O$-valued random objects generated analogously according to a Borel probability measure $P_2$. 
Two-sample testing in this setting concerns the  null \eqref{eq:null} and alternative \eqref{eq:alternative} hypotheses
\begin{align}
 & H_0: P_1 = P_2 \label{eq:null} \\ & H_1: P_1 \neq P_2 \label{eq:alternative}.
\end{align}

Nonparametric two-sample tests have been studied extensively in many settings. To extend this classical problem to object data poses new challenges.  Existing methods that are based on distances, such as the graph based tests \citep{chen:17:3} and the energy test \citep{szek:04}, either require tuning parameters for their practical implementation or lack theoretical guarantees on the power of the test, particularly when using permutation cut-offs for type I error control. We propose here a two-sample test based on the distance profiles of the observations. The proposed test is tuning parameter free, has rigorous asymptotic type I error control under $H_0$ \eqref{eq:null} and is guaranteed to be powerful against contiguous alternatives for sufficiently large sample sizes. 
While it was presented at the Rietz Lecture and derived independently, our test statistic is similar in spirit to a test proposed in \cite{wang:23:3}. We note that our results are derived under weaker assumptions and provide, in addition to consistency under the null hypothesis, power guarantees under contiguous alternatives, as well as theoretical guarantees for the corresponding permutation tests.

We require some notations. For $\omega \in \Omega$, the distance profile of $\o$ with respect to $X \sim P_1$ and $Y \sim P_2$ respectively are given by 
$F_{\o}^X(\cdot)$ and $F_{\o}^Y(\cdot)$,  where for $u \in \mathbb{R}$,
\begin{equation}
 \label{eq:depxy}
 F_{\o}^X(u) = \prob (d(x,X) \leq u) \quad \text{and} \quad F_{\o}^Y(u) = \prob (d(x,Y) \leq u).
\end{equation}
Let $\hf[XX1\cdot], \hf[XX2\cdot], \dots, \hf[XXn\cdot]$ be the estimated \emph{in-sample} distance profiles of $X_1, \dots, X_n$,
respectively, with respect to the observations from $P_1$, i.e.,
\begin{equation*}
\hf[XXiu]= \denom[n-1] \addneq[ji] \ind[d(X_i,X_j) \leq u].
\end{equation*}
Then we obtain the \emph{out-of-sample} distance profiles of $X_1, \dots, X_n$, respectively,  with respect to the observations from $P_2$, given by $\hf[YX1\cdot], \hf[YX2\cdot], \dots, \hf[YXn\cdot]$, where 
\begin{equation*}
\hf[YXiu]= \denom[m] \add[jm] \ind[d(X_i,Y_j) \leq u].
\end{equation*}

Similarly we estimate the in-sample and the out-of-sample distance profiles of \newline  $Y_1, \dots, Y_m$ with respect to the observations from $P_2$ and $P_1$, respectively, given by $\hf[YY1\cdot], \hf[YY2\cdot], \dots, \hf[YYm\cdot]$ and $\hf[XY1\cdot], \hf[XY2\cdot], \dots, \hf[XYm\cdot]$ respectively, where for $u \ge 0$ 
\begin{equation*}
\hf[YYju]= \denom[m-1] \addneq[ji] \ind[d(Y_j,Y_i) \leq u]
\end{equation*}
and
\begin{equation*}
\hf[XYju]= \denom[n] \add[in] \ind[d(Y_j,X_i) \leq u].
\end{equation*}
With $T^X_{nm}$ and $T^Y_{nm}$ defined as
\begin{equation*}
T^X_{nm}(X,Y)= \denom[n]\add[in] \int \left\lbrace \hf[XXiu] - \hf[YXiu] \right\rbrace^2 \df{u} 
\end{equation*}
and
\begin{equation*}
T^Y_{nm}(X,Y) = \denom[m]\add[im] \int \left\lbrace \hf[XYiu] - \hf[YYiu] \right\rbrace^2 \df{u},
\end{equation*}
the proposed test statistic $T_{nm}(X,Y)$ is 
\begin{equation}
\label{eq:teststat_unweighted}
 T_{nm}(X,Y) = \frac{nm}{n+m} \left \lbrace T^X_{nm} + T^Y_{nm} \right \rbrace.
\end{equation}

To enhance flexibility,  we also consider a generalized weighted version of the test statistic, where for each observation $X_i$ or $Y_i$, we allow for data adaptive weight profiles $\wght[Xi\cdot]$ and $\wght[Yi\cdot]$ that can be tuned appropriately to enhance the detection capacity of the test statistic. This leads to  
weighted versions of $T^X_{nm}$ and $T^Y_{nm}$, 
\begin{align*}
& T^{X,w}_{nm}(X,Y)= \denom[n]\add[in] \int \wght[Xiu] \left\lbrace \hf[XXiu] - \hf[YXiu] \right\rbrace^2 \df{u} \\ \text{and} \quad & T^{Y,w}_{nm}(X,Y) = \denom[m]\add[im] \int \wght[Yiu] \left\lbrace \hf[XYiu] - \hf[YYiu] \right\rbrace^2 \df{u}
\end{align*}
and the weighted test statistic 
\begin{equation}
\label{eq:teststat_weighted}
 T^w_{nm}(X,Y) = \frac{nm}{n+m} \left \lbrace T^{X,w}_{nm} + T^{Y,w}_{nm} \right \rbrace.
\end{equation}
Note that the test statistic in equation \eqref{eq:teststat_unweighted} is a version of the generalized test statistic in equation \eqref{eq:teststat_weighted} with $\wght[Xi\cdot]=\wght[Yi\cdot]\equiv1$. Hereafter  we suppress the dependence of $T^X_{nm}, T^Y_{nm}, T_{nm}, T^{X,w}_{nm}, T^{Y,w}_{nm}$ and $T^w_{nm}$ on $(X,Y)$ as this  will be clear from the context.

{Suppose that for each $x \in \Omega$, there exists a population limit of the estimated data adaptive weight profile $\hat{w}_x(\cdot)$ given by $w_x(\cdot)$ such that 
\begin{equation} \label{wt}
 \sup_{x \in \Omega} \sup_{u} \lvert \hat{w}_x(u)-w_x(u) \rvert = o_\prob(1).
\end{equation} 
The weight profile dependent quantities %$D^w_{XY}(P_1,P_2)$ as
\begin{align} \label{dw} 
 D^w_{XY}(P_1,P_2)&= \E\left\lbrace \int w_{X'}(u) \left(F_{X'}^X(u)-F_{X'}^Y(u)\right)^2\df{u}\right\rbrace \\  \nonumber & \quad+ \E\left\lbrace \int w_{Y'}(u) \left(F_{Y'}^X(u)-F_{Y'}^Y(u)\right)^2\df{u}\right\rbrace,
\end{align} 
%pd
where $F_\o^X(\cdot)$ and $F_\o^Y(\cdot)$ are as defined in equation \eqref{eq:depxy} and $X' \sim P_1$ and $Y' \sim P_2$ capture the population version of the proposed test statistic \eqref{eq:teststat_weighted}. First observe that under $H_0$, for any $w_x(\cdot), \, x \in \Omega$, it holds that $D^w_{XY}(P_1,P_2)=0$. Next we show that under mild conditions on the weight profiles, $(\Omega,d)$, $P_1$ and $P_2$, $D^w_{XY}(P_1,P_2)=0$ if and only if $P_1 = P_2$.

Let $\mathcal{P}_{(\Omega,d)}$ denote the class of all Borel probability measures on $(\O,d)$ that are uniquely determined by the measure of all open balls or equivalently by the set of distance profiles, that is, for $Q_1, Q_2 \in \mathcal{P}_{(\Omega,d)}$, $Q_1 = Q_2$ if and only if $F^{Q_1}_\omega(u) = F^{Q_2}_\omega(u)$ for all $\omega \in \Omega$ and $u \ge 0$. In fact $\mathcal{P}_{(\Omega,d)}$ corresponds to the set of all Borel probability measures on $(\O,d)$ when the metric $d$ is such that $d^\theta$ is of strong negative type \eqref{eq:snt}  for some $\theta > 0$ 
(see Proposition \ref{prop1}). If  $w_\o(u) > 0$ for any $\o \in \O$ and for any $u \ge 0$, then $D^w_{XY}=0$ implies that $F_\o^X(u)=F_\o^Y(u)$ for almost any $u \ge 0$ and for any $\o$ in the union of the supports of $P_1$ and $P_2$. Hence,  if $\O$ is contained in the union of the supports of $P_1$ and $P_2$, then $D^w_{XY}(P_1, P_2)=0$ implies that $P_1=P_2$ whenever $P_1,P_2 \in \mathcal{P}_{(\Omega,d)}$. In the following, we  will suppress $(P_1,P_2)$ in the notation $D^w_{XY}(P_1, P_2)$. % as it will be clear from context.

We will use $D^w_{XY}$ in Section~\ref{sec:test_theory} to evaluate the power performance of the test by constructing a sequence of contiguous alternatives that approach the null hypothesis \eqref{eq:null}. To obtain the asymptotic power of the test, we work with $q_\alpha$, the asymptotic critical value for rejecting $H_0$, where $q_\alpha = \inf\{t: \Gamma_L(t) \geq 1-\alpha \}$ and  $\Gamma_L(\cdot)$ is the cumulative distribution function of the asymptotic null distribution corresponding to the law of $L$ (see Theorem \ref{thm: null_dist}).  % in Section~\ref{sec:test_theory}. 
Since $L$ is an infinite mixture of chi-squares with mixing weights depending on the data distribution under $H_0$, we estimate $q_\alpha$ using a random permutation scheme in practice, as follows.

\vspace{0.1 in}

Let $\Pi_{nm}$ denote the collection of all $(n+m)!$ permutations of $\{1,2,\dots,n+m\}$ and  $\Pi$  a random variable that follows a uniform distribution on $\Pi_{nm}$ and is independent of the observations $X_1, \dots, X_n$ and $Y_1, \dots Y_m$. Let $V_1, V_2, \dots , V_{n+m}$ denote the pooled sample where $V_i = X_i $ if $i \leq n$ and $V_i = Y_{i-n}$ if $i \geq n+1$. Let $\Pi_1, \Pi_2, \dots \Pi_K$ denote i.i.d. replicates of $\Pi$. Each $\Pi_j=(\Pi_j(1), \dots, \Pi_j(n+m))$ is a random permutation of $\{1,2,\dots,n+m\}$ and when applied to the data yields $V_{\Pi_j} = \{V_{\Pi_j(1)}, V_{\Pi_j(2)}, \dots , V_{\Pi_j(n+m)}\}$, $\,j=1,\dots,K,$ which constitute a collection of randomly permuted pooled data. For each $j=1, \dots, K$, split the data $V_{\Pi_j}$ into $X_{\Pi_j}=\{V_{\Pi_j(1)}, V_{\Pi_j(2)}, \dots , V_{\Pi_j(n)}\}$ and $Y_{\Pi_j}=\{V_{\Pi_j(n+1)}, V_{\Pi_j(2)}, \dots , V_{\Pi_j(n+m)}\}$. With $X_{\Pi_j}$ and $Y_{\Pi_j}$ being the proxies for the two samples of sizes $n$ and $m$,  respectively, evaluate the test statistic replicates $T_{\Pi_j}=T^w_{nm}(X_{\Pi_j},Y_{\Pi_j})$. Define $\hat{\Gamma}_{nm}(\cdot)$ as
 \begin{equation}
 \label{eq:perm_cdf}
 \hat{\Gamma}_{nm}(t) = \frac{1}{K} \sum_{j=1}^K \ind [T_{\Pi_j} \leq t],
 \end{equation}
which approximates the randomization distribution of $T^w_{nm}$ using the random permutations $\Pi_1, \Pi_2, \dots \Pi_K$. Then a natural  estimate of $q_\alpha$ is  \be \label{qest} \hat{q}_\alpha = \inf\{t : \hat{\Gamma}_{nm}(t) \geq 1-\alpha \}.\ee

\subsection{Theoretical Guarantees for Type I Error Control and the Asymptotic Power of the Test}
\label{sec:test_theory}
To establish  theoretical guarantees of the proposed test, in particular  the limiting distribution of the test statistic under $H_0$ \eqref{eq:null} and the consistency of the test under the alternative \eqref{eq:alternative}, we require additional assumptions,  % beyond Assumption~\ref{ass:entropy} that we enlist below. 
including a modified version of Assumption~\ref{ass:dpfctn}, listed below as Assumption~\ref{ass:assumption_lipschitz}. Assumption~\ref{ass:assumption_weights} requires regularity conditions on the data adaptive weight profiles to ensure that they have a well behaved asymptotic limit. Assumption~\ref{ass:samp_ratio} is needed so that none of the group sizes is  asymptotically negligible. % as compared to the other. 
\begin{assumption}
 \label{ass:assumption_weights}
For each $x \in \Omega$, there exists a population limit of the estimated weight profiles %$\hat{w}_x(\cdot)$ given by $w_x(\cdot)$ 
such that \eqref{wt} is satisfied;
 there exists $C_w > 0$ such that $\sup_{x \in \Omega} \sup_{u} \lvert w_x(u) \rvert \leq C_w$; for some $L_w > 0$ it holds that 
%\begin{equation*}
 $\sup_{x \in \Omega} |w_x(u)-w_x(v)| \leq L_w |u-v|$.
% \end{equation*}
\end{assumption}\vspace{-.8cm}
\begin{assumption}
\label{ass:assumption_lipschitz}
For each $x \in \Omega$,  $X \sim P_1$ and $Y \sim P_2$,  $F_x^X(t)=\prob(d(x,X)\leq t)$ and $F_x^Y(t)=\prob(d(x,Y)\leq t)$ are absolutely continuous, with densities  $f_x^X(t)$ and $f_x^Y(t)$, respectively, that satisfy $\inf_{t \in \mathrm{supp}(f^X_x) }f^X_x(t) > 0$, $\inf_{t \in \mathrm{supp}(f^Y_x) }f^Y_x(t) > 0$.  There exist  $L_X, L_Y > 0$ such that
$\sup_{x \in \Omega} \sup_{t \in \mathbb{R}} |f_x^X(t)| \leq L_X$ and $\sup_{x \in \Omega} \sup_{t \in \mathbb{R}} |f_x^Y(t)| \leq L_Y$.
\end{assumption}\vspace{-.7cm}
\begin{assumption}
\label{ass:samp_ratio}
There exists $0 < c < 1$ such that  sample sizes $n$ and $m$ satisfy $n/(n+m) \rightarrow c$ as $n,m \rightarrow \infty$.
\end{assumption}

Theorem \ref{thm: null_dist} 
provides the framework for asymptotic type I error control of the test. The asymptotic distribution of the test statistic $T_{nm}$ \eqref{eq:teststat_unweighted}, which we will illustrate later in the simulations,  can be directly derived from Theorem \ref{thm: null_dist} by plugging in $\wt[Xi\cdot]=\wt[Yi\cdot]\equiv 1$ as this constant weight profile satisfies Assumption~\ref{ass:assumption_weights} trivially.

\begin{theorem}
\label{thm: null_dist}
 Under $H_0$ \eqref{eq:null} and Assumptions \ref{ass:entropy}, \ref{ass:assumption_weights}, \ref{ass:assumption_lipschitz} and \ref{ass:samp_ratio}, $T^w_{nm}$ converges in distribution to the law of a random variable $L=2\add[j\infty] Z_j^2 \E_{V}(\lambda_j^{V})$, where $Z_1, Z_2, \dots$ is a sequence of i.i.d. $N(0,1)$ random variables,  $V \sim P$ where $P=P_1=P_2$ under $H_0$ and for any $x \in \O$, $\lambda_1^x \geq \lambda_2^x \geq \dots$ are the eigenvalues of the covariance surface given by
\begin{equation*}
 C^x(u,v) = \sqrt{w_x(u)w_x(v)} \ \mathrm{Cov}(\ind[d(x,V')\leq u],\ind[d(x,V')\leq v])
\end{equation*}
with $ V' \sim P$.
\end{theorem}

To study the asymptotic power of the proposed test, we consider a
sequence of alternatives \be \label{hnm} H_{nm}=&&\{(P_1,P_2):\, X \sim P_1, Y \sim P_2,\\ \nonumber &&\text{with} \ D^w_{XY}=a_{nm},  \ a_{nm} \rightarrow 0, \ nm/(n+m) a_{nm} \rightarrow \infty, \ n,m \rightarrow \infty \},\ee  with $D^w_{XY}$ as in \eqref{dw}.  
The  $\{H_{nm}\}$ form a sequence of contiguous alternatives shrinking towards $H_0$. 
The  power of the test under this sequence is  
\begin{equation} \label{pow}
 \beta^w_{nm} = \prob_{H_{nm}} \left( T^w_{nm} > q_\alpha \right), 
\end{equation}
where $q_\alpha = \inf\{t: \Gamma_L(t) \geq 1-\alpha \}$ and $\Gamma_L(\cdot)$ is the cumulative distribution function of the asymptotic null distribution corresponding to the law of $L$ in Theorem \ref{thm: null_dist}.  Our next result  shows that the proposed test is consistent against the contiguous alternatives   $\{H_{nm}\}$ \eqref{hnm}. 
\begin{theorem}
\label{thm: power}
Under Assumptions \ref{ass:entropy}, \ref{ass:assumption_weights}, \ref{ass:assumption_lipschitz} and \ref{ass:samp_ratio}, for a sequence of alternatives $H_{nm}$,  %\eqref{hnm} where $a_{nm} \rightarrow 0$ as $n,m \rightarrow\infty$, 
the power of the level $\alpha$ test \eqref{pow} satisfies 
 $\beta^w_{nm} \rightarrow 1.$
\end{theorem}

Theorem~\ref{thm:perm} below provides theoretical guarantees for the permutation version of the test based on the empirical cut-offs %that are obtained in practice 
for the  randomization approximation of $\Gamma_L(\cdot)$ given by $\hat{\Gamma}_{nm}(\cdot)$ \eqref{eq:perm_cdf}. The  consistency of the estimated critical value $\hat{q}_\alpha$ \eqref{qest} 
under $H_0$ is given by \eqref{eq:perm_crit} and under alternatives $P_1 \neq P_2$ we consider  a mixture distribution $\bar{P}=cP_1+(1-c)P_2$ with $0 \le c \le 1$.  Assume $\bar{X}=\{\bar{X}_1, \dots, \bar{X}_{n}\}$ and $\bar{Y}=\{\bar{Y}_1, \dots, \bar{Y}_m\}$ are i.i.d. samples from $\bar{P}$ and $T^w_{nm}(\bar{X},\bar{Y})$ is the test statistic obtained using the samples $\bar{X}$ and $\bar{Y}$. We show in the proof of Theorem~\ref{thm:perm} in the Supplement that under Assumptions \ref{ass:entropy}, \ref{ass:assumption_weights}, \ref{ass:assumption_lipschitz} and \ref{ass:samp_ratio}, Theorem~\ref{thm: null_dist} can be utilized  to obtain the asymptotic distribution of $T^w_{nm}(\bar{X},\bar{Y})$ with  cumulative distribution $\bar{\Gamma}_L(\cdot)$ and $\bar{q}_\alpha = \inf\{ t \geq 0: \bar{\Gamma}_L(t) \geq 1 - \alpha \}$. Suppose that $\bar{\Gamma}_L(\cdot)$ is continuous and strictly increasing at $\bar{q}_\alpha$ and $n,m \rightarrow \infty$ such that $\frac{n}{n+m}-c=O\left((n+m)^{-1/2}\right)$ and $\frac{m}{n+m}-(1-c)=O\left((n+m)^{-1/2}\right)$. Then under Assumptions \ref{ass:entropy}, \ref{ass:assumption_weights}, \ref{ass:assumption_lipschitz} and \ref{ass:samp_ratio}, $\lvert \hat{q}_\alpha-\bar{q}_\alpha \rvert = o_{\prob}(1)$, i.e.,  $\hat{q}_\alpha$ converges to a deterministic limit also for the case where 
 $P_1 \neq P_2$. {This implies convergence of the power function $\tilde{\beta}^w_{nm} \rightarrow 1$ as $n,m \rightarrow \infty$, where 
\begin{equation} \label{tpower}
 \tilde{\beta}^w_{nm} = \prob_{H_{nm}} \left( T^w_{nm} > \hat{q}_\alpha \right)
\end{equation}
is the power function of the test} under the sequence of the contiguous alternatives $H_{nm}$ when  using the permutation-derived critical value $\hat{q}_\alpha$ instead of $q_\alpha$.

\begin{theorem}
\label{thm:perm}
 Under $H_0$ \eqref{eq:null} and Assumptions \ref{ass:entropy}, \ref{ass:assumption_weights}, \ref{ass:assumption_lipschitz} and \ref{ass:samp_ratio}, as $n,m \rightarrow \infty$ and $K \rightarrow \infty$ it holds that 
 $\left|\hat{\Gamma}_{m,n}(t)-\Gamma_L(t) \right| = o_{\prob}(1)$
for every $t$ which is a continuity point of $\Gamma_L(\cdot)$. Suppose that $\Gamma_L(\cdot)$ is continuous and strictly increasing at $q_\alpha$. Then under $H_0$ \eqref{eq:null} and Assumptions \ref{ass:entropy}, \ref{ass:assumption_weights}, \ref{ass:assumption_lipschitz} and \ref{ass:samp_ratio}, as $n,m \rightarrow \infty$ and $K \rightarrow \infty$,
\begin{equation}
\label{eq:perm_crit}
 \left|\hat{q}_\alpha-q_\alpha \right| = o_{\prob}(1). 
\end{equation}
Assume further that $\bar{\Gamma}_L(\cdot)$ is continuous and strictly increasing at $\bar{q}_\alpha$ and $\frac{n}{n+m}-c=O\left((n+m)^{-1/2}\right)$ and $\frac{m}{n+m}-(1-c)=O\left((n+m)^{-1/2}\right)$ as $n,m \rightarrow \infty$. Then under Assumptions \ref{ass:entropy}, \ref{ass:assumption_weights}, \ref{ass:assumption_lipschitz} and \ref{ass:samp_ratio} for the sequence of alternatives $H_{nm}$,  the power \eqref{tpower} of the permutation test satisfies  %where 
$\tilde{\beta}^w_{nm} \rightarrow 1$ as $n,m \rightarrow \infty.$
\end{theorem}

\subsection{Empirical Experiments}\label{sec:test_simu}
To  illustrate the finite-sample performance of the proposed test 
we performed simulation studies for various scenarios. 
Specifically, random objects included samples of random vectors with the Euclidean metric, samples of 2-dimensional distributions with  the $L^2$ metric between corresponding cumulative distribution functions (cdfs), and samples of random networks from the preferential attachment model \citep{bara:99} with the Frobenius metric between the adjacency matrices. 
In each scenario, we generated two samples of random objects of equal size $n = m = 100$ unless otherwise noted and performed $500$ Monte Carlo runs to construct empirical power functions 
as the distance of the distributions of the first and second sample varies.  The empirical power was assessed as the proportion of rejections of the test for the significance level $0.05$ among the $500$ Monte Carlo runs. 
We used the permutation version of the test and assessed $p$-values from $K=1000$ permutations through the proportion of permutations yielding test statistics greater than the test statistics computed from the original sample. This proportion is  $(K+1)^{-1} \sum_{j=0}^{K}\ind[T_{\Pi_j} \geq T^w_{nm}]$, with $T^w_{nm}$ in \eqref{eq:teststat_weighted} and  $T_{\Pi_j}$ for $j=1,\dots,K$ defined in the paragraph just before equation \eqref{eq:perm_cdf}, where the case  $j=0$ corresponds to the original sample without permutation. 

We compared the performance of the proposed test with the energy test \citep{szek:04} and the graph based test \citep{chen:17:3}. 
For the energy test, we obtained $p$-values based on $1000$ permutations. 
For the graph based test, similarity graphs of all the observations pooling the two samples together were constructed as $5$-MSTs, as suggested by \cite{chen:17:3}. Here, MST stands for minimum spanning tree, and a $k$-MST is the union of the $1\text{st},\dots,k\text{th}$ MST(s), where a $k$th MST is a spanning tree connecting all observations while minimizing the sum of distances between connected observations subject to the constraint that all the edges are not included in the $1\text{st},\dots,(k-1)\text{th}$ MST(s). 
In the scenarios with samples of multivariate data, we also included comparisons of the proposed test with the two-sample Hotelling's $T^2$ test. 

In the following figures illustrating power comparisons,  ``energy'' stands for the energy test  \citep{szek:04}; ``graph'' for the graph based test \citep{chen:17:3}; ``Hotelling''  for the two-sample Hotelling's $T^2$ test; and  ``DP'' for the proposed distance profile-based test \eqref{eq:teststat_unweighted}. 
For samples of multivariate data endowed with the Euclidean metric,  we generated the data in four scenarios. 
In the first two scenarios, we generated two samples of $\rdim$-dimensional random vectors $\{X_{i}\}_{i=1}^{n}$ and $\{Y_{i}\}_{i=1}^{m}$ from a Gaussian distribution $N(\mu,\Sigma)$ for dimensions $\rdim=30$, $90$, and $180$, respectively. 
In the first scenario, the population distributions of random vectors in the two samples differ only in the mean $\mu$ while the population covariance matrix $\Sigma = U\Lambda U\tps$ is the same for both samples, where $\Lambda$ is a diagonal matrix with $k$th diagonal entry $\cos(k\pi/\rdim) + 1.5$ for $k=1,\dots,\rdim$ and $U$ is an orthogonal matrix with first column $\rdim^{-1/2}(1,1,\dots,1)\tps$. 
Specifically, $\mu = \mathbf{0}_{\rdim}= (0,0,\dots,0)\tps$ for the first samples $\{X_{i}\}_{i=1}^{n}$, and $\mu =\pone\mathbf{1}_{\rdim} = \pone(1,1,\dots,1)\tps$ for the second samples $\{Y_{i}\}_{i=1}^{m}$, where $\pone$ ranges from $0$ to $1$. The results are shown in Figure~\ref{fig:mvnorm-mean-shift}. In addition,  we considered another location shift scenario for lower dimensional data; see Section~\ref{sec:simu_twosam_supp} in the Supplement for details.

In the second scenario, the population distributions of the two samples differ only in scale, while sharing the same mean $\mu = \mathbf{0}_{\rdim}$. For the first samples $\{X_{i}\}_{i=1}^{n}$, $\Sigma = 0.8 I_{\rdim}$ and for the second samples $\{Y_{i}\}_{i=1}^{m}$, $\Sigma = (0.8-\ptwo) I_{\rdim}$ with $\ptwo$ ranging from $0$ to $0.4$. The results are shown in Figure~\ref{fig:mvnorm-scale-change}. 
In the first scenario with location shifts, the proposed test outperforms the graph based test when the dimension is relatively low ($\rdim=30$); the graph based test catches up when the dimension is high.  Meanwhile, the energy test is always the winner in this scenario. The performance of Hotelling's $T^2$ test is the second best when the dimension $p$ is less than the sample size $n=m=100$ but drops to the bottom when $\rdim=180$ and $\pone>0.1$. 
In the second scenario with scale changes, we find that the proposed test always outperforms all the other tests. 

\begin{figure}[!hbt]
 \centering
 \includegraphics[width=\textwidth]{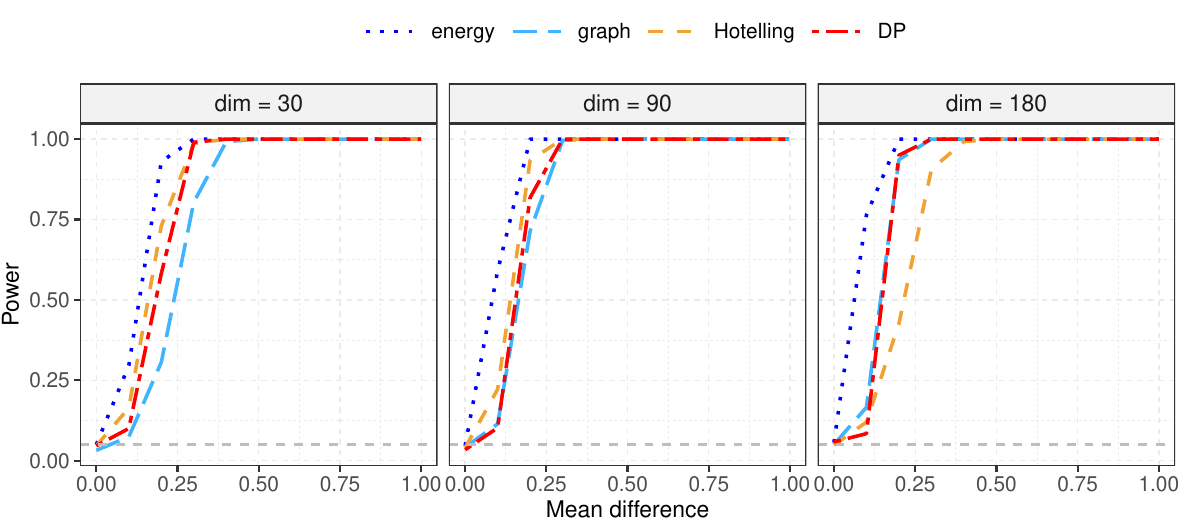}
 \caption{Power comparison for increasing values of mean difference $\pone$ for two samples of $\rdim$-dimensional random vectors sampled from $N(\mu,\Sigma)$. Here, $\mu =\mathbf{0}_{\rdim}= (0,0,\dots,0)\tps$ for the first samples and $\mu =\pone\mathbf{1}_{\rdim} = \pone(1,1,\dots,1)\tps$ for the second samples; $\Sigma = U\Lambda U\tps$ for both samples, where $\Lambda$ is a diagonal matrix with $k$th diagonal entry being $\cos(k\pi/\rdim) + 1.5$ for $k=1,\dots,\rdim$ and $U$ is an orthogonal matrix with the first column being $\rdim^{-1/2}(1,1,\dots,1)\tps$. The dashed grey line denotes the significance level $0.05$.}
 \label{fig:mvnorm-mean-shift}
\end{figure}

\begin{figure}[!hbt]
 \centering
 \includegraphics[width=\textwidth]{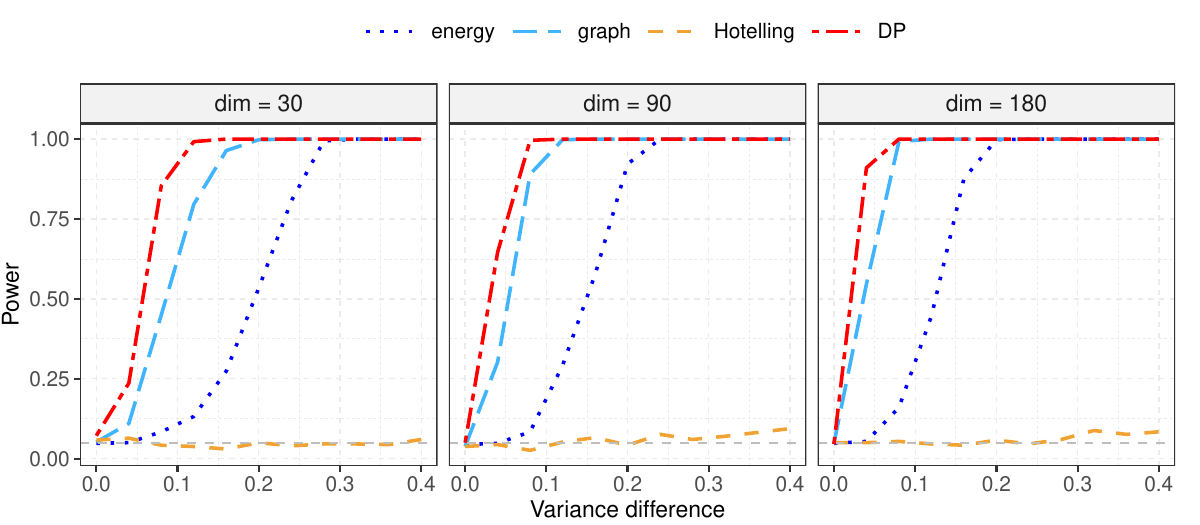}
 \caption{Power comparison for increasing values of variance difference $\ptwo$ for two samples of $\rdim$-dimensional random vectors sampled from $N(\mu,\Sigma)$. Here, $\mu = \mathbf{0}_{\rdim}$ for both samples; $\Sigma = 0.8 I_{\rdim}$ for the first samples and $\Sigma = (0.8 - \ptwo) I_{\rdim}$ for the second samples. The dashed grey line denotes the significance level $0.05$.}
 \label{fig:mvnorm-scale-change}
\end{figure}

In the third scenario, the first samples of random vectors $\{X_{i}\}_{i=1}^{n}$ are generated from the Gaussian distribution $N(\mathbf{0}_{\rdim}, I_{\rdim})$, and the second samples of random vectors $\{Y_{i}\}_{i=1}^{m}$ are generated from a mixture of two Gaussian distributions with the overall population mean equaling that of the first samples. 
Specifically, the second samples consist of independent copies of $AZ_1 + (1-A)Z_2$, where $A\sim\mathrm{Bernoulli}(0.5)$, $Z_1\sim N(-\mu, I_{\rdim})$, $Z_2\sim N(\mu, I_{\rdim})$, $\mu = (\pthr\mathbf{1}_{0.1\rdim}\tps, \mathbf{0}_{0.9\rdim}\tps)\tps$, and $A$, $Z_1$, and $Z_2$ are independent. Here, $\pthr$ ranges from $0$ to $1$. 
The results are shown in Figure~\ref{fig:mvnorm-vs-mixture}; again,  the proposed test outperforms all the other tests.

\begin{figure}[!hbt]
 \centering
 \includegraphics[width=\textwidth]{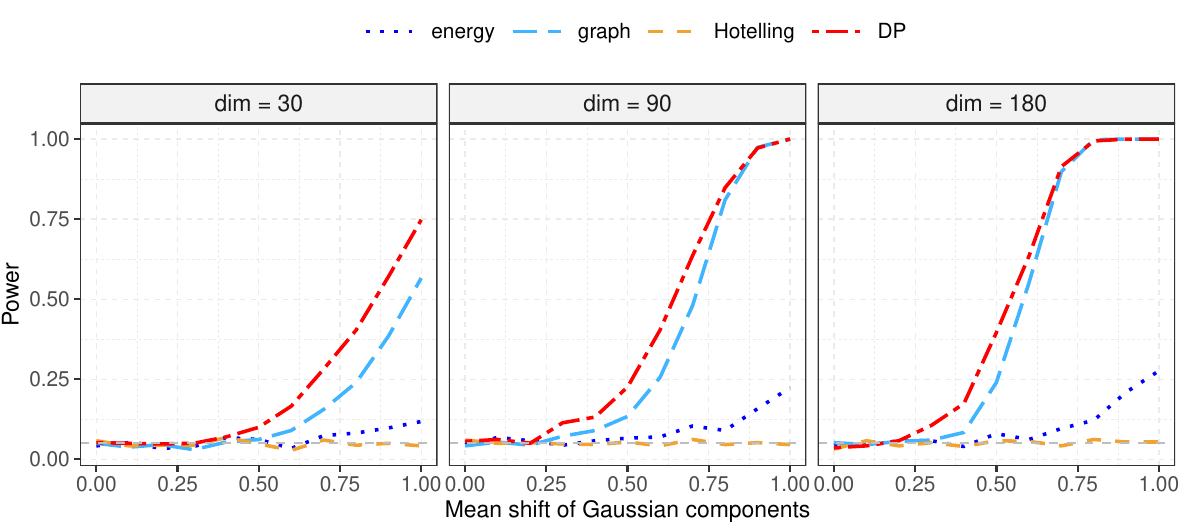}
 \caption{Power comparison for increasing values of mean shift $\pthr$ of Gaussian components for two samples of $\rdim$-dimensional random vectors. Here, the first samples are sampled from $N(\mathbf{0}_{\rdim},I_{\rdim})$; the second samples consist of $AZ_1 + (1-A)Z_2$, where $A\sim\mathrm{Bernoulli}(0.5)$, $Z_1\sim N(-\mu, I_{\rdim})$, $Z_2\sim N(\mu, I_{\rdim})$, $\mu = (\pthr\mathbf{1}_{0.1\rdim}\tps, \mathbf{0}_{0.9\rdim}\tps)\tps$, and $A$, $Z_1$, and $Z_2$ are independent. The dashed grey line denotes the significance level $0.05$.}
 \label{fig:mvnorm-vs-mixture}
\end{figure}

The fourth scenario compares Gaussian distributions with heavy-tailed distributions, where the first samples $\{X_{i}\}_{i=1}^{n}$ are generated from $N(\mathbf{0}_{\rdim},I_{\rdim})$, and the second samples $\{Y_{i}\}_{i=1}^{m}$ consist of random vectors with  components that are independent and identically distributed following a $t$ distribution with degrees of freedom ranging from $2$ to $22$. 
The results for $\rdim\in\{5,15,60\}$ are in Figure~\ref{fig:mvnorm-vs-mvt};  the proposed test outperforms all the other tests.

\begin{figure}[!hbt]
 \centering
 \includegraphics[width=\textwidth]{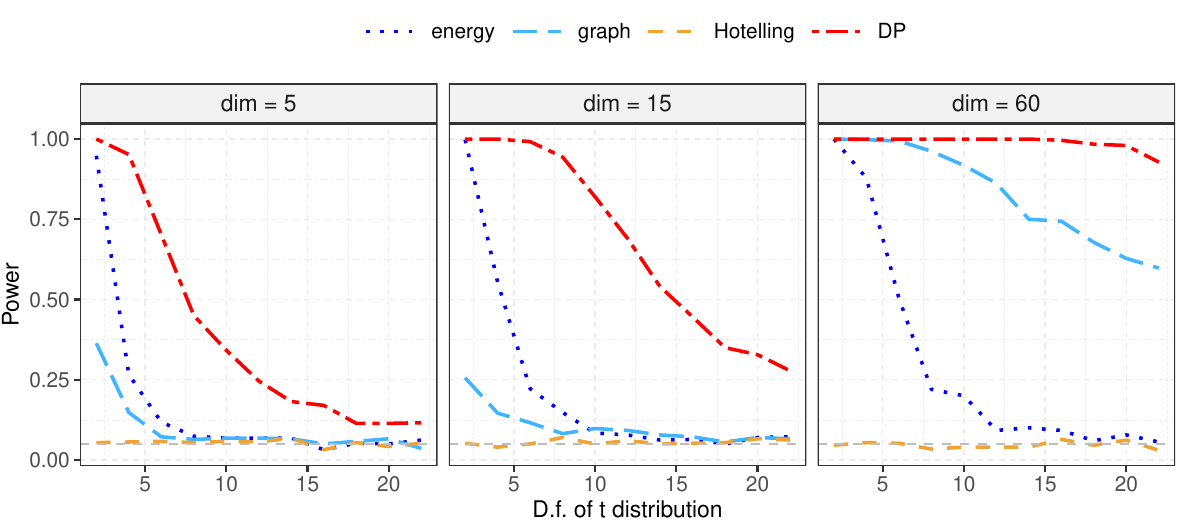}
 \caption{Power comparison for increasing values of $\pfou$ for two samples of $\rdim$-dimensional random vectors. Here, the first samples are sampled from $N(\mathbf{0}_{\rdim},I_{\rdim})$; the second samples consist of random vectors with independent components, where each component follows a $t$ distribution with $\pfou$ degrees of freedom (d.f.). The dashed grey line denotes the significance level $0.05$.}
 \label{fig:mvnorm-vs-mvt} \vspace{-.8cm} 
\end{figure}

Next, we considered  bivariate probability distributions as random objects, where  we use the $L^2$ distance between corresponding cdfs as the metric between two probability distributions. 
Each observation $X_{i}$ or $Y_{i}$ is a random $2$-dimensional Gaussian distribution $N(Z,0.25 I_{2})$, where $Z$ is a $2$-dimensional random vector, with two scenarios: 
In the first scenario, $Z\sim N(\mathbf{0}_{2}, 0.25 I_{2})$ for the first samples and $Z\sim N((\pfiv,0)\tps, 0.25 I_{2})$ for the second samples. 
In the second scenario, $Z\sim N(\mathbf{0}_{2}, 0.4^2 I_{2})$ for the first samples and $Z\sim N(\mathbf{0}_{2}, \diag((0.4+\psix)^2 I_{2}))$ for the second samples. 
The results  are presented in Figures~\ref{fig:2Ddistn-mean-shift-in-mean} and \ref{fig:2Ddistn-scale-change-in-mean}, respectively. 
The first scenario showcases location shifts of $Z$; the proposed test outperforms the graph based test but is outperformed by the energy test. 
The second scenario showcases scale changes of $Z$, where  the proposed test outperforms all the other tests. 

\begin{figure}[!hbt]
 \centering
 \includegraphics[width=.5\textwidth]{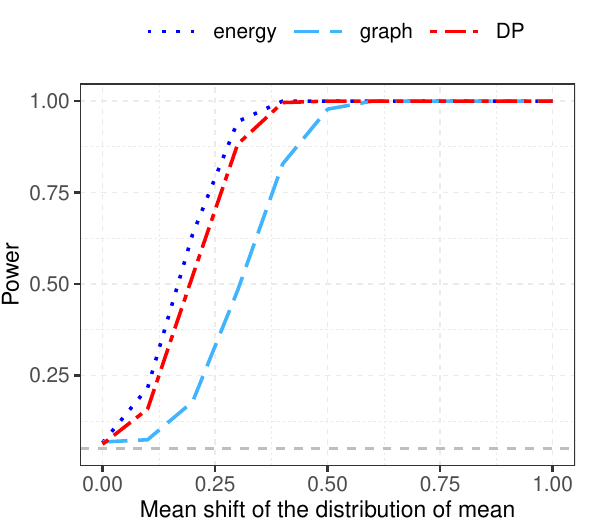}
 \caption{Power comparison for increasing values of mean shift $\pfiv$ of the distribution of the mean for two samples of random bivariate Gaussian distributions $N(Z,0.25 I_{2})$, where $Z\sim N(\mathbf{0}_{2}, 0.25 I_{2})$ for the first samples and $Z\sim N((\pfiv,0)\tps, 0.25 I_{2})$ for the second samples. The dashed grey line denotes the significance level $0.05$.}
 \label{fig:2Ddistn-mean-shift-in-mean}
\end{figure}

\begin{figure}[!hbt]
 \centering
 \includegraphics[width=.5\textwidth]{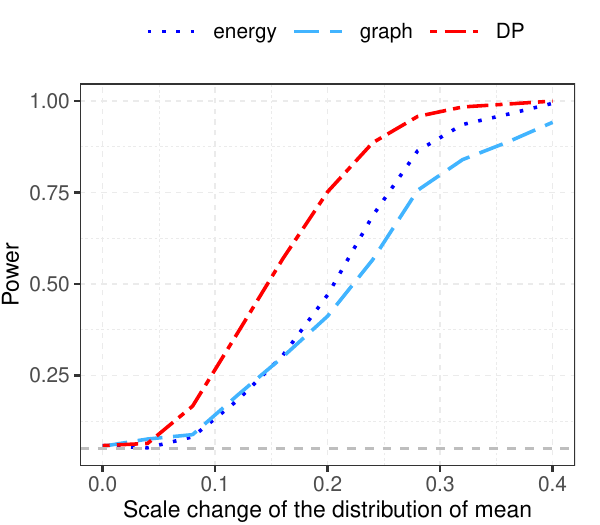}
 \caption{Power comparison for increasing values of scale change $\psix$ of the distribution of the mean for two samples of random bivariate Gaussian distributions $N(Z,0.25 I_{2})$, where $Z\sim N(\mathbf{0}_{2}, 0.4^2 I_{2})$ for the first samples and $Z\sim N(\mathbf{0}_{2}, \diag((0.4+\psix)^2 I_{2})$ for the second samples. The dashed grey line denotes the significance level $0.05$.}
 \label{fig:2Ddistn-scale-change-in-mean}
\end{figure}

We also studied the power of the proposed test for random networks endowed with the  Frobenius metric between adjacency matrices as random objects. Each datum $X_{i}$ or $Y_{i}$  is a random network with $200$ nodes generated from the preferential attachment model \citep{bara:99} with the attachment function proportional to $k^{\psev}$, where $\psev=0$ for the first samples $\{X_{i}\}_{i=1}^{n}$, 
and $\psev$ increases from $0$ to $0.5$ for the second samples $\{Y_{i}\}_{i=1}^{m}$. 
As shown in Figure~\ref{fig:network}, the proposed test 
outperforms both the energy test and the graph based test. 
In addition to the case with $n=m=100$, we performed simulations with larger samples of sizes $n=m=200$. The corresponding results are shown in Figure~\ref{fig:network200} in Section~\ref{sec:simu_twosam_supp} in the Supplement and they more or less match those for $n=m=100$. 

\section{Extensions and Data Illustrations}\label{sec:data}

\subsection{Profile Metric and Object Data Visualization}\label{sec:one_sample_app}

Distance profiles induce a new similarity measure in $\O$  that we refer to as  {\it profile metric} $d_P$. It complements the original metric $d$ and depends on $d$, the underlying probability measure $P$ and also on the distributional metric in the space of distance profiles, for which we select the Wasserstein metric.  
The profile metric quantifies the distance of the profile densities of elements of $\Om$,  
\be \label{tm} d_{P}(\om_1,\om_2)=d_W(F_{\om_1},F_{\om_2}), \ee
where $d_W$ is the Wasserstein metric \eqref{eq:dwass} and $F_{\om_1}$, $F_{\om_2}$ are the distance profiles of $\om_1,\om_2$, as defined in \eqref{dp}. It 
is not a genuine metric on $\Om$ but rather a measure of dissimilarity of the distance profiles of elements of $\Om.$ 

The profile metric $d_{P}$ generally may differ substantially from the original metric $d$. For example two outlying elements of $\Om$  may be far away from each other in terms of the original metric $d$ but if they have similar centrality and distance profiles  they will have small profile dissimilarity $d_{P}$ which could be 0 if their distance profiles coincide.
It turns out that the profile metric is very useful for data analysis, as we will demonstrate in the following. Its implementation  depends on distance profiles which must be estimated from the available data, and thus the profile metric itself  is only available in the form of an estimate.

To visualize random objects, low-dimensional projections of similarities as afforded by MDS are a prime tool and any  MDS version \citep{mard:78} can be based on either the original distance $d$, in the following  referred to as {\it object MDS} or alternatively on the profile metric,  in the following referred to as {\it profile MDS}. 
In profile MDS, we use the estimated distance profiles $\hfi$ and the Wasserstein metric $\dwass$ \eqref{eq:dwass}, while we use the distance $d$ in $\Om$ for object MDS. 
In the following, MDS is implemented with \texttt{cmdscale()} in the R built-in package \texttt{stats} \citep{R}. 

\begin{figure}[!hbt]
 \centering
 \includegraphics[width=.5\textwidth]{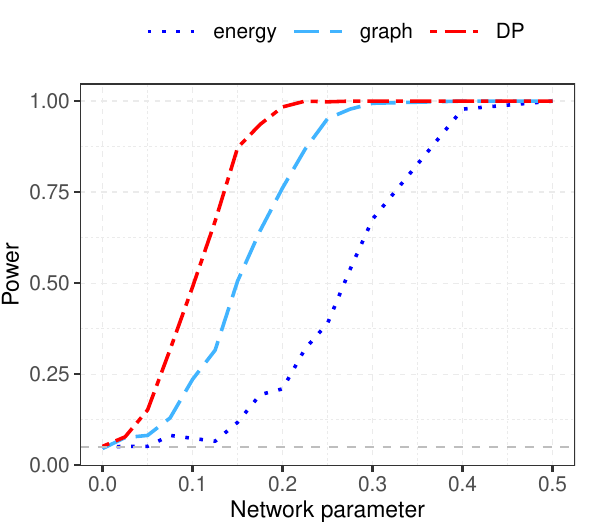}
 \caption{Power comparison for increasing values of $\psev$ for two samples of random networks with 200 nodes  from the preferential attachment model. The attachment function is   proportional to $k^{\psev}$ with  $\psev=0$ for the first samples,
 and $\psev$ increasing from $0$ to $0.5$ for the second samples. The dashed grey line denotes the significance level $0.05$.}
 \label{fig:network}
\end{figure}

To enhance the graphical illustration of the proposed transport ranks in \eqref{eq:rank}, for implementations, data applications and simulations we found that partitioning  the observed random objects into a not too large number of $k$  groups according to their transport ranks is advantageous for visualization and communicating results.  Specifically, the range of the transport ranks of observations within a sample $\{\obj_{\subidx}\}_{i=1}^{n}$ is partitioned into $k$ bins, $\clust_{k}=[0,\qt_{1/k}],\, \clust_{k-1}=(\qt_{1/k},\qt_{2/k}],\dots, \clust_{1}=(\qt_{(k-1)/k},1]$, where $\qt_{\alpha}$ is the $\alpha$-quantile of $\{\hrank_{\obj_{\subidx}}\}_{i=1}^{n}$ for $\alpha\in(0,1)$; then the $j$-th group consists of observations with transport ranks falling in $\clust_{j}$ for $j=1,\dots,k$. Arranging the bins in descending order of transport ranks, 
these groups are ordered from the innermost to the outermost, providing a center-outward description of the data; we found that the choice $k=10$ worked well, as illustrated in Figures~\ref{fig:2dGauss}--\ref{fig:energy2000_ternary_rankAndATD} below. The function \texttt{CreateDensity()} in the R package \texttt{frechet} \citep{frechet} was used to obtain Wasserstein barycenters of distance profiles for each group. 

The code for obtaining distance profiles, transport ranks, object MDS plots, and profile MDS plots is available on GitHub \citep{ODP}. 

\subsection{Illustrations With Simulated Data}\label{sec:rank_simu}
We start with a simple special case of a Euclidean vector space, where 
we sampled $n=500$ observations $\{\obj_{\subidx}\}_{\subidx=1}^n$ independently from a $p$-dimensional Gaussian distribution $N(\bm\mu,\bm\Sigma)$ for $p=2$ and $p=50$, with $\bm\mu=\bm{0}$ and $\bm\Sigma = \diag(p,p-1,\dots,1)$. The distance profiles $\hfi$ \eqref{eq:FhatO} and transport ranks $\hrank_{\obj_{\subidx}}$ \eqref{eq:hrank} were computed for each observation, adopting the Euclidean metric in $\real^p$.
Irrespective of the type of random objects $X_i$, the distance profiles $\hfi$ are situated in the space $\wsp$ of one-dimensional distributions with finite second moments  with the Wasserstein metric  \eqref{eq:dwass}.

For $p=2$, the transport ranks \eqref{eq:hrank} based on distance profiles capture the center-outward ordering of the $2$-dimensional Gaussian data 
and the Wasserstein barycenters of the distance profiles within each group shift to the right from group 1 to group 10, where the grouping is as described in Section 7.1, reflecting increased distances from the bulk of data (Figure~\ref{fig:2dGauss}). 
Figure~\ref{fig:50dGauss} demonstrates profile MDS  for a simulated sample of $n=500$ observations from a $50$-dimensional Gaussian distribution $N(\bm\mu,\bm\Sigma)$ with $\bm\mu=\bm{0}$ and $\bm\Sigma = \diag(50,49,\dots,1)$ and shows that profile MDS provides a simple representation by sorting these high-dimensional Euclidean data along dimension 1. 

Additional simulation results can be found in the Supplement for random objects corresponding to   $2$-dimensional random vectors generated from multi-modal distributions in 
Section~\ref{sec:simu_mv};
and for  distributional data in Section~\ref{sec:simu_distn}. 

\begin{figure}[!hbt]
 \centering
 \begin{subfigure}[b]{0.32\linewidth}
 \centering
 \includegraphics[width=\linewidth]{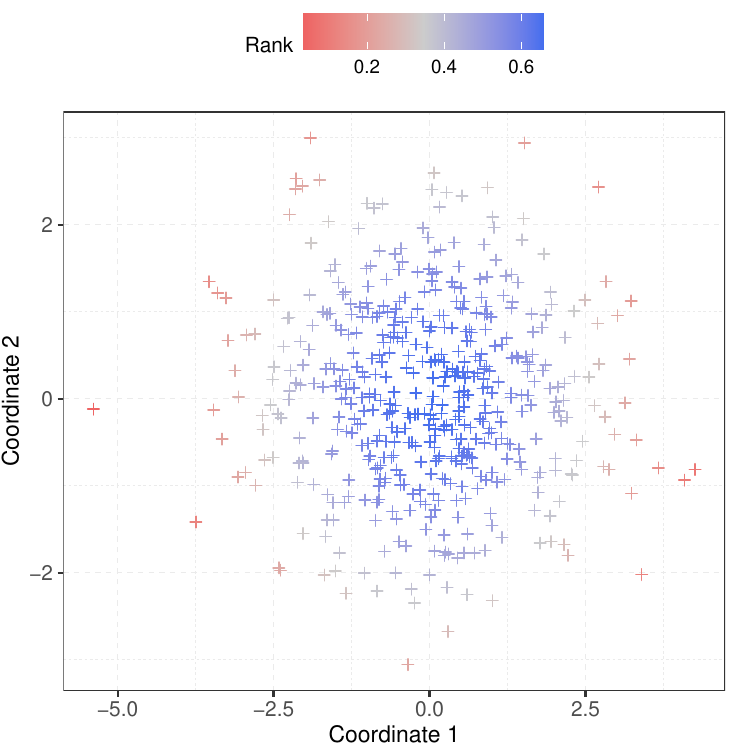}
 \end{subfigure}
  \begin{subfigure}[b]{0.32\linewidth}
 \centering
 \includegraphics[width=\linewidth]{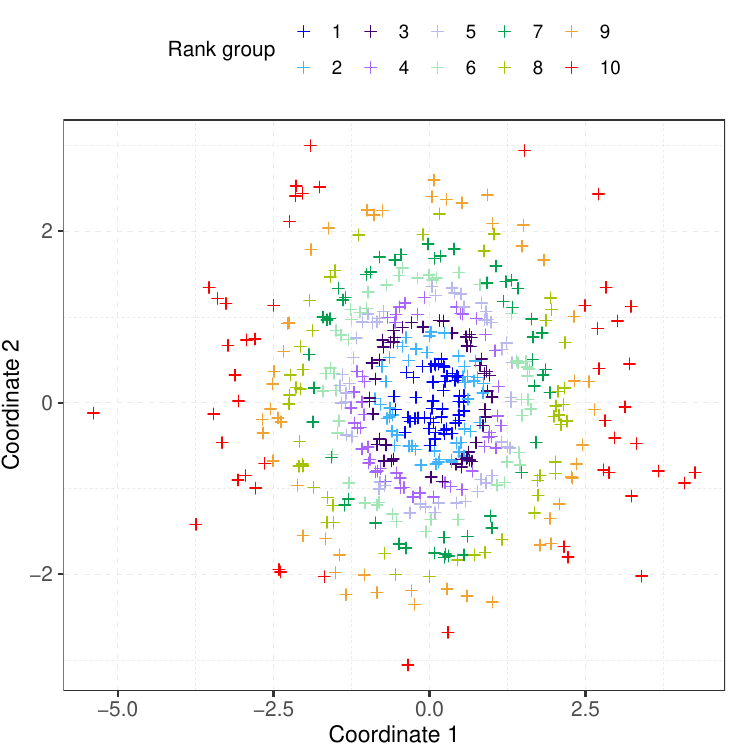}
 \end{subfigure}
 \begin{subfigure}[b]{0.32\linewidth}
 \centering
 \includegraphics[width=\linewidth]{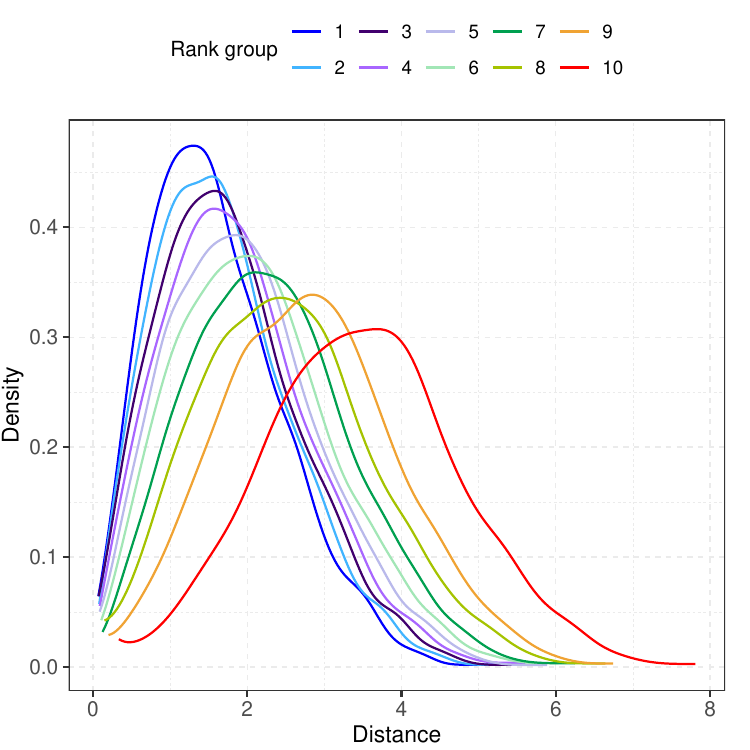}
 %\caption{Within-cluster barycenters of the distance profiles.} \label{subfig:2dGauss_clustMeanDepthFctn}
 \end{subfigure}
 \caption{Scatterplots of a sample of $n=500$ observations generated from a $2$-dimensional Gaussian distribution $N(\bm\mu,\bm\Sigma)$ with $\bm\mu=\bm{0}$ and $\bm\Sigma = \diag(2,1)$, where the points are colored according to their transport ranks \eqref{eq:hrank} (left) and grouped into 10 groups according to the quantiles of transport ranks (middle); % 
 Wasserstein barycenters of the distance profiles within each group represented by density functions (right).} \label{fig:2dGauss}
\end{figure}

\begin{figure}[!hbt]
 \centering
 \begin{subfigure}[b]{0.32\linewidth}
 \centering
 \includegraphics[width=\linewidth]{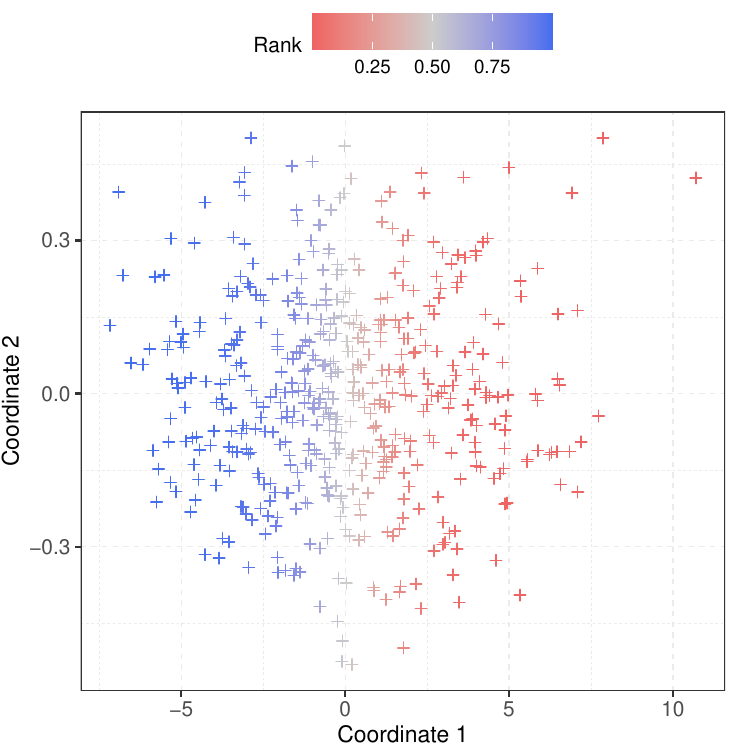}
 % \caption{Colored by ranks.} %\label{fig:50dGauss_colByRank}
 \end{subfigure}
 \begin{subfigure}[b]{0.32\linewidth}
 \centering
 \includegraphics[width=\linewidth]{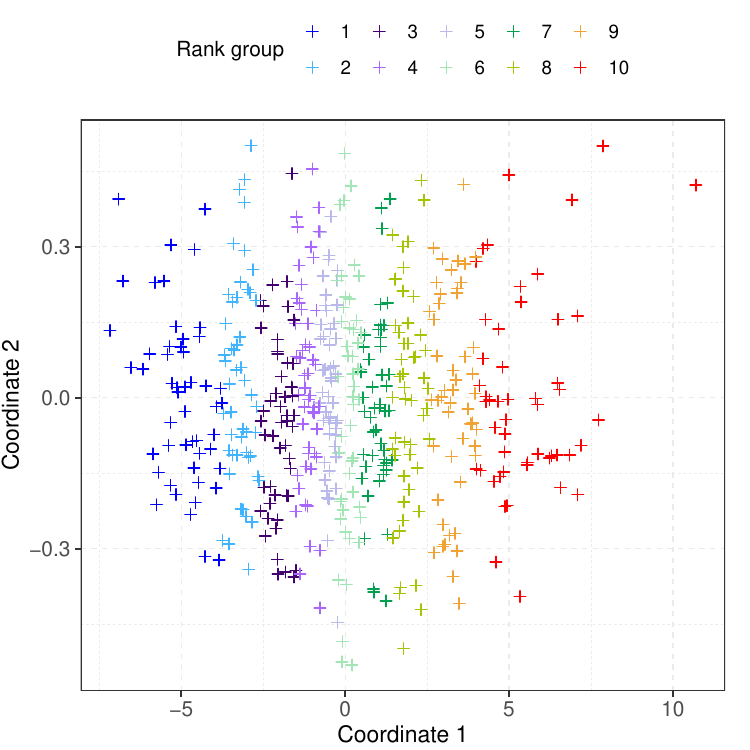}
 % \caption{$k$-W-means clustering.} %\label{subfig:50dGauss_dpMdsColByClust}
 \end{subfigure}
 \begin{subfigure}[b]{0.32\linewidth}
 \centering
 \includegraphics[width=\linewidth]{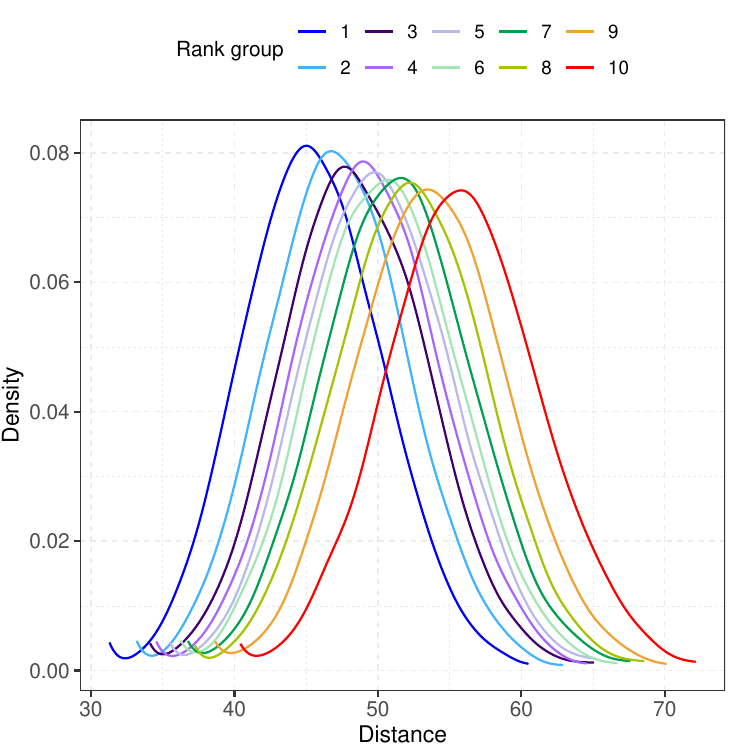}
 \end{subfigure}
 \caption{Two-dimensional (profile) MDS with respect to the Wasserstein metric $\dwass$ in \eqref{eq:dwass} of the distance profiles $\hfi$ \eqref{eq:FhatO} with $\o=\obj_{\subidx}$ of a sample of $n=500$ observations generated from a $50$-dimensional Gaussian distribution $N(\bm\mu,\bm\Sigma)$ with $\bm\mu=\bm{0}$ and $\bm\Sigma = \diag(50,49,\dots,1)$, where the points are colored according to their transport ranks \eqref{eq:hrank} (left) and grouped into 10 groups according to the quantiles of transport ranks (middle); % 
 Wasserstein barycenters of the distance profiles within each group (right). }
 \label{fig:50dGauss}
\end{figure}

\subsection{Illustration With Human Mortality Data}\label{sec:mort_rank}
%\vspace{-.3cm} 
Understanding human longevity has been of long-standing interest and age-at-death distributions are relevant random objects for this endeavor.  We consider age-at-death distributions for different countries, obtained from the Human Mortality Database (\url{http://www.mortality.org}) for the year 2000 for $n=34$ countries,  
separately for males and females. 
The age-at-death distributions are shown in the form of density functions in Figure~\ref{fig:mort2000_data} in Section~\ref{sec:mort_supp} in the Supplement.
To analyze the data geometry of this sample of random distributions $\{\obj_{\subidx}\}_{\subidx=1}^{34}$, we assume that they are situated in a space $(\Om,d)$ of distributions equipped  with the Wasserstein metric $d=d_W$    \eqref{eq:dwass} and then obtained distance profiles $\hfi$ \eqref{eq:FhatO} for  $\o=\obj_{\subidx}$ for each country.  

\begin{figure}[hbt!]
 \centering
 \begin{subfigure}[t]{0.45\linewidth}
 \centering
 \includegraphics[width=\linewidth]{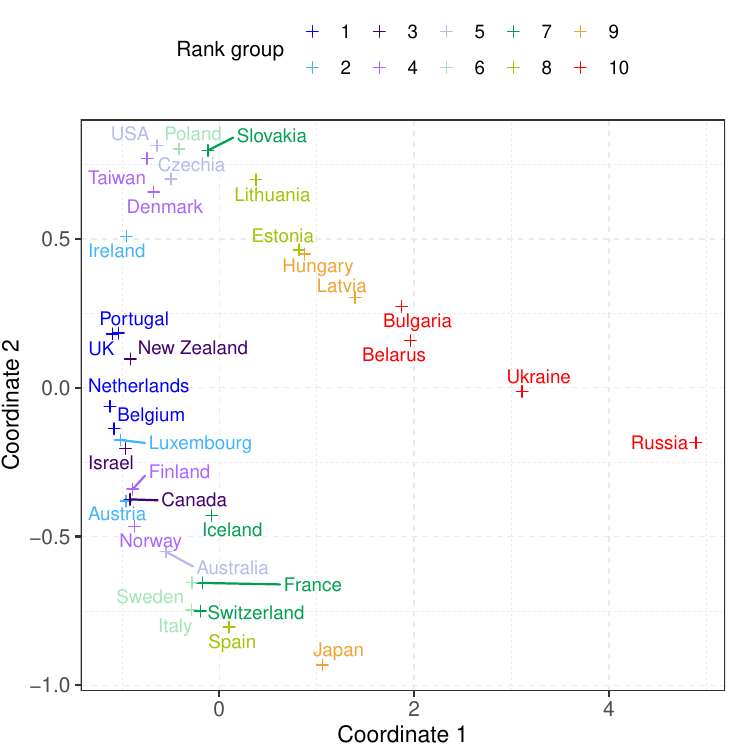}
 % \caption{Clustering.}
 \end{subfigure}
 \begin{subfigure}[t]{0.45\linewidth}
 \centering
 \includegraphics[width=\linewidth]{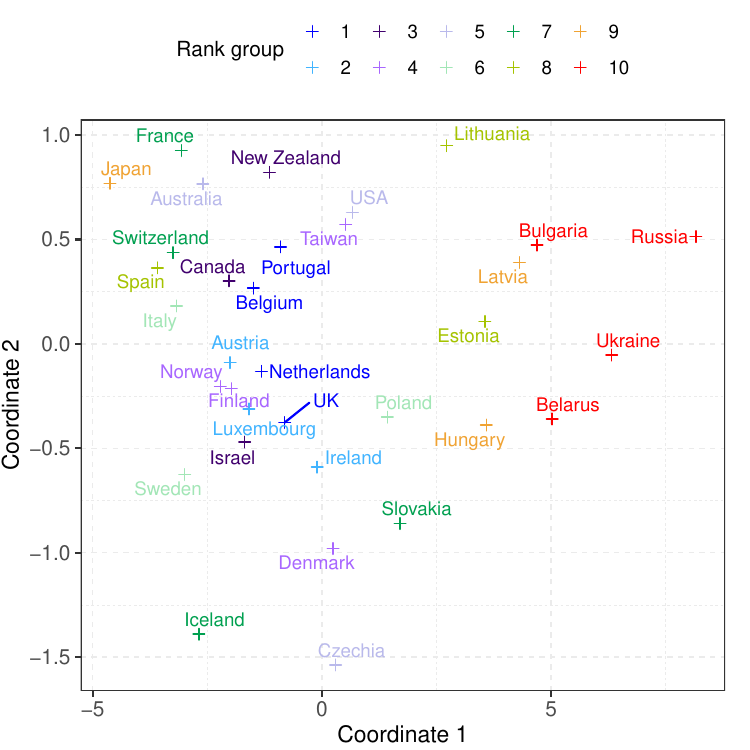}
 % \caption{Clustering.}
 \end{subfigure}
 \begin{subfigure}[t]{0.45\linewidth}
 \centering
 \includegraphics[width=\linewidth]{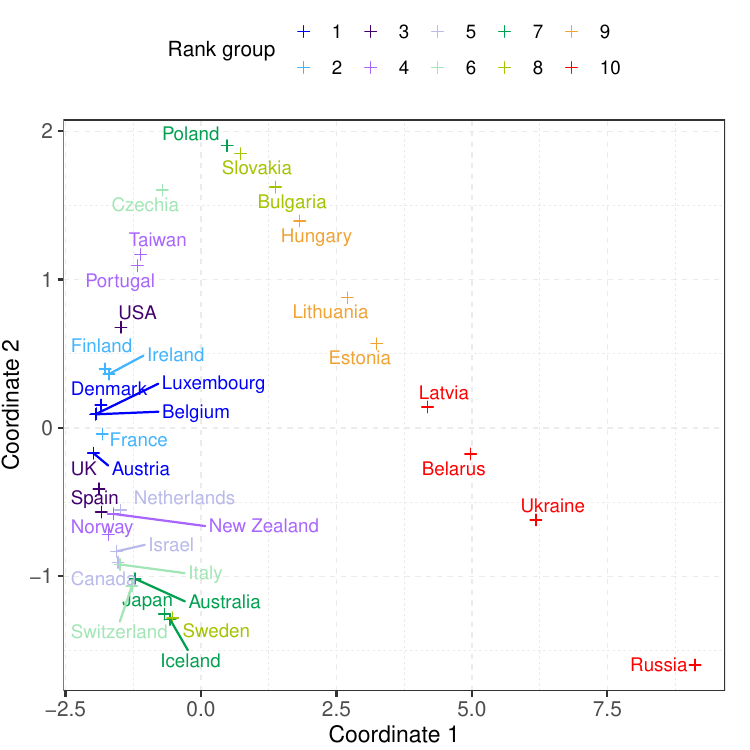}
 % \caption{Clustering.}
 \end{subfigure}
 \begin{subfigure}[t]{0.45\linewidth}
 \centering
 \includegraphics[width=\linewidth]{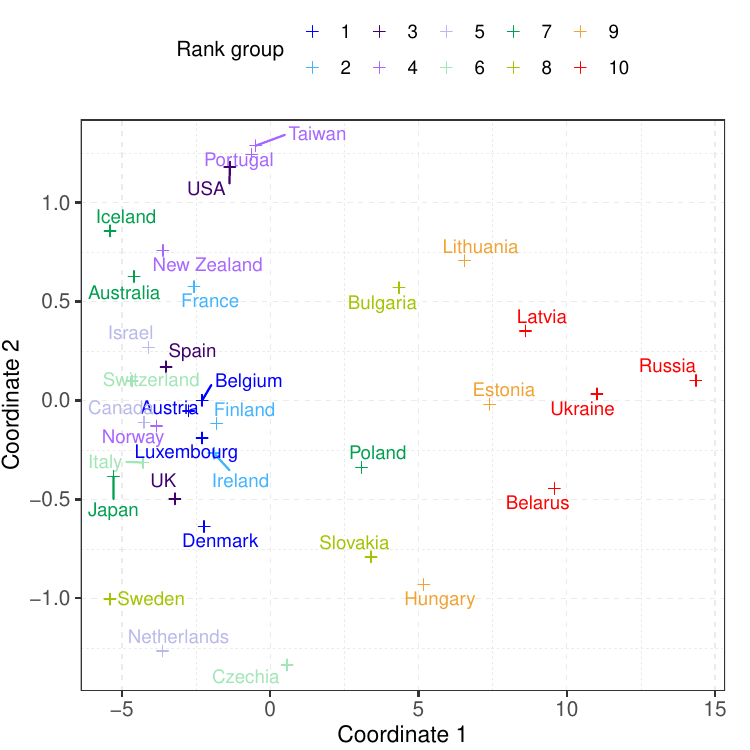}
 %\caption{}
 \end{subfigure}
 \caption{Two-dimensional profile MDS (left) and object MDS (right) 
 of the age-at-death distributions of females (top) and males (bottom) in 2000 for 34 countries, where the objects are grouped into 10 groups and colored according to the quantiles of their transport ranks.}  
 \label{fig:mort2000_mds}
\end{figure}

In Figure~\ref{fig:mort2000_mds} we compare  profile MDS (based on the distance of profiles, where the Wasserstein metric is applied for the distributional space where the profiles are situated) in the left panels and 
object MDS (based on the original metric $d$ in the object space of 
distributions); mortality for females is shown in the top panels and for males in the bottom panels. 
We find that profile MDS leads to a clearly interpretable one-dimensional manifold representation for both females and males, where extremes appear at each end, at the red colored end corresponding to age-at-death distributions indicating reduced  and at the green colored end enhanced longevity.  The groups of countries  that form the extreme ends are Japan at the enhanced  and  Eastern European countries, such as Russia, Ukraine, Belarus, Latvia, and Estonia at the reduced longevity end.
Luxembourg and Belgium belong to the most central group for both females and males. Spain is among the more outlying countries for females only, with longevity increase for females but not for males. One can observe many other interesting features in terms of the similarity and contrast between females' and males' longevity for specific countries, e.g., for Denmark and Netherlands. We find  that the one-dimensional ordering provided by profile MDS facilitates the interpretation and communication of the main data features, while object MDS  is less informative.

Another finding of interest that emerges from profile MDS  is that  the age-at-death distributions for males for the outlying countries are more outlying than the corresponding age-at-death distributions  for  females. 
In particular, the empirical \F variance of the distance profiles, $n\inv\sum_{i=1}^n\dwass^2(\hfi,\hfmean)$, of age-at-death distributions for females and males of different countries is 2.08 and 8.22, respectively, where $\hfmean = \argmin_{\o\in\wsp}\sum_{i=1}^n\dwass^2(\hfi,\o)$ is the empirical \F mean of the distance profiles, indicating that male age-at-death is especially sensitive to unfavorable country-specific factors such as the lingering effects of societal upheaval in  Eastern Europe.

\subsection{Illustration With U.S. Electricity Generation Data}
\label{sec:energy}

Compositional data comprise another type of data that do not lie in a vector space. Such data are commonly encountered and consist of vectors of non-negative elements that sum up to 1. Examples include geochemical compositions and microbiome data. Various approaches to handle the nonlinearity that is inherent in such data have been developed \citep{aitc:86,scea:14,filz:18}. We consider here the U.S. electricity generation data which are publicly available on the website of the U.S. Energy Information Administration (\url{http://www.eia.gov/electricity}). 
The data consist of net generation of electricity from different sources for each state. We considered the data for the year 2000. In preprocessing, we excluded the ``pumped storage'' category due to errors in these data and then 
merged the other energy sources into three categories: Natural Gas, consisting of ``natural gas'' alone; Other Fossil, consisting of ``coal'', ``petroleum'' and ``other gases''; Renewables and Nuclear, combining the remaining sources ``hydroelectric conventional'', ``solar thermal and photovoltaic'', ``geothermal'', ``wind'', ``wood and wood derived fuels'', ``other biomass'', ``nuclear'' and ``other''. 
Hence, we have a sample of $n=50$ observations $\{\obj_{\subidx}\}_{\subidx=1}^{n}$, each of which takes values in a 2-simplex $\splx =\{ \bm{x}\in\real^3: \bm{x}^\top\mathbf{1}_3 = 1 \}$, where $\mathbf{1}_3 = (1,1,1)^\top$. 
Since the component-wise square root $\sqrt{\bm{x}} = (\sqrt{x_1},\sqrt{x_2},\sqrt{x_3})^\top$ of an element $\bm{x}\in\splx$ lies in the unit sphere $\sphe$, 
we adopted the geodesic metric on this sphere 
\bal\label{eq:dsphe}
\dsphe(\bm{x},\bm{y}) = \arccos(\sqrt{\bm{x}}^\top \sqrt{\bm{y}}),\text{ for } \bm{x},\bm{y}\in\splx. \eal

We then compared the proposed transport ranks \eqref{eq:hrank} for each state with the angular Tukey depths \citep[ATDs,][]{liu:92:2} of $\{\sqrt{\obj_{\subidx}}\}_{\subidx=1}^{n}$. 
At first glance, the proposed transport ranks and ATDs yield similar center-outward ordering of the 50 states for these data (Figure~\ref{fig:energy2000_ternary_rankAndATD}). 
Maryland emerges as the transport median and is also at the median in terms of ATDs. 
On closer inspection, one finds some interesting discrepancies between transport ranks and the ATDs, especially for the states that are either close to or far away from the center Maryland in terms of their outlyingness. 
The states near Maryland, as shown in orange and light violet in the bottom panels of Figure~\ref{fig:energy2000_ternary_rankAndATD}, all have high transport ranks, while their ATDs vary widely. 
In particular, Montana, with an electricity generation pattern very similar to that of Maryland, has the lowest ATD level while it has a high transport rank.
A subset of states that are colored in turquoise and light violet in the bottom panels of Figure~\ref{fig:energy2000_ternary_rankAndATD} have the lowest ATDs among all states but have a much wider range of transport ranks. 
For example, 
Hawaii and Delaware for which energy sources are similar to those of Maryland have high transport ranks and low ATD levels. The overall conclusion is that transport ranks are better suited than ATDs for studying the geometry of this data set and for quantifying outlyingness. 

\begin{figure}[htbp]%[hbt!]
 \centering
  \begin{subfigure}[t]{0.45\linewidth}
 \centering
 \includegraphics[width=\linewidth]{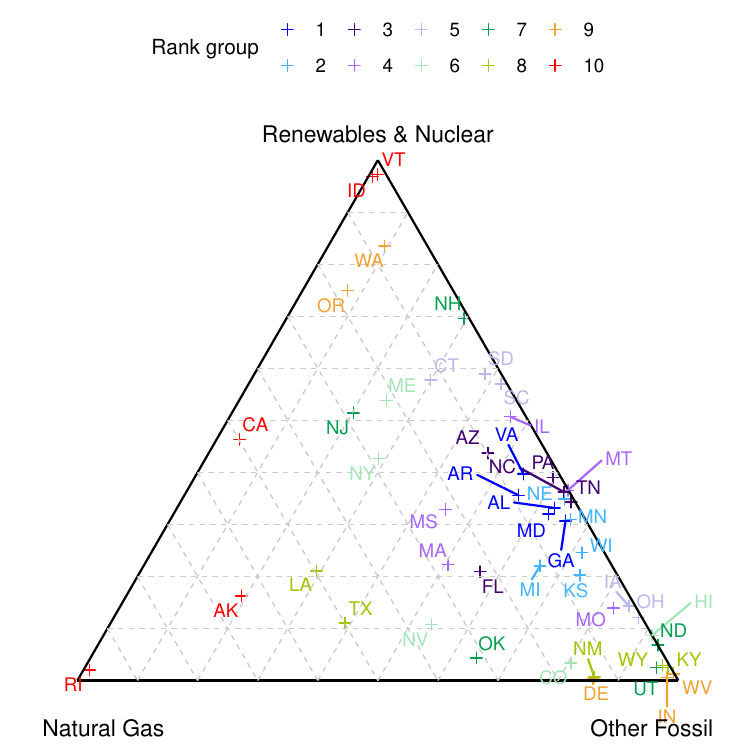}
 % \caption{Clustering.}
 \end{subfigure}
 \begin{subfigure}[t]{0.45\linewidth}
 \centering
 \includegraphics[width=\linewidth]{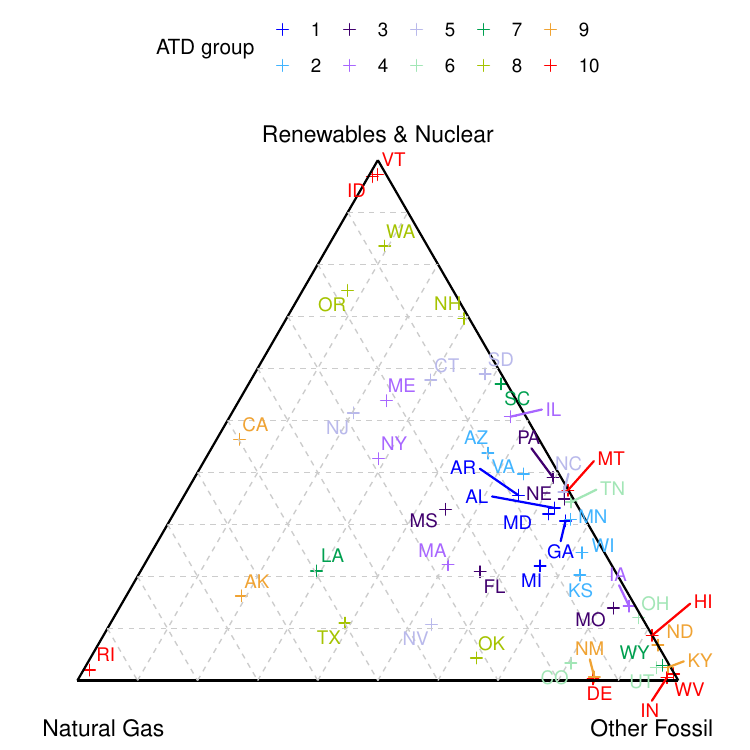}
 % \caption{Clustering.}
 \end{subfigure}\\\vspace{1em}
 \begin{subfigure}[t]{0.45\linewidth}
 \centering
 \includegraphics[width=\linewidth]{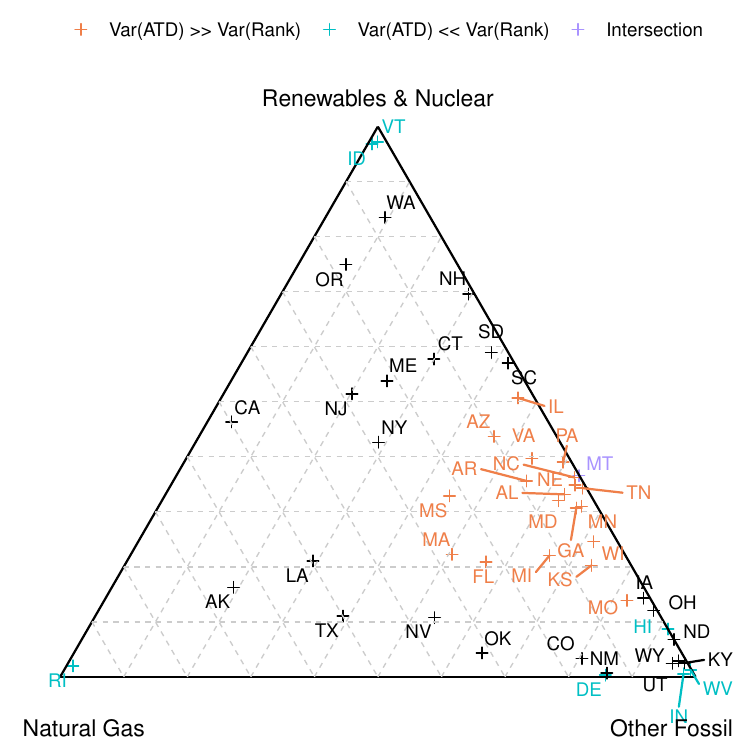}
 \end{subfigure}
 \begin{subfigure}[t]{0.45\linewidth}
 \centering
 \includegraphics[width=\linewidth]{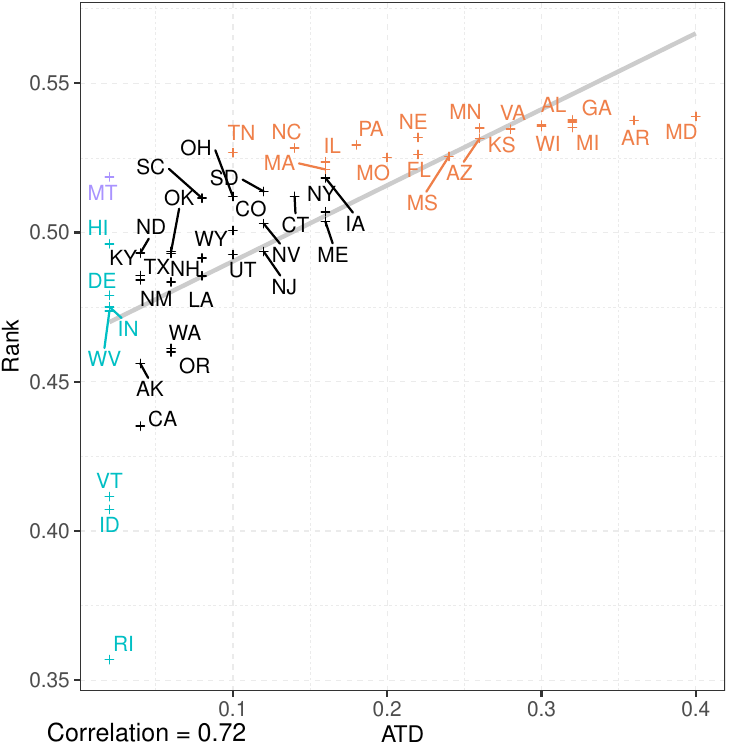}
 \end{subfigure}
 \caption{Ternary plot of compositions of electricity generation in the year 2000 for the 50 states in the U.S., where the points are colored as per their grouping into 10 groups according to quantiles of transport ranks 
 (top left); the corresponding grouping based on angular Tukey depths (ATDs, top right);
 highlighted subsets that show differences between transport ranks and ATDs (bottom left); and a scatterplot of transport ranks \eqref{eq:hrank} versus ATDs (bottom right), where the straight line shows the least squares fit (to provide perspective). 
 In the bottom two panels, a subset of states with similarly small ATDs but varying transport ranks is highlighted in orange, and another subset with similarly high transport ranks but varying ATDs is highlighted in turquoise, where the intersection of these two subsets is colored in light violet.}
 \label{fig:energy2000_ternary_rankAndATD}
\end{figure}

Networks as random objects are illustrated in another data application for New York City taxi trips; details can be found in Section~\ref{sec:taxi} of the Supplement.%\newpage

\subsection{Illustrations of the Two-Sample Test}\label{sec:two_sample_app}
\subsubsection{Human Mortality Data}
We illustrate the proposed two-sample test with the age-at-death distributions from the Human Mortality Database as described in Section~\ref{sec:mort_rank}. 
The countries we considered are Belarus, Bulgaria, Czechia, Estonia, Hungary, Latvia, Lithuania, Poland, Russia, Slovakia, and Ukraine, which are all Eastern European countries at the lowest longevity levels. 
One question of interest is whether the age-at-death distributions of these Eastern European countries changed after the dissolution of the Soviet Union. 

To this end, we compared the age-at-death distributions in 1990 and the distributions in 1993 for these countries separately for females and males,  utilizing the proposed test, as well as the energy test \citep{szek:04} and the graph based test \citep{chen:17:3}. 
The densities of these distributions are shown in Figure~\ref{fig:mort_eastEuro-dens} and the test results are summarized in Table~\ref{tab:mort_eastEuro}, where the tests are implemented and referred to in the same way as in the simulations in Section~\ref{sec:test_simu} and $p$-values less than $0.05$ are highlighted in bold. 

In Figure~\ref{fig:mort_eastEuro-dens}, it can be seen that the age-at-death densities of males in 1993 vary more across the Eastern European countries than in 1990 while the variation of those for females is more similar. 
While the proposed test does not find a significant difference between age-at-death distributions for females in 1990 and 1993, the $p$-value of the proposed test for males is below $0.05$, which provides evidence that  a systematic change occurred in the age-at-death distributions for males in these Eastern European countries between 1990 and 1993. 
In contrast, the energy test and the graph based test do not find  significant differences at the $\alpha=0.05$ significance level for either females or males. 

\begin{figure}[!hbt]
 \centering
 \includegraphics[width=.8\linewidth]{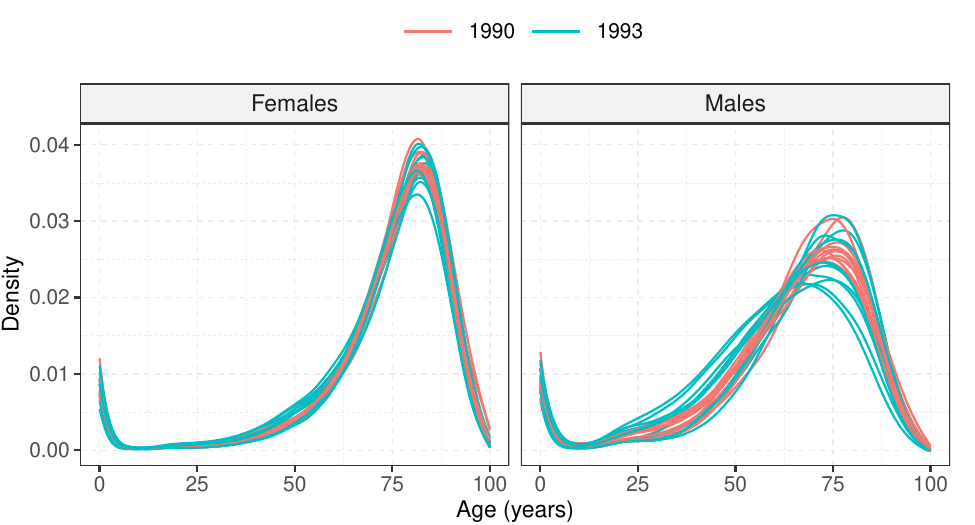}
 \caption{Age-at-death distributions for females and males in the eleven Eastern European countries in 1990 (red) and 1993 (blue), all shown as density functions.}
 \label{fig:mort_eastEuro-dens}
\end{figure}

\begin{table}[!hbt]
 \centering
 \caption{$p$-values for testing whether the age-at-death distributions in 1990 and the distributions in 1993 have the same distribution, for females and males, respectively.}
 \label{tab:mort_eastEuro}
 \begin{tabular}{l r r}
 \toprule
 Test & Females & Males \\
 \midrule
 energy & 0.326 & 0.089 \\
 graph & 0.107 & 0.055 \\
 \DP & 0.159 & \textbf{0.044} \\
 \bottomrule
 \end{tabular}
\end{table}

\subsubsection{Functional Connectivity Networks Based on fMRI Data}

Functional connectivity in neuroimaging refers to temporal association of a neurophysiological measure obtained from different regions in the brain \citep{fris:93}. Functional magnetic resonance imaging (fMRI) techniques record time courses of blood oxygenation level dependent (BOLD) signals, which are a proxy for neural activity in the brain \citep{lind:08:2}. 
Specifically, resting state fMRI (rs-fMRI) records signals when subjects are resting and not performing an explicit task. 
Functional connectivity networks can be constructed across various brain  regions of interest (ROIs)  by applying a threshold to certain measures of temporal association for each pair of ROIs that in an initial step are represented as symmetric correlation matrices. 

The rs-fMRI data in our analysis were obtained from the Alzheimer's Disease Neuroimaging Initiative (ADNI) database (\url{http://adni.loni.usc.edu}), including $400$ clinically normal (CN) subjects and $85$ mild Alzheimer’s disease dementia (AD) subjects. 
For each subject, we took only their first scan. 
Preprocessing of the BOLD signals was implemented following the standard procedures of head motion correction, slice-timing correction, co-registration, normalization, and spatial smoothing. Average signals of voxels within spheres of diameter $8$~mm centered at the seed voxels of each ROI were extracted, with linear detrending and band-pass filtering  to account for signal drift and global cerebral spinal fluid and white matter signals, including only frequencies between $0.01$ and $0.1$~Hz.  These steps were performed in MATLAB using the Statistical Parametric Mapping (SPM12, \url{http://www.fil.ion.ucl.ac.uk/spm}) and Resting-State fMRI Data Analysis Toolkit V1.8 (REST1.8, \url{http://restfmri.net/forum/?q=rest}). 

We considered the $264$ ROIs of a brain-wide graph identified by \cite{powe:11} and use temporal Pearson correlations (PCs) \citep{bisw:95} as the measure of temporal correlation between pairs of ROIs, a common approach in fMRI studies. 
Functional connectivity networks were then obtained as  adjacency matrices by imposing  an absolute threshold $0.25$ on the  $264\times 264$ matrices of temporal PCs, where  entries less than $0.25$ are replaced with zeros and  diagonal entries are set to zero. 
As distance between two functional connectivity networks we chose the Frobenius metric between the adjacency matrices.

We then employed the proposed test, 
the energy test 
and the graph based test  
to compare the functional connectivity networks of CN subjects and AD subjects. 
Prior to performing the tests, we subsampled the CN and AD subjects such that the age distributions of these two groups of subjects are similar. 
The results are presented in Table~\ref{tab:adni-CN-vs-AD}. The proposed test and the energy test have $p$-values  below $0.05$, providing evidence for a significant differences between the distributions of functional connectivity networks of CN subjects and AD subjects, while the $p$-value of the graph based test is close to $1$. 

\begin{table}[!hbt]
	\centering
	\caption{$p$-values for testing whether the functional connectivity networks of CN subjects and AD subjects have the same distribution.}
	\label{tab:adni-CN-vs-AD}
	\begin{tabular}{lr}
		\toprule
		Test & $p$-value \\
		\midrule
		energy & \textbf{0.003} \\
		graph & 0.930 \\
		\DP & \textbf{0.034} \\
		\bottomrule
	\end{tabular}
\end{table}

In a second analysis, we compared the functional connectivity networks of CN subjects for (first) scans taken  at various age groups, with their distribution across age groups  summarized in Table~\ref{tab:adni_age-distn}. 
Empirical power was obtained as the proportion of rejections at significance level $\alpha=0.05$ based on 100 Monte Carlo runs, for each of a sequence of tests.  For all tests, the first sample consisted of functional connectivity networks of 80 subjects randomly sampled from the 159 CN subjects with scans taken in the age interval  $[55,70)$.  The second samples were drawn from the remaining 320 CN subjects and  consisted of functional connectivity networks of subjects with scans taken in defined age intervals. For the first test, this age interval was  $[55,70)$; for the third test it was  $[60,75),\dots,$ for the second-to-last test it was $[75,90)$, and for the last test it was $[80,96)$.  Since it is known that these networks change with age, this sequence of tests provides an empirical  power function for detecting the age-related change. The empirical power results are presented in Figure~\ref{fig:adni_power_abs25}, indicating that  the proposed test outperforms both the energy test and the graph based test.

\begin{table}[!hbt]
 \centering
 \caption{Age distribution at first scans for the CN subjects.}
 \label{tab:adni_age-distn}
 \begin{tabular}{l r}
 \toprule
 Age interval & \# CN subjects \\
 \midrule
 $[55,60)$ & 14 \\
 $[60,65)$ & 20 \\
 $[65,70)$ & 125 \\
 $[70,75)$ & 84 \\
 $[75,80)$ & 76 \\
 $[80,85)$ & 49 \\
 $[85,90)$ & 21 \\
 $[90,95)$ & 10 \\
 $[95,96)$ & 1 \\
 \bottomrule
 \end{tabular}
\end{table}

\begin{figure}[!hbt]
 \centering
 \includegraphics[width=.5\textwidth]{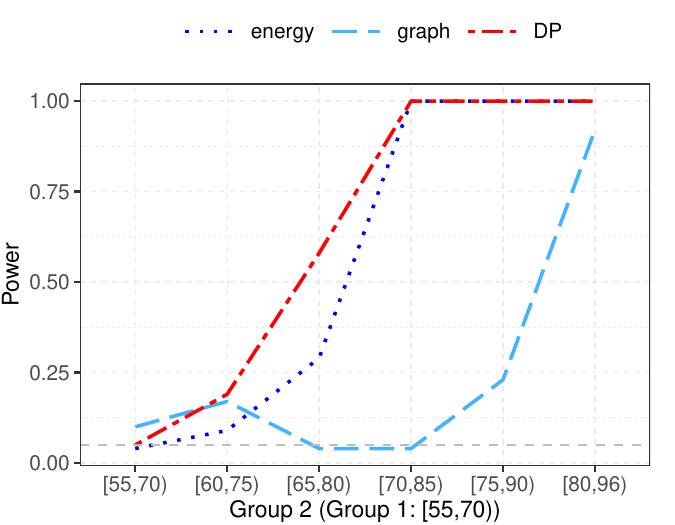}
 \caption{Power comparison for increasing age difference for two samples of functional connectivity networks of CN subjects, where the first samples consist of CN subjects with scans taken in the age interval $[55,70)$ years and the second samples consist of CN subjects with scans taken in the age intervals $[55,70), [60,75),\dots,[75,90),[80,96)$, respectively. For each Monte Carlo run, the first sample consists of 80 subjects which are randomly sampled from the 159 CN subjects with ages in $[55,70)$, and from among  the remaining 320 CN subjects, the second sample consists of the subjects with scans taken in the corresponding age interval. The dashed grey line denotes the significance level $0.05$.}
 \label{fig:adni_power_abs25}
\end{figure}

\vspace{2cm}

\section{Discussion and Outlook} \label{sec:disc}

\subsection{Metric Selection}\label{sec:ms} 

To deploy the tools of metric statistics for a given space of data objects, the choice of a metric is essential. For some data types such as Euclidean data the metric is usually preordained to be the geodesic, i.e., the usual Euclidean metric, but even in this simple special case there are still various choices;  one could consider weighted metrics that de-emphasize or emphasize specific vector components. Similarly, for data on Riemannian manifolds such as spheres,  the geodesic metric is an inherent feature of the geometry and therefore is the canonical choice. This applies also to compositional data if they are represented on the positive orthant of a unit sphere \cp{scea:11}, as described in Section~\ref{sec:energy}; an alternative selection for compositional data is the Aitchison metric \cp{aitc:86}. Both choices have specific advantages and disadvantages \cp{scea:14}, notably the Aitchison metric requires the compositional components to be positive (otherwise requiring a numerical fudge), which is not satisfied for the U.S. energy generation data that we study in Section~\ref{sec:energy}.

The metric selection problem is more complex for other data objects such as distributions, where a large number of metrics have been proposed 
and popular choices in the context of random objects  include the 2-Wasserstein  metric \eqref{eq:dwass}. The Wasserstein metric  has been shown to work well for one-dimensional distributions in distributional data analysis \cp{pana:20,pete:22} and various applied scenarios \cp{bols:03},  but poses thorny theoretical and computational problems for multivariate distributions. This 
incentivizes the  study of alternative metrics  such as the sliced Wasserstein metric \cp{kolo:16, kolo:19,mull:23:4}
and  the Fisher--Rao metric \cp{dai:22}. 
For the space of symmetric positive matrices that play an important role for applications such as fMRI and DTI brain imaging, 
one can choose among a large class of metrics, ranging from the Frobenius metric to power metrics \cp{pigo:14}, the recently proposed Cholesky metric \cp{lin:19:1}
 and metrics that reflect the geometry of eigenvectors \cp{jung:15}.

While it is clearly important, the metric selection problem in a statistical framework has been largely neglected.  If one has a class of metrics that 
is indexed by a parameter, a data-based selection criterion to find the best metric within the class may lead to consistent data-based metric selection for a specific target criterion; an example is metric selection for the family of power metrics for symmetric positive matrices \cp{mull:16:2}.
Absent a statistical framework for metric selection, a basic criterion is that the metric to some power should be of strong negative type \eqref{eq:snt}, which then implies that distance profiles characterize the underlying distribution. This property is satisfied  for the metrics that we have discussed  in examples, simulations and data analysis.  We refer to Section~3 of \cite{lyon:13} for a detailed discussion of examples and counterexamples of metric spaces that are of (strong) negative type. 
Specifically, the space of one-dimensional distributions with finite second moments endowed with the 2-Wasserstein metric in eq.~\eqref{eq:dwass}, the space of multivariate distributions with $L^2$ metric between corresponding cdfs, the space of network adjacency matrices with the Frobenius metric, and spheres with the geodesic metric are of strong negative type, while Grassmannian manifolds and cylinders with their geodesic metrics are not \citep{fera:15,vene:19}.

Other  criteria  for metric selection include feasibility and ease of implementation  (e.g., the Fisher--Rao metric can be easily deployed irrespective of the  dimensionality of distributions) and  matching of metric proximity with perceived or known similarity. A good  metric should also facilitate meaningful interpretation of the results obtained when deploying the tools of metric statistics and entail   sensible inference, so that detected differences between groups are indeed relevant. In some spaces  feature preservation when  transitioning from one data object to another along connecting geodesics that are determined by the metric is often desirable. This may include  preservation of shape features, e.g., unimodality  in the case of distributions, which is a forte of the Wasserstein metric,  or avoidance of the swelling effect in the case of symmetric positive definite matrices, as provided for example by the Cholesky metric \citep{dryd:09}. 
A metric in the space of distributions that complements the given metric in the object space is utilized in  the space of distance profiles, which correspond to one-dimensional distributions. In our approach we use the 2-Wasserstein metric and exploit its connection with optimal transports; other metrics could be explored as well.  
An important property of the proposed distance profiles is that in metric spaces of (strong) negative type they characterize the underlying probability measure $P$, which guarantees that the proposed test attains asymptotic power against any alternative $Q\ne P$, and any alternative metric would need to match this property. Another obvious extension to consider is to use weighted transports in the definition of the distance profiles, where  one could give more weight to the transported mass situated closer to $u=0$.

\subsection{Outlook and Future Research}\label{sec:of}
Distance profiles and their metric lead to a new type of MDS for random objects, providing a  representation of data objects that complements the more standard MDS representations based on the original metric in the object space, as exemplified in Section~\ref{sec:mort_rank}. The resulting visualization proved to be useful and interpretable in the examples we studied. But this is only a small start and visualization of random objects remains a widely open topic. 

As we demonstrate in the numerical experiments in Section~\ref{sec:test_simu}, 
in various scenarios 
the distance profile based test outperforms the energy test 
in terms of power. This is likely due to the fact that profile distances provide  a more fine tuned assessment of the underlying measure $P$  than 
means do, where energy tests are based on the latter. In our experiments, the profile based tests are less powerful than the energy test when the alternatives are based on mean shift expectations but more powerful for alternatives based on scale changes and where heavy tails are involved.  
Further investigation for this phenomenon as well as potential improvements and modifications of the proposed test, for example by
judicious choice of the weights in \eqref{eq:teststat_weighted},
will be left for future research.

The proposed transport ranks can also be utilized to arrive at a new measure of object depth, complementing recent developments that  extend classical notions of depth from Euclidean data to random objects \cp{dai:21:2,geen:23,chol:23}. Exploring these connections 
is a topic for future research. The properties of the proposed transport quantiles also deserve  further study in view of the 
importance and challenges of 
defining quantiles in metric spaces. Specifically, rates and especially optimal rates of convergence for transport medians, transport quantiles and other tools of metric statistics  
will require future research efforts.

While we have utilized optimal transports of distance profiles to obtain transport ranks and transport quantiles, the concept of transports can be extended to objects in uniquely geodesic spaces \cp{mull:23:6} and can then be used 
as a general modeling tool.
 Other recent developments include more sophisticated representations of random objects in reproducing kernel Hilbert spaces \cp{bhat:23}. 
These and similar developments are expected to provide valuable new tools for the nascent field of metric statistics. For random objects, essentially all relevant statistical methods for Euclidean data  need to be redesigned  with new theoretical justifications.
Examples include  deep learning that could be applied for Fr\'echet regression when the predictors are high-dimensional and   principal component analysis for random objects, where at this point there is no general theoretically supported method.   
These are just a few of the  many challenging open problems for future exploration.

\begin{funding}
This research was supported in part by NSF grants DMS-2311034 (PD), DMS-2311035 (YC), DMS-2014626 and DMS-2310450 (HGM).
\end{funding}

%%%%%%%%%%%%%%%%%%%%%%%%%%%%%%%%%%%%%%%%%%%%%%
%% Supplementary Material, including data %%
%% sets and code, should be provided in %%
%% {supplement} environment with title %%
%% and short description. It cannot be %%
%% available exclusively as external link. %%
%% All Supplementary Material must be %%
%% available to the reader on Project %%
%% Euclid with the published article. %%
%%%%%%%%%%%%%%%%%%%%%%%%%%%%%%%%%%%%%%%%%%%%%%

\begin{supplement}
% \stitle{Supplement}
% \sdescription{
The Supplement contains proofs and auxiliary results, additional simulations for the two-sample test, additional simulations for distance profiles and transport ranks for multimodal multivariate data and distributional data 
as well as additional details for human mortality data and applications of distance profiles for Manhattan Yellow Taxi networks.
\end{supplement}

\clearpage

\appendix
% \begin{supplement}
% \stitle{Title of Supplement B}
% \sdescription{Short description of Supplement B.}
% \end{supplement}

\renewcommand{\thesection}{S.\arabic{section}}
\renewcommand{\theequation}{S.\arabic{equation}}
\renewcommand{\thetable}{S.\arabic{table}}
\renewcommand{\thefigure}{S.\arabic{figure}}
\renewcommand{\theLemma}{S.\arabic{Lemma}}
\renewcommand{\theCorollary}{S.\arabic{Corollary}}

\centerline{\bf SUPPLEMENT}

\section{Proofs for Section~\ref{sec:method}}
\begin{proof}[Proof of Proposition~\ref{prop1}] Consider the {$\O$}-indexed stochastic process $\{d(\o,\obj)\}_{\o \in \O}$ and for each $k \in \mathbb{N}$ and any collection of indices $i_1,i_2,\dots,i_k$ the $\real^k$ valued random variables $(d(\o_{i_1},X),\dots, d(\o_{i_k},X))$ that define a 
	probability measure $\dveck_{i_1,i_2,\dots,i_k}$ by 
	\begin{equation*}
	\dveck_{i_1,i_2,\dots,i_k} \left(A_1 \times A_2 \times \dots \times A_k\right) \coloneqq \prob\left(d(\o_{i_1},X) \in A_1, d(\o_{i_2},X) \in A_2,\dots, d(\o_{i_k},X) \in A_k \right) 
	\end{equation*}
	for any Borel sets $A_1, A_2, \dots, A_k \subseteq \real$. 
	Note that $\dveck_{i_1,i_2,\dots,i_k}$ satisfies the following conditions:
	\begin{enumerate}[label = (\roman*)]
		\item for any permutation $\pi=(\pi(1),\dots,\pi(k))$ of $\{1,\dots,k\}$ and measurable sets $A_j \subseteq \real$, 
		\begin{equation*}
		\dveck_{\pi(i_1),\pi(i_2),\dots,\pi(i_k)} \left(A_{\pi(1)} \times A_{\pi(2)} \times \dots \times A_{\pi(k)}\right)=\dveck_{i_1,i_2,\dots,i_k} \left(A_1 \times A_2 \times\dots\times A_k\right).
		\end{equation*}
		\item for all measurable sets $A_j \subseteq \real$ and for any $m \in \mathbb{N}$
		\begin{align*}
		& \dveck_{i_1,i_2,\dots,i_k} \left(A_1 \times A_2 \times \dots \times A_k\right)
		\\  =  & \dveck_{i_1,i_2,\dots,i_k,i_{k+1},\dots,i_{k+m}} \left(A_1 \times A_2 \times \cdots \times A_k \times \real \times\dots\times \real\right).
		\end{align*}
	\end{enumerate}
	Then by the Kolmogorov extension theorem, there exists a unique probability measure $\nu$ on $\real^\O \coloneqq \{\o\mapsto g(\o): \o\in\Omega,\, g(\o)\in\real\}$, the underlying law of the stochastic process $\{d(\o,\obj)\}_{\o \in \O}$, whose finite-dimensional marginals are given by $\dveck_{i_1,i_2,\dots,i_k}$, whence the stochastic process $\{d(\o,\obj)\}_{\o \in \O}$ is well defined.
	
	For $\o \in \O$ and $r > 0$, define the open ball $O_{\o,r}=\{x \in \O: d(\o,x) < r \}$  with radius $r$ and centered at $\om$. Starting with the open balls $\{O_{\o,r}\}_{\o \in \O, r>0}$, we form an algebra  $\mathcal{B}_0$ of subsets of $\O$, which includes the empty set and open balls and is closed under complements, finite unions and finite intersections. 
	On $\mathcal{B}_0$, we define a pre-measure $P_0$, given by the marginals of the law of $\{d(\o,\obj)\}_{\o \in \O}$ such that $P_0(B)=P(B)=\prob(X^{-1}(B))$ for all $B \in \mathcal{B}_0$. When $\O$ is separable, $\mathcal{B}_0$ generates the Borel sigma algebra on $\O$ since it is an algebra containing the open balls. 
	Hence by the Hahn--Kolmogorov theorem, a version of the Carath\'eodory's extension theorem, there exists a unique extension of $P_0$ to the Borel sigma algebra of $\O$ whose restriction to $\mathcal{B}_0$ coincides with $P_0$. By uniqueness, the extension of $P_0$ is $P$. Hence the marginals of the law of $\{d(\o,\obj)\}_{\o \in \O}$ uniquely characterize the underlying Borel probability measure of $X$ on separable metric spaces. 
	
	The \deps are the one-dimensional marginals of $\{d(\o,\obj)\}_{\o \in \O}$. Suppose that for some $\theta > 0$, the space $(\O,d^\theta)$ is of strong negative type, i.e., $\int_{\O}\int_{\O} d^\theta(\o,x) \diffop P_1(\o)\diffop P_1(x)$   +  $\int_{\O}\int_{\O} d^\theta(\o,x) \diffop P_2(\o)\diffop P_2(x) -2\int_{\O}\int_{\O} d^\theta(\o,x) \diffop P_1(\o)\diffop P_2(x) \le 0$ for all probability measures $P_1,P_2$ on $\O$, where  equality holds if and only if $P_1 = P_2$ \citep{lyon:13}. Then $F_{\o}^{P_1}(u) = F_{\o}^{P_2}(u)$ for all $\o\in\O$ and $u \in \mathbb{R}$, where $F_{\o}^{P_1}(\cdot)$ and $F_{\o}^{P_2}(\cdot)$ are the \deps of $\o$ with respect to $P_1$ and $P_2$, implies that $P_1 = P_2$. This is because $F_{\o}^{P_1}(u) = F_{\o}^{P_2}(u)$ implies that $\expect_{P_1}(d^\theta(\o,X)) = \expect_{P_2} (d^\theta(\o,X))$ for all $\o \in \O$ and all $\theta > 0$. Hence it follows from \cite{lyon:13} that as long as $d^\theta$ is of strong negative type for some $\theta > 0$, the \deps with respect to $P$ uniquely characterize $P$.

\end{proof} 

\section{Proofs for Section~\ref{sec:prop}}

\begin{proof}[Proof of Theorem~\ref{prop:properties}]
	
	Part (a) follows immediately from the isometry of $h$. 
	
	For part (b), given any $X \sim P$, i.e., $P = \prob X\inv$, by the definition of a mode, $F_{\o_\oplus}(u) \geq F_{X}(u)$ and therefore $F_{\o_\oplus}\inv(u) \leq F_{X}\inv(u)$ for all $u \in \mathbb{R}$ almost surely. This implies that \newline $\int_0^1[\fobj\inv(u)-F_{\o_\oplus}\inv(u)]\diffop u \geq 0$ and therefore $\rank_{\o_\oplus} \geq \expit(0)=1/2$. Moreover observe that
	\begin{align}
	\label{mode_monotone}
	\int_0^1[\fobj\inv(u)-F_{\o_\oplus}\inv(u)]\diffop u - \int_0^1[\fobj\inv(u)-\fo\inv(u)]\diffop u = \int_0^1[\fo\inv(u)-F_{\o_\oplus}\inv(u)]\diffop u \geq 0
	\end{align}
	where the inequality in \eqref{mode_monotone} follows from the definition of $\o_\oplus$ that implies that $F_{\o_\oplus}\inv(u) \leq \fo\inv(u)$ for all $u \in \mathbb{R}$. Since $\expit(\cdot)$ is non-decreasing one has $$\expit\left\{\int_0^1[\fobj\inv(u)-F_{\o_\oplus}\inv(u)]\diffop u\right\} \geq \expit\left\{\int_0^1[\fobj\inv(u)-\fo\inv(u)]\diffop u\right\},$$ which completes the proof of part (b).
	
	For part (c), the proof is straightforward and is similar to that of part (b), 
	observing that along $\gamma(\cdot)$ for $s < t$, $F\inv_{\gamma(s)}(u) \leq F\inv_{\gamma(t)}(u)$ for all $u \in \mathbb{R}$. This leads to $\rank_{\gamma(s)} \geq \rank_{\gamma(t)}$ 
	as \newline $\expit\left\{\int_0^1[\fobj\inv(u)-F_{\gamma(s)}\inv(u)]\diffop u\right\} \geq \expit\left\{\int_0^1[\fobj\inv(u)-F_{\gamma(s)}\inv(u)]\diffop u\right\}$.
	
	For part (d), note that $\rank_{\o}^{P_1} = \rank_{\o}^{P_2}$ for all $\o\in\O$ implies 
	\begin{align*}
	\expect\{d(X',X)\} - \expect\{d(\o,X)\} = \expect\{d(Y',Y)\} - \expect\{d(\o,Y)\}
	\end{align*}
	for all $\o\in\O$, where $X,X'\sim P_1$, $Y,Y'\sim P_2$, and $(X,X',Y,Y')$ are independent. 
	Let $\delta = \expect\{d(X',X)\} - \expect\{d(Y',Y)\}$. Then $\expect\{d(\o,X)\} - \expect\{d(\o,Y)\} = -\delta$, for all $\o\in\Omega$. 
	Hence, 
	\begin{align*}
	\expect\{d(X,Y)\mid X\} - \expect\{d(X,X')\mid X\} &= \delta; \\
	\expect\{d(Y,X)\mid Y\} - \expect\{d(Y,Y')\mid Y\} &= -\delta.
	\end{align*}
	For $(\O,d)$ of strong negative type, the energy metric between $P_1$ and $P_2$ is given by  
	\begin{align*}
	& \expect\{d(X,Y)\} - \expect\{d(X',X)\} + \expect\{d(Y,X)\} - \expect\{d(Y',Y)\}\\
	=\ & \expect\left(\expect\{d(X,Y)\mid X\} - \expect\{d(X',X)\mid X\}\right) + \expect\left(\expect\{d(Y,X)\mid Y\} - \expect\{d(Y',Y)\mid Y\}\right)\\
	=\ &\delta -\delta = 0.
	\end{align*} 
\end{proof}

\section{Proofs and Auxiliary Results for Section~\ref{sec:est}}\label{sec:proof_est}

\begin{Lemma}
	\label{lma:donsker_auxiliary}
	Let $\mcF$ be a class of measurable functions such that $\mcF$ satisfies 
	\begin{equation}
	\label{eq:unif_entropy}
	\int_{0}^{\infty} \sqrt{\log N_{[]}(\eps, \mcF, L_{1}(P))} d\eps < \infty. 
	\end{equation}
	Then $\mcF$ is pre-Gaussian.
\end{Lemma}

\begin{proof}[Proof of Lemma~\ref{lma:donsker_auxiliary}]
	It is possible to use the $L_{1}(P)$ brackets of size $2^{-q}$ of $\mcF$ and disjointify them so as to obtain a partition of $\mcF$. Let $N_q$ be the number of sets in this partition. Then $N_q \leq N_{[]}(2^{-q}, \mcF, L_{1}(P))$. By the finiteness of the integral in \eqref{eq:unif_entropy} one has
	\begin{equation}
	\label{eq:summability}
	\sum_{q} 2^{-q} \sqrt{\log N_q} < \infty.
	\end{equation}
	For each $q$, denote the partitioning cover by $\{\mcF_{qi}\}_{i=1}^{N_q}$, i.e., $\mcF = \cup_{i=1}^{N_q} \mcF_{qi} $. Choose a fixed element $f_{qi} \in \mcF_{qi}$ and let
	\begin{equation*}
	\pi_qf = f_{qi} \ \text{if} \ f \in \mcF_{qi},
	\end{equation*}
	and
	\begin{equation*}
	\mcF_{q}f = \mcF_{qi} \ \text{if} \ f \in \mcF_{qi},
	\end{equation*}
	for $f\in\mcF$. 
	Now for the construction of the Gaussian semimetric $\rho$, for $f,g \in \mcF$ define $\rho(f,g)= 2^{-q_0+1}$, where $q_0$ is the first value such that $f$ and $g$ do not belong to the same partitioning set at level $q_0$. Then $\rho$ defines a semimetric on $\mcF$ such that the $\rho$-ball centered around $f$ of size $ 2^{-q+1}$, denoted by $B(f,2^{-q+1})$, is the set $\mcF_{q}f$ for every $q$, i.e., $B(f,2^{-q+1})= \mcF_{q}f$.
	Next define a Gaussian process $G$ indexed by $\mcF$ as
	\begin{equation*}
	G(f) = \sum_{q} 2^{-q} X_{q,\pi_qf},
	\end{equation*}
	{where $X_{q,\pi_qf}$ are i.i.d standard normal random variables}. For $f,g \in \mcF$ with $\rho(f,g)=2^{-q_0+1}$ observe that
	\begin{align*}
	\var\{G(f)-G(g)\} 
	&= \var\left\{ \sum_{q \geq q_0} 2^{-q} (X_{q,\pi_q f}-X_{q,\pi_q g}) \right\} 
	= 2 \sum_{q \geq q_0} 2^{-2q} = \frac{2 \rho^2(f,g)}{3}.
	\end{align*}
	Next \eqref{eq:summability} together with the fact that $N(2^{-q+1},\mcF, \rho) = N_q$ implies that the entropy for the semimetric $\rho$ satisfies the integrability condition 
	\begin{equation}
	\label{eq:rho_integral}
	\int_{0}^{\infty} \sqrt{\log N(\eps, \mcF, \rho) } d\eps < \infty.
	\end{equation}
	{By Corollary~2.2.8 in \cite{well:96}, there exists a constant $C$ such that} for any $\delta > 0$,
	\begin{equation*}
	\expect \sup_{\rho(f,g) < \delta} |G(f)-G(g)| \leq C \int_{0}^\delta \sqrt{\log N(\eps, \mcF, \rho) } d\eps.
	\end{equation*}
	{In conjunction with \eqref{eq:rho_integral}, this implies}
	\begin{equation*}
	\sup_{\rho(f,g) < \delta} |G(f)-G(g)| \overset{\prob}{\to} 0,\quad \text{as }\delta \rightarrow 0.
	\end{equation*}
	Therefore for any sequence of numbers $a_n \rightarrow 0$, there exists a sequence $b_n \rightarrow 0$ such that $\prob \left( \sup_{\rho(f,g) < b_n} |G(f)-G(g)| > a_n \right) < 2^{-n}$ for every $n$. 
	Applying the Borel--Cantelli Lemma, $|G(f)-G(g)|\leq a_n$ whenever $\rho(f,g) \leq b_n$, for all sufficiently large $n$ almost surely. Hence $G$ has a version with uniformly $\rho$-continuous sample paths, which implies that $\rho$ is a Gaussian semimetric. 
	Following Example~1.5.10 and the definition of a pre-Gaussian function class on page 89 in \cite{well:96}, we conclude that the class $\mcF$ is pre-Gaussian as it is totally bounded with respect to the semimetric $\rho$ and there exists a version of $G$ with uniformly $\rho$-continuous sample paths {$f \mapsto G(f)$}.
\end{proof}

\begin{proof}[Proof of Theorem~\ref{thm:fhat}]
	Observe that 
	\begin{equation*}
	\sqrt{n}(\hfo (t)-\fo (t)) = n^{-1/2} \sum_{i=1}^n \{y_{\o,t}(X_i)-\expect(y_{\o,t}(X_i))\}
	\end{equation*}
	and therefore we need to show that $\mcF$ is {$P$}-Donsker. Note that $\mcF$ is a measurable class of indicator functions, uniformly bounded above by the fixed envelope $G(x) \equiv 1$. 
	Theorem~5.7 in \cite{gine:84} provides necessary and sufficient conditions for establishing the Donsker property of classes of sets and equivalently the classes of indicator functions of the corresponding sets, however condition (ii) of this theorem 
	is troublesome and we therefore use a slightly stronger result (given in Corollary~6.5), which is established by bounding the metric entropy with inclusion of the classes of sets. 
	
	First we establish that $\int_{0}^\infty \sqrt{\log N_{[]}(\eps, \mcF, L_{1}(P))} d\eps < \infty$. By Lemma~\ref{lma:donsker_auxiliary} this implies that $\mcF$ is pre-Gaussian,  which verifies condition (i) of Corollary~6.5 of \cite{gine:84}. Let $\mathcal{C}=\{\mathcal{C}_{(\o,t)}: \, \o \in \O,\, t \in \real\}$,  where $\mathcal{C}_{(\o,t)}=\{x: \, d(\o,x)\leq t\}$ is the collection of sets generating the function class $\mcF$,  and define $\mathcal{C}'_{\eps} = \{(A\cap B^c) \cup (A^c\cap B): A,B\in\mathcal{C},\, P((A\cap B^c) \cup (A^c\cap B)) \le\eps\}$. Note that each set in $\mathcal{C}$ and also in $\mathcal{C}'_\eps$ belongs to the Borel sigma algebra of $\O$, which we denote by $\mathcal{B}$. 
	For $\mathcal{C}'_\eps$, define $M_0(\mathcal{C}'_{\eps},\eps, P)$ by
	\begin{align*}
	M_0(\mathcal{C}'_\eps,\eps, P) = \inf\{&r : r \geq 1,\, \exists A_1, A_2, \dots, A_r \in \mathcal{B} \ \text{such that for all} \ A \in \mathcal{C}'_\eps \\ & \text{and for some} \ j,\, A \subset A_j \ \text{and} \ P(A_j \backslash A) \leq \eps/2 \}.
	\end{align*}
	
	We establish condition (ii) of Corollary~6.5 of \cite{gine:84} by showing that \newline $n^{-1/2} \log M_0(\mathcal{C}'_{\eps/\sqrt{n}},\eps/\sqrt{n}, P) \rightarrow 0$ as $n \rightarrow \infty$ for all $\eps > 0$, which together with condition (i) implies that $\mcF$ is $P$-Donsker. 
	To achieve this, we use another quantity $N_I(\eps,\mathcal{C},P)$, the metric entropy with inclusion of $\mathcal{C}$ \citep{dudl:78}, which is defined as
	\begin{align*}
	N_I(\eps,\mathcal{C}, P) = \inf\{&r : r \geq 1,\, \exists B_1, B_2, \dots, B_r \in \mathcal{B} \ \text{such that for all} \ A \in \mathcal{C} \\ & \text{and for some} \ i,j,\, B_i \subset A \subset B_j \ \text{and} \ P(B_j \backslash B_i) < \eps \}.
	\end{align*}
	Let $A,B \in \mathcal{C}$ 
	and $A_\star, A^\star, B_\star, B^\star \in \mathcal{B}$ be such that $A_\star \subset A \subset A^\star$, $B_\star \subset B \subset B ^\star$, $P(A^\star \backslash A_\star) < \eps/4$ and $P(B^\star \backslash B_\star) < \eps/4$. Then, $(A\cap B^c) \cup (A^c\cap B) \subseteq (A^\star \cup B^\star) \backslash (A_\star \cap B_\star)$ and $[(A^\star \cup B^\star) \backslash (A_\star \cap B_\star)] \backslash [(A\cap B^c) \cup (A^c\cap B)] \subseteq (A^\star \backslash A_\star) \cup (B^\star \backslash B_\star)$. Hence, 
	\begin{align*}
	& P\left([(A^\star \cup B^\star) \backslash (A_\star \cap B_\star)] \backslash [(A\cap B^c) \cup (A^c\cap B)]\right) \\ 
	&\leq P\left(A^\star \backslash A_\star \right)+P\left(B^\star \backslash B_\star \right) 
	\leq \eps/2.
	\end{align*}
	This implies that $M_0(\mathcal{C}'_{\eps},\eps, P) \leq N_I(\eps/4,\mathcal{C},P)^4$. Therefore to obtain the result, it is enough to show that $n^{-1/2} \log N_I(\eps/\sqrt{n},\mathcal{C}, P) \rightarrow 0$ as $n \rightarrow \infty$ for all $\eps > 0$.
	
	Let $\eps'=\eps/(4\overline{\Delta})$ and $\{\o_1, \o_2, \dots, \o_{N_{\eps'}}\}$ be an $\eps'$-net of $\O$ where $ N_{\eps'} = N(\eps',\O,d)$, i.e., $\O \subset \cup_{j=1}^{N_{\eps'}} {B_{\eps'}}(\o_j) $ where ${B_{\eps'}}(\o_j)$ is a ball of radius $\eps'$ centered at $\o_j$. Let ${\{t_{i,j}\}}_{i=1}^{m_j}$, $j=1, \dots, N_{\eps'}$ be such that for each fixed $j$, $[\I\{d(\o_j,\cdot) \leq t_{i-1,j}\}, \I\{d(\o_j,\cdot) \leq t_{i,j}\}]$, $i=1, \dots, m_j,$ form $\eps/2$-brackets in $L_1(P)$ for the function class $\I\{d(\o_j,\cdot) \leq t\}$, $t \in \real$. These brackets are obtained by taking ${\{t_{i,j}\}}_{i=1}^{m_j}$ such that
	\begin{equation*}
	\lim_{t\uparrow t_{i,j}}F_{\o_j}(t) - F_{\o_j}(t_{i-1,j}) < \eps/2. 
	\end{equation*}
	The total number $m_j$ of $\eps/2$-brackets in $L_1(P)$ can be upper bounded by $4/\eps$ for any $j$. 
	
	Let $\o \in \O$ and $t \in R$. First we locate $\o_j \in \{\o_1, \o_2, \dots, \o_{N_{\eps'}}\}$ such that $d(\o,\o_j) < \eps'$. The triangle inequality implies $\I\{d(\o_j,\cdot) \leq t-\eps'\} \leq \I\{d(\o,\cdot) \leq t\} \leq \I\{d(\o_j,\cdot) \leq t+\eps' \}$. By definition, $\I\{d(\o_j,\cdot) \leq t+\eps'\} \leq \I\{d(\o_j,\cdot) \leq t_{i,j}+\eps'\} $ and $\I\{d(\o_j,\cdot) \leq t-\eps'\} \geq \I\{d(\o_j,\cdot) \leq t_{i-1,j}-\eps'\} $ whenever $t_{i-1,j} \leq t \leq t_{i,j}$. Therefore the brackets $[\I\{d(\o_j,\cdot) \leq t_{i-1,j}-\eps'\},\I\{d(\o_j,\cdot) \leq t_{i,j}+\eps']$, $i=1,\dots,m_j$, $j=1, \dots, N_{\eps'}$ cover $\mcF$. The $L_1(P)$ width of the brackets are upper bounded by $\eps$; specifically,
	\begin{align}
	& |F_{\o_j}(t_{i,j}+\eps')-F_{\o_j}(t_{i-1,j}-\eps')| \nn \\ 
	& \leq |f_{\o_j}(t_u) \eps' + F_{\o_j}(t_{i,j}) - F_{\o_j}(t_{i-1,j}) - f_{\o_j}(t_l) \eps'| \label{ineq1} \\ & \leq 2 \overline{\Delta} \eps' + \eps/2 = \eps, \label{ineq2}
	\end{align}
	for some $t_u \in [t_{i,j}, t_{i,j}+\eps']$ and $t_l \in [t_{i-1,j}-\eps', t_{i-1,j}]$. Here, \eqref{ineq1} follows from the mean value theorem and \eqref{ineq2} follows from Assumption \ref{ass:dpfctn}. 
	Therefore the $L_1(P)$ bracketing number $N_{[]}(\eps, \mcF, L_1(P))$ is bounded above by $4N_{\eps'}/\eps$ for any $0 < \eps \leq 1$. 
	For any $\eps>1$, any function in $\mcF$ can be uniformly bracketed by the bracket $[0,1]$ whose $L_1(P)$ width is 1, hence less than $\eps$. Therefore,
	\begin{equation*}
	N_{[]}(\eps, \mcF, L_1(P)) \leq 4\eps\inv N(\eps/(4\overline{\Delta}), \O, d) \ \I\{0<\eps\leq1\}+ \I\{\eps>1\}.
	\end{equation*}
	In conjunction with Assumption \ref{ass:entropy}, this implies 
	\begin{align*}
	\int_{0}^{\infty} \sqrt{\log N_{[]}(\eps, \mcF, L_{1}(P))} \diffop\eps \leq \int_{0}^{1} \sqrt{\log(4\eps\inv)} \diffop\eps + \int_{0}^{1} \sqrt{\log N(\eps/(4\overline{\Delta}), \O, d)} \diffop\eps < \infty.
	\end{align*}
	
	Finally note that for any $C_{(\o,t)} \in \mathcal{C}$, $C_{(\o_j,t_{i-1,j}-\eps')} \subset C_{(\o,t)} \subset C_{(\o_j,t_{ij}+\eps')}$ whenever $d(\o,\o_j)<\eps'$ and $t_{i-1,j} \leq t \leq t_{i,j}$ as described above. Note that
	\begin{align*}
	& P(C_{(\o_j,t_{ij}+\eps')} \backslash C_{(\o_j,t_{i-1,j}-\eps')}) 
	= F_{\o_j}(t_{i,j}+\eps')-F_{\o_j}(t_{i-1,j}-\eps') < \eps,
	\end{align*}
	by \eqref{ineq2}. Hence $N_I(\eps,\mathcal{C},P) \leq 4 N(\eps/(4\overline{\Delta}), \O, d)/\eps,$ which implies that $$N_I(\eps/\sqrt{n},\mathcal{C},P) \leq 4 \sqrt{n} N(\eps/(4\overline{\Delta} \sqrt{n}), \O, d)/\eps.$$
	Therefore
	\begin{equation}
	n^{-1/2} \log N_I(\eps/\sqrt{n},\mathcal{C}, P) \leq n^{-1/2} \log(4\sqrt{n}/\eps)+ n^{-1/2} \log N(\eps/(4\overline{\Delta} \sqrt{n}), \O, d). \label{inclusion}
	\end{equation}
	For any $\eps>0$, the right hand side in \eqref{inclusion} converges to zero as $n \rightarrow \infty$ by Assumption \ref{ass:entropy}, completing the proof.
\end{proof}

\begin{Corollary}%[Auxiliary result]
	\label{thm:qhat}
	Under Assumptions \ref{ass:entropy} and \ref{ass:dpfctn}, 
	\begin{equation*}
	{\sqrt{n}} \sup_{\o \in \O} \sup_{u \in [0,1]} \left| \hfo^{-1} (u)-\fo^{-1} (u) \right|  = O_{\prob}(1).
	\end{equation*}
\end{Corollary}

\begin{proof}[Proof of Corollary~\ref{thm:qhat}]
	Let $a=0$ and $b=\sup_{\o,\o' \in \O}d(\o,\o')$ where $b<\infty$ as $\O$ is bounded. Let $D[a,b]$ be the Banach space of all cadlag functions $z\colon [a,b]\rightarrow \real$ equipped with the uniform norm and $\mathbb{D}_2$ be the set of distribution functions of measures that concentrate on $(a,b]$. For any distribution function $\tilde{G}$ that is continuously differentiable on $[a,b]$ with a strictly positive derivative $f$, Lemma~3.9.23 in \cite{well:96} states that the inversion $G \mapsto G\inv$ as a map from $\mathbb{D}_2 \subset D[a,b] \rightarrow l^{\infty}(0,1)$ is Hadamard-differentiable at $\tilde{G}$ tangentially to $C[a,b]$, the space of continuous functions on $[a,b]$. 
	The derivative is the map $\alpha \mapsto -(\alpha/f) \circ \tilde{G}\inv$. 
	Applying the delta method \citep[Theorem~3.9.4,][]{well:96} for the asymptotic result obtained in Theorem~\ref{thm:fhat}, $\left \lbrace \sqrt{n} (\hfo\inv(u)-F_\o\inv(u)): \o \in \O, u \in (0,1) \right \rbrace$ converges weakly to a zero mean Gaussian process with covariance $$D_{(\o_1,u_1),(\o_2,u_2)}=\Cov \left(\frac{y_{\o_1,F_{\o_1}\inv(u_1)}(X)}{f_{\o_1}(F_{\o_1}\inv(u_1))},\frac{y_{\o_2,F_{\o_2}\inv(u_2)}(X)}{f_{\o_2}(F_{\o_2}\inv(u_2))}\right),$$ where the covariance function $D_{(\o_1,u_1),(\o_2,u_2)}$ is finite as we assume that $f_\o$ is strictly positive for any $\o \in \O$ as described in Assumption \ref{ass:dpfctn}. This completes the proof.
	
\end{proof}

\begin{Lemma}%[Auxiliary result]
	\label{lma:UC}
	For any $\o_1,\o_2 \in \O$ under Assumption \ref{ass:dpfctn},
	\bea
	&&\sup_{u \in [0,1]} |F^{-1}_{\o_1}(u)-F^{-1}_{\o_2}(u)| \leq d(\o_1,\o_2),\\
	&&\sup_{u \in [0,1]} |\widehat{F}^{-1}_{\o_1}(u)-\widehat{F}^{-1}_{\o_2}(u)| \leq d(\o_1,\o_2) \ \text{(almost surely)} \ \text{and}\\
	&&\sup_{u \in [0,1]} |F_{\o_1}(u)-F_{\o_2}(u)| \leq \overline{\Delta} d(\o_1,\o_2).
	\eea
\end{Lemma}

\begin{proof}[Proof of Lemma~\ref{lma:UC}]
	Let $y=F^{-1}_{\o_1}(u)$. Then under Assumption~\ref{ass:dpfctn}, $u=F_{\o_1}(y)=\prob(d(\o_1,X)\leq y)$. For any $\o_1,\o_2\in\O$ {such that $d(\o_1,\o_2)>0$}, one has
	\begin{equation} \label{lipschitz}
	\prob(d(\o_2,X)\leq y-d(\o_1,\o_2)) \leq \prob(d(\o_1,X)\leq y) \leq \prob(d(\o_2,X)\leq y+d(\o_1,\o_2)),
	\end{equation}
	which is equivalent to $F_{\o_2}(y-d(\o_1,\o_2)) \leq u \leq F_{\o_2}(y+d(\o_1,\o_2))$. This implies that $y - d(\o_1,\o_2) \leq F^{-1}_{\o_2}(u) \leq y+d(\o_1,\o_2)$ or equivalently, $F^{-1}_{\o_1}(u) - d(\o_1,\o_2) \leq F^{-1}_{\o_2}(u) \leq F^{-1}_{\o_1}(u)+d(\o_1,\o_2)$. Therefore one has that $\sup_{u \in [0,1]} |F^{-1}_{\o_1}(u) - F^{-1}_{\o_2}(u)| < d(\o_1,\o_2)$. 
	
	For the sample version let $y=\widehat{F}^{-1}_{\o_1}(u)$. One has almost surely for any $i$,
	\begin{equation*}
	\I(d(\o_2,X_i)\leq y-d(\o_1,\o_2)) \leq \I(d(\o_1,X_i)\leq y) \leq \I(d(\o_2,X_i)\leq y+d(\o_1,\o_2)),
	\end{equation*}
	which implies that $\widehat{F}_{\o_2}(y-d(\o_1,\o_2)) \leq u \leq \widehat{F}_{\o_2}(y+d(\o_1,\o_2))$ almost surely. The rest follows using previous arguments.
	
	Finally observe that using the mean value theorem, \eqref{lipschitz} implies that
	\begin{equation*}
	F_{\o_2}(y)-f_{\o_2}(\tilde{y}_1) d(\o_1,\o_2) \leq F_{\o_1}(y) \leq F_{\o_2}(y)+f_{\o_2}(\tilde{y}_2) d(\o_1,\o_2) 
	\end{equation*}
	for some $\tilde{y}_1 \in (y-d(\o_1,\o_2),y)$ and $\tilde{y}_2 \in (y,y+d(\o_1,\o_2))$. This concludes the proof as $\sup_{y} f_{\o_2}(y) \leq \overline{\Delta}$ using Assumption~\ref{ass:dpfctn}. 
\end{proof}

\begin{proof}[Proof of Theorem~\ref{thm:Rhat}]
	To establish the convergence of $\hrank_\o$, we introduce an oracle version of $\hrank_\o$  
	that is composed of independent summands and is given by
	\begin{equation}\label{eq:tildeR}
	\trank_{\o} = \expit\left[\frac{1}{n}\sum_{\subidx=1}^{n} \left\{\int_0^1\left[{F}_{X_i}\inv(u)-{F}_{\o}\inv(u)\right]\diffop u\right\}\right].
	\end{equation}
	First we establish that $\sqrt{n}\sup_{\o \in \O}|\hrank_\o-\trank_\o| = O_{\prob}(1)$ and then show that $\sqrt{n}\sup_{\o \in \O}|\trank_\o-\rank_\o| = O_\prob(1)$. The triangle inequality then leads to $\sqrt{n}\sup_{\o \in \O}|\hrank_\o-\rank_\o| = O_\prob(1)$ and hence the result.\vspace{.3cm}

	\noindent \textit{\underline{Step 1:}} $\sqrt{n}\sup_{\o \in \O}|\hrank_\o-\trank_\o| = O_\prob(1)$.
	\\Note that the $\expit(\cdot)$ function is uniformly Lipschitz with the Lipschitz constant upper bounded by $1$. Hence
	\begin{align*}
	\left| \hrank_\o - \trank_\o \right| \leq & \left| \frac{1}{n}\sum_{\subidx=1}^{n} \int_0^1\left[F_{X_i}\inv(u)-\hat{F}_{X_i}\inv(u)-F_{\o}\inv(u)+\hat{F}_{\o}\inv(u)\right]\diffop u \right| \\ \leq & 2 \sup_{\o \in \O} \sup_{u \in [0,1]} \left|\hfo\inv(u)-{F}_{\o}\inv(u)\right|.
	\end{align*}
	By Corollary~\ref{thm:qhat} the proof of Step 1 is complete.\vspace{.3cm}
	
	\noindent \textit{\underline{Step 2:}} $\sqrt{n}\sup_{\o \in \O}|\trank_\o-\rank_\o| = O_\prob(1).$
	\\
	For Step 2, note that by the Lipschitz property of the $\expit(\cdot)$ function,
	\begin{equation}
	\label{rank_upp_bound}
	\left| \trank_\o-\rank_\o \right| \leq \left| \frac{1}{n} \sum_{i=1}^n \{h_\o(X_i)-\expect(h_\o(X_i))\} \right|,
	\end{equation}
	where $h_\o(x)=\int_0^1\left[F_{x}\inv(u)-F_{\omega}\inv(u)\right]\diffop u$. 
	The process \newline  $\{\sqrt{n}\left(\frac{1}{n} \sum_{i=1}^n \{h_\o(X_i)-\expect(h_\o(X_i))\} \right): \o \in \O \}$ is an empirical process indexed by the function class $\mathcal{H}=\{h_\o: \o \in \O\}$,
	where the functions in $\mathcal{H}$ are uniformly bounded by $2\mathrm{diam}(\O)$. 
	By Theorem~2.14.2 in \cite{well:96} and Markov's inequality, it is enough to show that
	\begin{equation} \label{ent} 
	\int_0^1 \sqrt{1+\log N_{[]}(2\mathrm{diam}(\O)\eps,\mathcal{H},L_2(P))} \diffop{\eps} < \infty
	\end{equation}
	to establish that $\{\sqrt{n}\left|\frac{1}{n} \sum_{i=1}^n \{h_\o(X_i)-\expect(h_\o(X_i))\} \right|=O_{\prob}(1)$, which in conjunction with \eqref{rank_upp_bound} completes the proof of Step 2. Observe that by Lemma~\ref{lma:UC},
	\begin{align}
	\label{hbound}
	|h_\o(x)-h_{\o^\star}(x)| = \left| \int_0^1\left[F_{\o^\star}\inv(u)-F_{\omega}\inv(u)\right]\diffop u \right| \leq d(\o^\star,\o).
	\end{align}
	This leads to $h_\o(x) \in [h_{\o^\star}(x) -\mathrm{diam}(\O)\eps, h_{\o^\star}(x) + \mathrm{diam}(\O)\eps]$ whenever $\o^\star$ is such that $d(\o,\o^\star) < \mathrm{diam}(\O) \eps$. Therefore one has
	\begin{equation*}
	\log N_{[]}(2\mathrm{diam}(\O)\eps,\mathcal{H},L_2(P)) \leq \log N\left(\mathrm{diam}(\O)\eps,\O,d\right).
	\end{equation*}
	By Assumption \ref{ass:entropy}, $$\int_0^1 \sqrt{ \log N\left(\mathrm{diam}(\O)\eps,\O,d\right)} \diffop\eps < \infty,$$ which concludes the proof. 
\end{proof}

\begin{Lemma}%[Auxiliary result]
	\label{lma:LUC}
	Under Assumption \ref{ass:dpfctn}, for any $\o_1,\o_2 \in \O$ 
	\begin{equation*}
	|\rank_{\o_1}-\rank_{\o_2}| \leq d(\o_1,\o_2).
	\end{equation*}
\end{Lemma}

\begin{proof}[Proof of Lemma~\ref{lma:LUC}]
	We borrow notations from Step 2 in the proof of Theorem \ref{thm:Rhat}. Observe that by the Lipschitz property of the $\expit(\cdot)$ function and by \eqref{hbound} one has
	\begin{align*}
	\lvert \rank_{\o_1}-\rank_{\o_2} \rvert
	\leq  \left\lvert \expect \left( h_{\omega_1}(X)-h_{\omega_2}(X)  \right) \right\rvert \leq d(\o_1,\o_2),
	\end{align*}
	which completes the proof.
\end{proof}

\begin{proof}[Proof of Theorem~\ref{thm:Mhat}]
	
	First we show that 
	\begin{gather}
	\sup_{\hop \in \hM,\op \in \M} \left|\rank_{\hop}-\rank_{\op}\right|=o_\prob(1). \label{rank_hop-op}
	\end{gather}
	Applying the triangle inequality,
	\begin{equation*}
	\left|\rank_{\hop}-\rank_{\op}\right| \leq \left|\rank_{\hop}-\hrank_{\hop}\right|+ \left|\hrank_{\hop}-\hrank_{\tiop}\right|+ \left|\hrank_{\tiop}-\rank_{\op}\right|,
	\end{equation*}
	where $\tiop \in \tM \coloneqq \argsup_{\o \in \po} \hrank_\o$, i.e., $\tiop$ belongs to the set of maximizers of the empirical transport rank in $\po$.
	Observe that $$\sup_{\hop \in \hM} \left|\rank_{\hop}-\hrank_{\hop}\right| \leq \sup_{\o \in \O} \left|\rank_{\o}-\hrank_{\o}\right|=o_\prob(1)$$ by Theorem~\ref{thm:Rhat}, which also implies that for any $\tiop \in \tM$, $\sup_{\op \in \M}\left|\hrank_{\tiop}-\rank_{\op}\right|=o_\prob(1)$, since $\left|\hrank_{\tiop}-\rank_{\op}\right|$ can be upper bounded as
	\begin{align*}
	\left|\hrank_{\tiop}-\rank_{\op}\right| 
	&= \left|\sup_{\o \in \po}\hrank_{\o}-\sup_{\o \in \po} \rank_{\o}\right| 
	\leq \sup_{\o \in \po} \left|\hrank_{\o}-\rank_{\o}\right|.
	\end{align*}
	
	Hence it remains to show that $\sup_{\hop \in \hM} \left|\hrank_{\hop}-\hrank_{\tiop}\right|=o_\prob(1)$ for any $\tiop \in \tM$ almost surely. 
	We show that for any $\eta > 0$, $\prob\left(\sup_{\hop \in \hM}\left|\hrank_{\hop}-\hrank_{\tiop}\right|>\eta\right) \rightarrow 0$ {as} $n \rightarrow \infty$. 
	For any $r > 0$, 
	we define events $A_r = \{\tiop \in B^{(r)}\}$ where $B^{(r)}=\cup_{i=1}^n B_r(X_i)$ with $B_r(X_i)=\{\o \in \O: d(\o,X_i) \leq r\}$. 
	Note that $\prob(A_r^C)=(1-F_{\tiop}(r))^n$. 
	By the mean value theorem and Assumption \ref{ass:dpfctn}, $F_{\tiop}(r) > 0$ almost surely, which implies that 
	\begin{gather}\label{ArC}
	\prob(A_r^C) \rightarrow 0,\text{ as }n \rightarrow \infty.
	\end{gather}
	
	When $\tiop \in B^{(r)}$, there exists $j \equiv j(\tiop)$ such that $d(X_{j},\tiop)<r$ almost surely, which by Lemma~\ref{lma:LUC} leads to $|\rank_{X_j}-\rank_{\tiop}| \leq r$ almost surely. %for some constant $C > 0$. 
	Note that $\hrank_{\tiop} \geq \hrank_{\hop} \geq \hrank_{X_j}$ almost surely, where the first inequality holds because $\tiop$ maximizes $\hrank_\o$ in a larger set and the second inequality is implied by the fact that $\hop$ is a within-sample maximizer of $\hrank_\o$. Hence it follows that
	\begin{align*}
	\left|\hrank_{\hop}-\hrank_{\tiop}\right| 
	& \leq \left|\hrank_{\tiop}-\hrank_{X_j}\right| \\ 
	& \leq \left|\hrank_{X_j}-\rank_{X_j}\right|+\left|\rank_{X_j}-\rank_{\tiop}\right|+\left|\rank_{\tiop}-\hrank_{\tiop}\right| \\ 
	& \leq 2 \sup_{\o \in \O} \left|\hrank_\o-\rank_\o\right|+ r \quad \quad  \text{almost surely}.
	\end{align*}
	This leads to the  upper bound 
	\begin{align*}
	& \prob\left(\sup_{\hop \in \hM}\left|\hrank_{\hop}-\hrank_{\tiop}\right|>\eta \right) \\ 
	&\leq \prob\left(\sup_{\hop \in \hM}\left|\hrank_{\hop}-\hrank_{\tiop}\right|>\eta , A_r\right)+\prob(A_r^C) \\ 
	&\leq \prob\left(\sup_{\hop \in \hM}\left|\hrank_{\hop}-\hrank_{\tiop}\right|>\eta, A_r, \sup_{\o \in \O} \left|\hrank_\o-\rank_\o\right| \leq M/\sqrt{n} \right)+\prob(A_r^C)\\
	&\quad+\prob\left(\sup_{\o \in \O} \left|\hrank_\o-\rank_\o\right|>M/\sqrt{n}\right) .
	\end{align*}
	The first term can be made arbitrarily small by choosing $r$ sufficiently small and $n$ sufficiently large such that $r+2M/\sqrt{n} \leq \eta$. 
	In conjunction with \eqref{ArC} and Theorem~\ref{thm:Rhat}, \eqref{rank_hop-op} follows. 
	
	Next we show that for any $\eta > 0$, \begin{gather}\label{MhatConv}
	\prob\left(\rho_H(\hM,\M) > \eta \right) \rightarrow 0,\text{ as }n \rightarrow \infty.
	\end{gather}
	Let $\alpha(\cdot)$ be as defined in Assumption \ref{ass:separation2}. Observe that
	\begin{align*}
	& \prob\left(\rho_H(\hM,\M) > \eta \right) \\ 
	&\leq \prob\left( \text{there exists} \ \hop \in \hM \ \text{and}\ \op \in \M \ \text{such that} \ d(\hop,\op) > \eta \right) \\ 
	&\leq \prob\left( \sup_{\hop \in \hM,\,{\op\in\M}}\left|\rank_{\hop}-\rank_{\op}\right|>\alpha(\eta) \right),
	\end{align*}
	where the second inequality follows from Assumption~\ref{ass:separation2}. Noting that $\alpha(\eta)>0$, \eqref{MhatConv} follows from \eqref{rank_hop-op}, which concludes the proof. 
\end{proof}

\section{Proofs for Section~\ref{sec:test_theory}}
In what follows, for random elements $U_1$ and $U_2$ taking values in $\Omega^k$ and $\Omega^l$ respectively and a functional $g\colon \Omega^k\times \Omega^l \rightarrow \mathbb{R}$, we use notations $\E_{U_1}$ and $\prob_{U_1}$ as 
\begin{align*}
\prob_{U_1}\{ g(U_1,U_2) \le t\} &= G_1(U_2)\quad \text{and} \quad
\E_{U_1}\{ g(U_1,U_2) \} = g_1(U_2)
\end{align*}
where $G_1(u_2)=\prob\{g(U_1,u_2)\le t\}$ and $g_1(u_2) = \E\{ g(U_1,u_2) \}$ for $u_2\in\Omega$.

\subsection{Proof of Theorem \ref{thm: null_dist}} Let $\mcM=[0,M]$ where $M=\diam(\Omega)$ is the diameter of $\Omega$. Define the function 
\begin{equation*}
G^{XY}(u)= \denom[n]\add[in] \wght[Xiu] \left(\hf[XXiu]-\hf[YXiu]\right)^2 
\end{equation*}
and consider the process given by $u\mapsto \frac{nm}{n+m} G^{XY}(u)$ with $u \in \mcM$. Similarly define $G^{YX}(u)$ given by $G^{YX}(u)= \denom[m]\add[im] \wght[Yiu] \left(\hf[YYiu]-\hf[XYiu]\right)^2$ and the corresponding process $u\mapsto \frac{nm}{n+m} G^{YX}(u)$ with $u \in \mcM$. We will establish the asymptotic behaviour of these processes and then use the continuous mapping theorem to derive the limiting distribution of $T^w_{n,m}$ where
\begin{equation*}
T^w_{n,m}= \int \frac{nm}{n+m} \left( G^{XY}(u) + G^{YX}(u) \right) \df{u}.
\end{equation*}
We will focus on deriving the weak limit for $u\mapsto \frac{nm}{n+m} G^{XY}(u)$ as the arguments for $u\mapsto \frac{nm}{n+m} G^{YX}(u)$ follow analogously. 

\vspace{0.1in}

\noindent Define the quantities $\f[XXiu]=\prob_{X'} \left( d(X_i,X') \leq u \right)$ and $\f[YXiu]=\prob_{Y'} \left( d(X_i,Y') \leq u \right)$ where $X'$ is an independent copy of $X_1, \dots, X_n$ generated according to $P_1$ and $Y'$ is an independent copy of $Y_1, \dots, Y_n$ generated according to $P_2$. 
Observe that 
\begin{align*}
G^{XY}(u) & =  \denom[n]\add[in] \wght[Xiu] \left \lbrace \left( \hf[XXiu]-\f[XXiu] \right) - \left( \hf[YXiu] -\f[YXiu] \right)\right \rbrace^2 \\  +  & \frac{2}{n}\add[in] \wght[Xiu] \left \lbrace \left( \hf[XXiu]-\f[XXiu] \right) - \left( \hf[YXiu] -\f[YXiu] \right)\right \rbrace \left \lbrace \f[XXiu]-\f[YXiu] \right \rbrace \\ + &  \denom[n]\add[in] \wght[Xiu] \left \lbrace \f[XXiu] -\f[YXiu] \right \rbrace^2.
\end{align*}
Since under $H_0$: $P_1=P_2$, $\f[XXi\cdot]=\f[YXi\cdot]$ almost surely for each $i=1,\dots,n $, the weak limit of the process $u\mapsto\frac{nm}{n+m} G^{XY}(u)$ with $u \in \mcM$ is identical to the weak limit of the centered process $u\mapsto \frac{nm}{n+m} {G}_c^{XY}(u)$ with $u \in \mcM$,  where 
\begin{equation*}
{G}_c^{XY}(u)=\denom[n]\add[in] \wght[Xiu] \left \lbrace \left( \hf[XXiu]-\f[XXiu] \right) - \left( \hf[YXiu] -\f[YXiu] \right)\right \rbrace^2.
\end{equation*}
In fact one may replace the data adaptive weights $\wght[Xiu]$ with the oracle weights $\wt[Xiu]$ as defined in Assumption~\ref{ass:assumption_weights}. The process with oracle weights given by $u\mapsto \frac{nm}{n+m} \tilde{G}_c^{XY}(u)$ with $u \in \mcM$, where
\begin{equation*}
\tilde{G}_c^{XY}(u)=\denom[n]\add[in] \wt[Xiu] \left \lbrace \left( \hf[XXiu]-\f[XXiu] \right) - \left( \hf[YXiu] -\f[YXiu] \right)\right \rbrace^2,
\end{equation*}
has an identical weak limit as $u\mapsto \frac{nm}{n+m} {G}_c^{XY}(u)$ with $u \in \mcM$. 
To see this, defining $\hat{F}_x^X(u)=\denom[n]\add[jn]\ind[d(x,X_j)\leq u]$ , $\hat{F}_x^Y(u)=\denom[m]\add[km]\ind[d(x,Y_k)\leq u]$, $F_x^X(u)=\prob\left(d(x,X)\leq u\right)$ and $F_x^Y(u)=\prob\left(d(x,Y)\leq u\right)$, one has
\begin{align}
\label{eq:pop_weights}
\frac{nm}{n+m} \sup_{u \in \mcM} | {G}_c^{XY}(u) - \tilde{G}_c^{XY}(u) | \leq  \frac{nm}{n+m} \sup_{x \in \Omega, u \in \mcM} |\hat{w}_x(u)-w_x(u)| \Delta_{XY},
\end{align}
where 
\begin{equation*}
\Delta_{XY}=\left\lbrace \frac{n}{n-1} \sup_{x \in \Omega, u \in \mcM} |\hat{F}_x^X(u)-F_x^X(u)| + \frac{1}{n-1} + \sup_{x \in \Omega, u \in \mcM} |\hat{F}_x^Y(u)-F_x^Y(u)| \right \rbrace^2.
\end{equation*}
The reason is that one may write $\hf[XXiu]= \denom[n-1]\add[jn] \ind[d(X_i,X_j)\leq u]$,  which gives $\left|\hf[XXiu] -\f[XXiu]\right| \leq \frac{n}{n-1}\sup_{x \in \Omega, u \in \mcM} |\hat{F}_x^X(u)-F_x^X(u)|+\denom[n-1]$, and by definition one has $\left|\hf[YXiu] -\f[YXiu]\right| \leq \sup_{x \in \Omega, u \in \mcM} |\hat{F}_x^Y(u)-F_x^Y(u)|$. Combining these two arguments implies equation \eqref{eq:pop_weights}. By Assumption \ref{ass:assumption_weights}, 
$\sup_{x \in \Omega,u \in \mcM} |\hat{w}_x(u)-w_x(u)|=o_{\prob}(1)$. Under Assumptions \ref{ass:assumption_lipschitz} and \ref{ass:entropy}, Theorem~\ref{thm:fhat} holds. Since $m/(n+m) \leq 1$, $\frac{nm}{n+m} \sup_{x \in \Omega, u \in \mcM} |\hat{F}_x^X(u)-F_x^X(u)|^2 = O_{\prob}(1)$ and $\frac{nm}{n+m} \sup_{x \in \Omega, u \in \mcM} |\hat{F}_x^Y(u)-F_x^Y(u)|^2 = O_{\prob}(1)$, leading to  \newline  $\frac{nm}{n+m} \Delta_{XY}=O_{\prob}(1)$ and therefore $\frac{nm}{n+m} \sup_{u \in \mcM} | {G}_c^{XY}(u) - \tilde{G}_c^{XY}(u) | = o_{\prob}(1)$.

\vspace{0.1cm}

\noindent Next observe that $\hf[XXiu]-\f[XXiu] = \denom[n-1]\addneq[ji] \{\ind[d(X_i,X_j) \leq u)]-\f[XXiu] \}$ and $\hf[YXiu]-\f[YXiu] = \denom[m]\add[jm] \{\ind[d(X_i,Y_j) \leq u)]-\f[YXiu] \}$. Hence we can decompose $\tilde{G}_c^{XY}(u)$ into three components as $\tilde{G}_c^{XY}(u)= \I + \II - \III$, where 
\begin{align*}
\I =  \hspace{13cm} \\   \denom[n]\add[in] \frac{\wt[Xiu]}{(n-1)^2} \addneq[ji]\addneq[ki] \left \lbrace \ind[d(X_i,X_j) \leq u)]-\f[XXiu]\right \rbrace \left \lbrace \ind[d(X_i,X_k) \leq u)]-\f[XXiu]\right \rbrace,\\ 
\II = \hspace{13cm}  \\  \denom[n]\add[in] \frac{\wt[Xiu]}{m^2} \add[jm]\add[km] \left \lbrace \ind[d(X_i,Y_j) \leq u)]-\f[YXiu]\right \rbrace \left \lbrace \ind[d(X_i,Y_k) \leq u)]-\f[YXiu]\right \rbrace,\\ 
\III =  \hspace{13cm}  \\  \frac{2}{n}\add[in] \frac{\wt[Xiu]}{m(n-1)} \addneq[ji]\add[km] \left \lbrace \ind[d(X_i,X_j) \leq u)]-\f[XXiu]\right \rbrace \left \lbrace \ind[d(X_i,Y_k) \leq u)]-\f[YXiu]\right \rbrace.
\end{align*}

\noindent We introduce functions $U^{X,X}_{x,y,z}(u)$, $U^{X,Y}_{x,y,z}(u)$ and $U^{Y,Y}_{x,y,z}(u)$, 
\begin{equation}
\label{eq:u_function1}
U^{X,X}_{x,y,z}(u) = w_x(u) \left \lbrace \ind[d(x,y) \leq u]-F^X_x(u)\right \rbrace \left \lbrace \ind[d(x,z) \leq u]-F^X_x(u) \right \rbrace,
\end{equation}
\begin{equation}
\label{eq:u_function2}
{U^{X,Y}_{x,y,z}(u)} = w_x(u) \left \lbrace \ind[d(x,y) \leq u]-F^X_x(u)\right \rbrace \left \lbrace \ind[d(x,z) \leq u]-F^Y_x(u) \right \rbrace
\end{equation}
and
\begin{equation}
\label{eq:u_function3}
{U^{Y,Y}_{x,y,z}(u)} = w_x(u) \left \lbrace \ind[d(x,y) \leq u]-F^Y_x(u)\right \rbrace \left \lbrace \ind[d(x,z) \leq u]-F^Y_x(u) \right \rbrace.
\end{equation}

\vspace{0.2in}

\noindent \underline{\textit{Term $\I$}}: 
We decompose 
\begin{align*}
\I = & \denom[n (n-1)^2 ]\add[in]\addneq[ji]\addneq[ki] \uijk \\ = & \denom[n (n-1)^2 ]\add[in]\addneq[ji]\addneq[ki] \left \lbrace \uijk - \E_{X_i}\left( \uijk \right) \right \rbrace \\ & + \denom[n (n-1)^2 ]\add[in]\addneq[ji]\addneq[ki] \E_{X_i}\left( \uijk \right) \\ = & \IA + \IB + \IC,
\end{align*}
where 
\begin{equation*}
\IA=\denom[n (n-1)^2 ] \sum_{i=1}^{n}\sum_{j\ne i}\sum_{k\notin \{i,j\}}
\left \lbrace \uijk - \E_{X_i}\left( \uijk \right) \right \rbrace,
\end{equation*}
\begin{equation*}
\IB=\denom[n (n-1)^2 ]\add[in]\addneq[ji] \left \lbrace \uijj - \E_{X_i}\left( \uijj \right) \right \rbrace
\end{equation*}
and
\begin{equation*}
\IC=\denom[n (n-1) ]\add[jn]\add[kn] \E_{X'}\left( U^{X,X}_{X',X_j,X_k}(u)\right).
\end{equation*}
We will show that $\frac{nm}{n+m} \sup_{u \in \mcM} |\IA|$ and $\frac{nm}{n+m} \sup_{u \in \mcM} |\IB|$ are asymptotically negligible, while $\frac{nm}{n+m} \IC$ features in the limit  distribution. 

\noindent For term $\IA$, note that for each 
$u \in \mcM$, $\E_{X_i}\left \lbrace U^{X,X}_{X_i,y,z}(u) - \E_{X_i}\left(U^{X,X}_{X_i,y,z}(u)\right) \right \rbrace = 0$ for all $y,z\in\Omega$,   
$\E_{X_j}\left \lbrace U^{X,X}_{x,X_j,z}(u) - \E_{X_i}\left(U^{X,X}_{X_i,X_j,z}(u)\right) \right \rbrace = 0$ for all $x,z\in\Omega$ and \newline  $\E_{X_k}\left \lbrace U^{X,X}_{x,y,X_k}(u) - \E_{X_i}\left(U^{X,X}_{X_i,y,X_k}(u)\right) \right \rbrace = 0$ for all $y,z\in\Omega$. 
Therefore $\frac{n-1}{n-2}\IA$ is asymptotically a degenerate $U$-process of order 3 \citep{sher:94} indexed by the function class 
{$\mathcal{F}=\{f_u\colon \Omega\times\Omega\times\Omega\rightarrow \mathbb{R}\mid u \in \mcM\}$ with $f_u(x,y,z)=U^{X,X}_{x,y,z}(u)-\E_{X'}\left( U^{X,X}_{X',y,z}(u)\right)$ for $x,y,z\in\Omega$,} where $X'$ follows $P_1$ and $U^{X,X}_{x,y,z}(u)$ is as in \eqref{eq:u_function1}. Note that \newline  $\sup_{u \in \mcM} \sup_{x,y,z \in \Omega} |f_u(x,y,z)| \leq 2 C_w$ where $C_w$ is as  in Assumption \ref{ass:assumption_weights}. 
Hence the constant function $2C_w$ is an envelope function for the function class $\mathcal{F}$. For $u,v \in \mcM$ one has
\begin{align*}
& |f_u(x,y,z)-f_v(x,y,z)| \\ \leq & \lvert U^{X,X}_{x,y,z}(u) -U^{X,X}_{x,y,z}(v)\rvert + \lvert \E_{X'}(U^{X,X}_{X',y,z}(u)) - \E_{X'}(U^{X,X}_{X',y,z}(v))\rvert.
\end{align*}
Observe that there exists a constant $C_U > 0$,
\begin{align*}
& \lvert U^{X,X}_{x,y,z}(u) -U^{X,X}_{x,y,z}(v)\rvert \\ = & |w_x(u) \{\ind[d(x,y)\leq u]-F^X_x(u)\}\{\ind[d(x,z)\leq u]-F^X_x(u)\} \\ & - w_x(v) \{\ind[d(x,y)\leq v]-F^X_x(v)\}\{\ind[d(x,z)\leq v]-F^X_x(v)\}| \\ \leq & C_U \left \lbrace |w_x(u)- w_x(v)| + \left|\ind[d(x,y)\leq u]-\ind[d(x,y)\leq v]\right| \right. \\ & + \left. \left|\ind[d(x,z)\leq u]-\ind[d(x,z)\leq v]\right| + \left|F^X_x(u)-F^X_x(v)\right| \right \rbrace.
\end{align*}
Let $0 < \epsilon < 1$ and 
{$C_0=\max\{L_X,L_Y,L_w\}$} 
{with $L_w$ as in Assumption~\ref{ass:assumption_weights} and $L_X$ and $L_Y$ as in Assumption~\ref{ass:assumption_lipschitz}.} Provided that $|u-v|<\epsilon$, by Assumption~\ref{ass:assumption_weights}, $\|w_{X'}(u)- w_{X'}(v)\|_{L_2(P_1)} \leq L_w \epsilon \leq C_0 \epsilon $. Using the mean value theorem and Assumption \ref{ass:assumption_lipschitz}, one has for any $x \in \Omega$, $|F_x^X(u)-F_x^X(v)|\leq L_X \epsilon$, which implies that $\|F_{X'}^X(u)- F_{X'}^X(v)\|_{L_2(P_1)} \leq L_X \epsilon \leq C_0 \epsilon$. Observe that by the property of indicator functions,
\begin{align*}
& \left \lVert \ind[d(X,X')\leq u]-\ind[d(X,X')\leq v] \right \rVert_{L_2(P_1 \times P_1)} \\ = & \sqrt{\left|\prob(d(X,X')\leq u)-\prob(d(X,X')\leq v) \right|} \\ = & \sqrt{\left|\E_{X'}\left\lbrace \prob(d(X',X)\leq u\mid X') \right\rbrace-\E_{X'}\left\lbrace \prob(d(X',X)\leq v\mid X')\right\rbrace \right|} \\ \leq & \sqrt{\E_{X'} \left\lvert \prob(d(X',X)\leq u\mid X')-\prob(d(X',X)\leq v\mid X') \right \rvert},
\end{align*}
where the last step follows using Jensen's inequality. By the independence of $X$ and $X'$ and Assumption \ref{ass:assumption_lipschitz}, one has $\left|\prob(d(X',X)\leq u\mid X'=x)-\prob(d(X',X)\leq v\mid X'=x) \right| = \left|F_{x}^X(u)-F_{x}^X(v) \right| \leq L_X \epsilon$, which gives $ \left|\prob(d(X',X)\leq u\mid X')-\prob(d(X',X)\leq v\mid X') \right| \leq C_0 \epsilon$ almost surely. Hence 
\begin{equation}
\label{eq:indicators}
\left \lVert \ind[d(X,X')\leq u]-\ind[d(X,X')\leq v] \right \rVert_{L_2(P_1 \times P_1)} \leq \sqrt{C_0 \epsilon}. 
\end{equation}
Therefore for any given $u \in \mcM$ by picking $v \in \mcM$ such that $|u-v|<\epsilon$ 
it turns out that 
$\lVert U^{X,X}_{X',X'',X'''}(u) -U^{X,X}_{X',X'',X'''}(v)\rVert_{L_2(P_1 \times P_1 \times P_1)} < C \sqrt{\epsilon}$, where $C = 2 C_U \{C_0+\sqrt{C_0}\}$, since $\epsilon < \sqrt{\epsilon}$ for $0 < \epsilon \leq 1$.
By Jensen's inequality, 
\begin{align*}
& \lVert \E_{X'}(U^{X,X}_{X',X'',X'''}(u)) -\E_{X'}(U^{X,X}_{X',X'',X'''}(v)) \rVert_{L_2(P_1 \times P_1 )} \\ \leq & \lVert U^{X,X}_{X',X',X''}(u) -U^{X,X}_{X',X'',X'''}(v) \rVert_{L_2(P_1 \times P_1 \times P_1)} 
\\ < & C \sqrt{\epsilon}.
\end{align*}
This implies that 
$\lVert f_u(X',X'',X''') -f_v(X',X'',X''') \rVert_{L_2(P_1 \times P_1 \times P_1)} < 2 C \sqrt{\epsilon}$. Therefore, the packing number $D(\epsilon,\mathcal{F},L_2(P_1 \times P_1 \times P_1))$, i.e., the maximum number of $\epsilon$-separate elements in $\mathcal{F}$ endowed with $L_2(P_1 \times P_1 \times P_1)$ metric, is upper bounded by the packing number $D(\tfrac{\epsilon^2}{4C^2}, \mcM, d_E)$, i.e., the maximum number of $\tfrac{\epsilon^2}{4C^2}$-separate elements in $\mcM$ endowed with the Euclidean metric $d_E$. Note that $D(\tfrac{\epsilon^2}{4C^2}, \mcM, d_E) = O\left(\epsilon^{-2}\right)$, whence %} 
$\mathcal{F}$ is a Euclidean class as per Definition 3 of \cite{sher:94}. By Corollary 4 in \cite{sher:94}, $n^{3/2} \frac{n-1}{n-2} \sup_{u \in \mcM} |\IA| = O_{\prob}(1)$, which implies that $\frac{nm}{n+m}\sup_{u \in \mcM} |\IA| = o_{\prob}(1)$ since $\frac{m}{n+m} \leq 1$. 

\vspace{0.1in}
\noindent For term $\IB$, observe that $(n-1)\IB$ is a non-degenerate $U$-process of order 2 indexed by the function class $\mathcal{G}=\{g_u\colon \Omega\times\Omega\rightarrow\mathbb{R}\mid u \in \mcM\}$ with $g_u(x,y)=f_u(x,y,y)$ for $x,y\in\Omega$, where $f_u$ is as in the analysis of term $\IA$. It turns out that whenever $|u-v|<\epsilon$ using Assumptions \ref{ass:assumption_weights} and \ref{ass:assumption_lipschitz},
\begin{align*}
|g_u(x,y)-g_v(x,y)| \leq G_{v}(x,y),
\end{align*}
where 
\begin{align*}
G_{v}(x,y)= & C_{G} \left \lbrace |\ind[d(x,y)\leq v+\epsilon]-\ind[d(x,y)\leq v-\epsilon]| + 2\epsilon \right \rbrace
\end{align*}
for some constant $C_G>0$. By the property of indicator functions and Assumption \ref{ass:assumption_lipschitz} and using similar arguments as in $\IA$,
\begin{align*}
& \lVert \ind[d(X,X')\leq v+\epsilon]-\ind[d(X,X')\leq v-\epsilon] \rVert_{L_1(P_1 \times P_1)} \\ = & \left|\prob(d(X,X')\leq v+\epsilon)-\prob(d(X,X')\leq v-\epsilon) \right| \\ \leq & 2C_0 \epsilon,
\end{align*}
which leads to 
\begin{align*}
\|G_{v}(X,X') \|_{L_1(P_1 \times P_1)} < 2 C_G (C_0+1) \epsilon.
\end{align*}
Let $v_1, \dots, v_{N_\epsilon}$ denote an $\frac{\epsilon}{2 C_G (C_0+1)}$-net of $\mcM$ where $N_\epsilon$ is the covering number of $\mcM$ with $\frac{\epsilon}{2 C_G (C_0+1)}$-radius balls. Then $\mathcal{I} = \{v_1,\dots,v_{N_\epsilon}\}$ 
forms a collection such that for any $u \in \mcM$, there exists $v_k \in \mathcal{I}$ such that $|g_u(x,y)-g_{v_k}(x,y)| \leq G_{v_k}(x,y)$ and $ \|G_{v_k}(X,X') \|_{L_1(P_1 \times P_1)} < \epsilon$. As per the definition of brackets in page 1512 of \cite{arco:93}, the bracketing number $N^{(1)}_{[]}(\epsilon, \mathcal{G}, P_1 \times P_1) $ is upper bounded by the cardinality of $\mathcal{I}$, i.e., by $N_\epsilon$, and therefore finite. Since $\E_{X_i,X_j}(\uijj-\E_{X_i}(\uijj))=0$, by Corollary 3.5 in \cite{arco:93}, $(n-1)\sup_{u \in \mcM} |\IB| =o_{\prob}(1)$. This together with the fact that $\frac{m}{m+n} \leq 1$ 
implies that $\frac{nm}{n+m}\sup_{u \in \mcM} |\IB| =o_{\prob}(1)$ as $m,n \rightarrow \infty$. Therefore the term $\IB$ can be ignored in the asymptotic distribution.

\vspace{0.1in}

\noindent Observe that $\IC$ is asymptotically equivalent to $\frac{1}{n^2}\add[jn]\add[kn] \E_{X'} \left( U^{X,X}_{X',X_j,X_k}(u) \right)$. We postpone the discussion of the term $\IC$ to the end of the proof.

\vspace{0.1in}

\noindent \underline{\textit{Term $\II$}}: %\pd{Check consistency of the new notations and the constants.} 
Observe that 
\begin{align*}
\II = \IIA + \IIB,
\end{align*}
where 
\begin{equation*}
\IIA = \denom[m^2] \add[jm]\add[km] \zeta^X(Y_j,Y_k,u) 
\end{equation*}
and
\begin{equation*}
\IIB = \denom[m^2] \add[jm]\add[km] \E_{X'}\left(U^{Y,Y}_{X',Y_j,Y_k}(u) \right),
\end{equation*}
with $\zeta^X(Y_j,Y_k,u)= \denom[n] \add[in] \left \lbrace \uxyy -\E_{X_i} \left(\uxyy\right) \right\rbrace$. Consider the process $\left\{\zeta^X(\omega_1,\omega_2,u)= \denom[n] \add[in] \left \lbrace \uxww -\E_{X_i} \left(\uxww \right) \right\rbrace: (\omega_1,\omega_2)\in\Omega\times\Omega,u \in\mcM\right\}$. 
First we establish that $\sup_{\omega_1,\omega_2 \in \Omega} \sup_{u \in \mcM} \left| \zeta^X(\omega_1,\omega_2,u) \right| = o_{\prob}(1)$. To this end, consider  the function class $\mathcal{L}=\{l_{\omega_1,\omega_2,u}(x): \omega_1,\omega_2 \in \Omega, u \in \mcM \}$ where $l_{\omega_1,\omega_2,u}(x)=U^{Y,Y}_{x,\omega_1,\omega_2}(u)$. 
With some constant $C_L > 0$ one has
\begin{align*}
& |l_{\omega_1,\omega_2,u}(x)-l_{\omega'_1,\omega'_2,u'}(x)| \\ \leq & C_L \left \lbrace |w_x(u)- w_x(u')| + |\ind[d(x,\omega_1)\leq u]-\ind[d(x,\omega_1)\leq u']| \right. \\ & + \left. |\ind[d(x,\omega_2)\leq u]-\ind[d(x,\omega_2)\leq u']| + |F^Y_x(u)-F^Y_x(u')| \right.  \\ & + \left. |\ind[d(x,\omega_1)\leq u']-\ind[d(x,\omega'_1)\leq u']| + |\ind[d(x,\omega_2)\leq u']-\ind[d(x,\omega_2)\leq u']| \right\rbrace.
\end{align*}
\noindent 
Let $|u-u'|<\epsilon$, $d(\omega_1,\omega_1')<\epsilon$ and $d(\omega_2,\omega_2')<\epsilon$. By Assumption \ref{ass:assumption_weights}, $\left|w_x(u)-w_x(u')\right| \leq L_w |u-u'| < C_0 \epsilon$ and by Assumption \ref{ass:assumption_lipschitz} and the mean value theorem, $\left|F^Y_x(u)-F^Y_x(u')\right| \leq L_Y |u-u'| < C_0 \epsilon$. Observe that if $d(x,\omega_1) \leq u'$ and $|u-u'|<\epsilon$, then $|\ind[d(x,\omega_1)\leq u]-\ind[d(x,\omega_1)\leq u']| \leq |\ind[d(x,\omega_1)\leq u'+\epsilon]-\ind[d(x,\omega_1)\leq u'-\epsilon]]|$ and similarly one has $|\ind[d(x,\omega_2)\leq u]-\ind[d(x,\omega_2)\leq u']| \leq |\ind[d(x,\omega_2)\leq u'+\epsilon]-\ind[d(x,\omega_2)\leq u'-\epsilon]]|$. On top of this if $d(\omega_1,\omega_1')<\epsilon$ and $d(\omega_2,\omega_2')<\epsilon$ then
$|\ind[d(x,\omega_1)\leq u]-\ind[d(x,\omega_1)\leq u']| \leq |\ind[d(x,\omega_1')\leq u'+2\epsilon]-\ind[d(x,\omega_1')\leq u'-2\epsilon]]|$ and $|\ind[d(x,\omega_2)\leq u]-\ind[d(x,\omega_2)\leq u']| \leq |\ind[d(x,\omega_2')\leq u'+2\epsilon]-\ind[d(x,\omega_2')\leq u'-2\epsilon]]|$. Moreover one has  $|\ind[d(x,\omega_1)\leq u']-\ind[d(x,\omega'_1)\leq u']|  \leq |\ind[d(x,\omega'_1)\leq u'+\epsilon]-\ind[d(x,\omega'_1)\leq u'-\epsilon]| $ and $|\ind[d(x,\omega_2)\leq u']-\ind[d(x,\omega'_2)\leq u']|  \leq |\ind[d(x,\omega'_2)\leq u'+\epsilon]-\ind[d(x,\omega'_2)\leq u'-\epsilon]| $ whenever $d(\omega_1,\omega_1')<\epsilon$ and $d(\omega_2,\omega_2')<\epsilon$ by the  triangle inequality. Hence if $|u-u'|<\epsilon$, $d(\omega_1,\omega_1')<\epsilon$ and $d(\omega_2,\omega_2')<\epsilon$, one has 
\begin{align*}
& |l_{\omega_1,\omega_2,u}(x)-l_{\omega'_1,\omega'_2,u'}(x)| \leq b_{\omega_1',\omega_2',u'}(x) 
\end{align*}
with
\begin{align*}
b_{\omega_1',\omega_2',u'}(x) =\ & C'_L \left \lbrace \epsilon + |\ind[d(x,\omega_1')\leq u'+2\epsilon]-\ind[d(x,\omega'_1)\leq u'-2\epsilon]| \right. \\ & + \left. |\ind[d(x,\omega_2')\leq u'+2\epsilon]-\ind[d(x,\omega'_2)\leq u'-2\epsilon]|\right \rbrace
\end{align*}
for some constant $C_L > 0$. 
Therefore if $|u-u'|<\epsilon$, $d(\omega_1,\omega_1')<\epsilon$ and $d(\omega_2,\omega_2')<\epsilon$, one has
\begin{align*}
& \|b_{\omega_1',\omega_2',u'}(X)\|_{L_1(P_1)} \\ \leq & C'_L \left \lbrace \epsilon + |F^X_{\omega_1'}(u'+2\epsilon)-F^X_{\omega'_1}(u'-2\epsilon)| + |F^X_{\omega_2'}(u'+2\epsilon)-F^X_{\omega'_2}(u'-2\epsilon)| \right \rbrace,
\end{align*}
which in conjunction with Assumption \ref{ass:assumption_lipschitz} implies that $\|b_{\omega_1',\omega_2',u'}(X)\|_{L_1(P_1)} \leq C^{''}_L \epsilon$ for some $C^{''}_L > 0$. Let $u_1, u_2, \dots, u_{N'_\epsilon}$ be an $\frac{\epsilon}{2C^{''}_L}$-net of $\mcM$ and $\omega_1, \dots, \omega_{M'_\epsilon}$ be an $\frac{\epsilon}{2C^{''}_L}$-net of $\Omega$, where $N'_{\epsilon}$ and $M'_{\epsilon}$ are the covering numbers of $\mcM$ and $\Omega$ with balls of radius $\frac{\epsilon}{2C^{''}_L}$,  respectively. Then the brackets $\left\lbrace l_{\omega_j,\omega_k,u_l} \pm b_{\omega_j,\omega_k,u_l} \right\rbrace_{j,k\in[M'_{\epsilon}],\,l\in[N'_{\epsilon}]}$ cover $\mathcal{L}$ and the $L_1(P_1)$ width of each bracket, that is $2 \|b_{\omega_1',\omega_2',u'}(X)\|_{L_1(P_1)}$, is upper bounded by $\epsilon$. Hence for any $\epsilon>0$ the $L_1(P_1)$-bracketing entropy of $\mathcal{L}$, $N_{[]}(\epsilon,\mathcal{L},L_1(P_1))$, is upper bounded by $N'_\epsilon {M'_\epsilon}^2$ and therefore $N_{[]}(\epsilon,\mathcal{L},L_1(P_1)) < \infty$ for any $\epsilon > 0$. 
By Theorem 2.4.1 in \cite{well:96}, $\mathcal{L}$ is a Glivenko--Cantelli class of functions, which implies that 
\begin{equation}
\label{eq:glivenko_cantelli}
\sup_{\omega_1,\omega_2 \in \Omega} \sup_{u \in \mcM} \left|\zeta^X(\omega_1,\omega_2,u)\right| = o_{\prob}(1).
\end{equation}

Next observe that
\begin{align*}
\denom[m^2] \add[jm]\add[km] \zeta^X(Y_j,Y_k,u) = \denom[m^2] \add[jm] \zeta^X(Y_j,Y_j,u) + \denom[m^2] \addneq[jk] \zeta^X(Y_j,Y_k,u)
\end{align*}
and that $\sup_{u \in \mcM} \left| \frac{nm}{n+m} \denom[m^2] \add[jm] \zeta^X(Y_j,Y_j,u) \right| \le \sup_{\omega_1,\omega_2 \in \Omega} \sup_{u \in \mcM} \left|\zeta^X(\omega_1,\omega_2,u)\right|$. 
In conjunction with \eqref{eq:glivenko_cantelli}, $\sup_{u \in \mcM} \left| \frac{nm}{n+m} \denom[m^2] \add[jm] \zeta^X(Y_j,Y_j,u) \right|=o_{\prob}(1)$. 
Consider the event $A_t=\left \lbrace \sup_{\omega_1,\omega_2 \in \Omega} \sup_{u \in \mcM} \zeta^X(\omega_1,\omega_2,u) > t \right \rbrace$ for $t>0$ and observe that for any $\eta > 0$ and $t>0$,
\begin{align*}
& \prob \left( \sup_{u \in \mcM} \left| \frac{nm}{n+m} \denom[m^2] \addneq[jk] \zeta^X(Y_j,Y_k,u) \right| > \eta \right) \\ \leq & \prob \left( \sup_{u \in \mcM} \left| \frac{nm}{n+m} \denom[m^2] \addneq[jk] \zeta^X(Y_j,Y_k,u) \right| > \eta, A_t^C \right) + \prob(A_t).
\end{align*}
By \eqref{eq:glivenko_cantelli}, $\prob(A_t) \rightarrow 0$ as $n \rightarrow \infty$ for any $t>0$. 
Observe that 
$\E_{Y_j}\left\{ \zeta^X(Y_j,\cdot,\cdot)\right\} = \E\left\{ \zeta^X(Y_j,\cdot,\cdot)\mid X_1,\dots,X_n\right\} \equiv 0$ and $\E_{Y_k}\left\{ \zeta^X(\cdot,Y_k,\cdot)\right\} = \E\left\{ \zeta^X(\cdot,Y_k,\cdot)\mid X_1,\dots,X_n\right\} \equiv 0$ 
almost surely, 
which implies that $ \denom[m(m-1)] \addneq[jk] \zeta^X(Y_j,Y_k,u) $ is a degenerate $U$-process of order 2 indexed by the function class $\mathcal{V}=\{\zeta^X(\cdot,\cdot,u): u \in \mcM\}$ conditional on $X_1, \dots, X_n$. In addition, note that
\begin{align*}
& \left|\zeta^X(y,z,u)-\zeta^X(y,z,v)\right| \\ = & \left| \denom[n] \add[in] \left \lbrace U^{Y,Y}_{X_i,y,z}(u) -\E_{X_i} \left(U^{Y,Y}_{X_i,y,z}(u)\right) \right\rbrace-\denom[n] \add[in] \left \lbrace U^{Y,Y}_{X_i,y,z}(v) -\E_{X_i} \left(U^{Y,Y}_{X_i,y,z}(v)\right) \right\rbrace \right| \\ \leq & \denom[n]\add[in] \left\lbrace \left|U^{Y,Y}_{X_i,y,z}(u) -U^{Y,Y}_{X_i,y,z}(v) \right|+\left| \E_{X_i} \left(U^{Y,Y}_{X_i,y,z}(u)\right)-\E_{X_i} \left(U^{Y,Y}_{X_i,y,z}(v)\right)\right| \right\rbrace.
\end{align*}
Using the bound on the functions $|U^{Y,Y}_{x,y,z}(u)-U^{Y,Y}_{x,y,z}(v)|$ similar to what was derived for the term $\IA$ one has 
\begin{align*}
& \denom[n]\add[in] \left|U^{Y,Y}_{X_i,y,z}(u) -U^{Y,Y}_{X_i,y,z}(v) \right| \\ \leq & C_U \denom[n]\add[in] \left \lbrace |w_{X_i}(u)- w_{X_i}(v)| + |\ind[d(X_i,y)\leq u]-\ind[d(X_i,y)\leq v]| \right. \\ & + \left. |\ind[d(X_i,z)\leq u]-\ind[d(X_i,z)\leq v]| + |F^Y_{X_i}(u)-F^Y_{X_i}(v)| \right \rbrace.
\end{align*}
Let $u,v$ be such that $|u-v|< \epsilon$. By Assumption \ref{ass:assumption_weights} and Assumption \ref{ass:assumption_lipschitz} together with the mean value theorem, conditional on $X_1, \dots, X_n$, $\lVert w_{X_i}(u)- w_{X_i}(v) \rVert_{L_2(P_2 \times P_2)} \leq C_0 \epsilon$ and $\lVert F^X_{X_i}(u)-F^X_{X_i}(v) \rVert_{L_2(P_2 \times P_2)} \leq C_0 \epsilon$. Note that by the independence of $\{X_i\}_{i=1}^{n}\cup\{Y_i\}_{i=1}^{m}$, the conditional distribution of $Y_1,\dots,Y_m$ on $X_1,\dots X_n$ is same as the unconditional distribution. Therefore conditional on $X_1, \dots, X_n$, $\lVert \left|\ind[d(X_i,Y)\leq u]-\ind[d(X_i,Y)\leq v]\right| \rVert_{L_2(P_2)}=\sqrt{|F^Y_{X_i}(u)-F^Y_{X_i}(v)|} \leq \sqrt{C_0 \epsilon}$. Hence conditional on $X_1,\dots X_n$ whenever $|u-v| < \epsilon$ and $0 < \epsilon \leq 1$,
\begin{align*}
\left\| U^{Y,Y}_{X_i,Y',Y''}(u) -U^{Y,Y}_{X_i,Y',Y''}(v) \right\|_{L_2(P_2 \times P_2)} \leq 2 C_U (C_0+\sqrt{C_0}) \sqrt{\epsilon},
\end{align*}
which implies that 
\begin{equation*}
\left\|\zeta^X(Y',Y'',u)-\zeta^X(Y',Y'',v) \right\|_{L_2(P_2 \times P_2)} \leq 4 C_U (C_0+\sqrt{C_0}) \sqrt{\epsilon}. 
\end{equation*}
Conditional on $X_1, \dots, X_n$, 
the packing number $D(\epsilon,\mathcal{V},L_2(P_2 \times P_2))$, i.e., the maximum number of $\epsilon$-separate elements in $\mathcal{V}$ is upper bounded by the maximum number of $\tfrac{\epsilon^2}{16C_U^2(C_0+\sqrt{C_0})^2}$-separate elements in $\mcM$, which equals  
a constant times $\epsilon^{-2}$. This implies that $\mathcal{V}$ is a Euclidean class as per Definition 3 of \cite{sher:94}. 
On $A^C_t$ an envelope function for $\mathcal{V}$ is the constant function $t$. By following the proof of Corollary 4 in \cite{sher:94} and using Markov's inequality, we find  that for some $\alpha \in(0,1)$,
\begin{align*}
& \prob \left( \sup_{u \in \mcM} \left| \frac{nm}{n+m} \denom[m^2] \addneq[jk] \zeta^X(Y_j,Y_k,u) \right| > \eta \mid X_1, \dots, X_n, A^C_t \right) \\ 
=\ & \prob \left( \sup_{u \in \mcM} \left| m \denom[m^2] \addneq[jk] \zeta^X(Y_j,Y_k,u) \right| > \frac{n+m}{n}\eta \mid X_1, \dots, X_n, A^C_t\right) \\ 
\leq\ & \frac{\E\left(\sup_{u \in \mcM} \left| m \denom[m^2] \addneq[jk] \zeta^X(Y_j,Y_k,u) \right| \mid X_1,\dots,X_n, A^C_t\right)}{\frac{n+m}{n}\eta} 
\\ \leq\ & \mathrm{const.} \frac{t^\alpha}{\frac{n+m}{n}\eta}, 
\end{align*}
almost surely. Hence,
\begin{align*}
& \prob \left( \sup_{u \in \mcM} \left| \frac{nm}{n+m} \denom[m^2] \addneq[jk] \zeta^X(Y_j,Y_k,u) \right| > \eta, A_t^C \right) \\ 
=\ & \E_{X_1,\dots,X_n} \left\lbrace \prob \left( \sup_{u \in \mcM} \left| \frac{nm}{n+m} \denom[m^2] \addneq[jk] \zeta^X(Y_j,Y_k,u) \right| > \eta \mid X_1, \dots, X_n, A^C_t \right) \right\rbrace \prob(A_t^C) \\ \leq\ & \mathrm{const.} \frac{t^\alpha}{\frac{n+m}{n}\eta} \prob(A_t^C).
\end{align*}
Let $t \rightarrow 0$ and $n \rightarrow \infty$ so that $\frac{nm}{n+m} \IIA = o_{\prob}(1)$. We postpone the discussion of the term $\IIB$ until the end of the proof.

\vspace{0.1in}

\noindent \underline{\textit{Term $\III$}}: 
Note that 
\begin{equation*}
\III = 2 \{ \IIIA + \IIIB \}
\end{equation*}
where $\IIIA= \denom[nm(n-1)]\add[in]\addneq[ji]\add[km] \left \lbrace \uxxy- \E_{X_i}\left( \uxxy \right) \right \rbrace$ and \newline  $\IIIB= \denom[nm] \add[jn] \add[km] \E_{X'}\left( U^{X,Y}_{X',X_j,Y_k}(u) \right)$. The asymptotic limit of $\IIIB$ will be investigated in the next step. 

\noindent Observe that $\IIIA =\denom[n(n-1)] \addneq[ji] \zeta^Y_{X_i,X_j}(u)$, where
\begin{equation*}
\zeta^Y_{x,y}(u)=\denom[m] \add[km] \left\lbrace U^{X,Y}_{x,y,Y_k}(u)-\E_{X'} \left( U^{X,Y}_{X',y,Y_k}(u)\right) \right \rbrace.
\end{equation*}
For the next few threads of the argument, we will condition on $Y_1, \dots, Y_m$. Note that 
$\E_{X_i}\left\lbrace \zeta^Y_{X_i,\cdot}(\cdot) \right \rbrace \equiv 0$ and $ \E_{X_j}\left\lbrace \zeta^Y_{\cdot,X_j}(\cdot) \right \rbrace \equiv 0$ 
almost surely. 
Hence conditioning on $Y_1, \dots, Y_m$, $\IIIA$ is a degenerate $U$-process of order 2 indexed by the function class $\mathcal{H}=\{h_{u}\colon\Omega\times\Omega\rightarrow\mathbb{R}\mid u \in \mcM \}$, where $h_{u}(x,y)= \zeta^Y_{x,y}(u)$ for $x,y\in\Omega$ and $u\in\mcM$. One has
\begin{align*}
& \lVert h_{u}(X',X'')-h_{v}(X',X'') \rVert_{L_2(P_1 \times P_1)} \\ \leq\ & \sup_{\omega \in \Omega} \|U^{X,Y}_{X',X'',\omega}(u) -U^{X,Y}_{X',X'',\omega}(v)\|_{L_2(P_1 \times P_1)} \\ & + \sup_{\omega \in \Omega} \|\E_{X'}(U^{X,Y}_{X',X'',\omega}(u)) - \E_{X'}(U^{X,Y}_{X',X'',\omega}(v))\|_{L_2(P_1 \times P_1)},
\end{align*}
where $X''\sim P_1$ and is independent of $X',X_1,\dots,X_n$.
Similarly to the arguments in the analysis of the term $\IA$,  one has for some constant $C'_U > 0$ that 
\begin{align*}
& \left \lvert U^{X,Y}_{x,y,\omega}(u) -U^{X,Y}_{x,y,\omega}(v) \right \rvert \\ \leq\ & C'_U \left \lbrace |w_x(u)- w_x(v)| + |\ind[d(x,y)\leq u]-\ind[d(x,y)\leq v]| \right. \\ & + \left. |\ind[d(x,\omega)\leq u]-\ind[d(x,\omega)\leq v]| + |F^X_x(u)-F^X_x(v)| + |F^Y_x(u)-F^Y_x(v)| \right \rbrace. 
\end{align*}
By Assumption \ref{ass:assumption_weights} and Assumption \ref{ass:assumption_lipschitz} together with the mean value theorem and $C_0$ as defined in term $\I$, for any $x \in \Omega$ one has $\|w_{X}(u) - w_{X}(v)\|_{L_2(P_1)} \leq C_0 |u-v| $, $\sup_{x\in\Omega}|F^X_{x}(u)-F^X_{x}(v)| \leq C_0 |u-v|$ and $\sup_{x\in\Omega}|F^Y_{x}(u)-F^Y_{x}(v)| \leq C_0 |u-v|$. In analogy to the term $\IA$, $\E_{X,X'}\left(\ind[d(X,X')\leq u]-\ind[d(X,X')\leq v]\right)^2 \leq C_0|u-v|$. Finally by Assumption \ref{ass:assumption_lipschitz} and the mean value theorem, one also has that \newline  $\sup_{\omega\in\Omega} \E_{X}\left( \ind[d(X,\omega)\leq u]-\ind[d(X,\omega)\leq v] \right)^2 = \sup_{\omega\in\Omega} |F^X_{\omega}(u)-F^X_{\omega}(v)| \leq C_0 |u-v|$. Whenever $0 < \epsilon \leq 1$ and $|u-v|<\epsilon^2$, it holds that for some constant
\begin{equation*}
\sup_{\omega\in\Omega}\lVert U^{X,Y}_{X',X'',\omega}(u) -U^{X,Y}_{X',X'',\omega}(v) \rVert_{L_2(P_1 \times P_1)} \leq C'_U \left \lbrace 3C_0 + 2 \sqrt{C_0} \right \rbrace \epsilon,
\end{equation*}
and by Jensen's inequality,
\begin{equation*}
\lVert h_{u}(X',X'')-h_{v}(X',X'') \rVert_{L_2(P_1 \times P_1)} \leq 2 C'_U \left \lbrace 3C_0 + 2 \sqrt{C_0} \right \rbrace \epsilon. 
\end{equation*}
Therefore, conditioning on $Y_1, \dots, Y_m$, the packing number $D(\epsilon, \mathcal{H},L_2(P_1 \times P_1))$, i.e., the maximum number of $\epsilon$-separate elements in $\mathcal{H}$ is upper bounded by the packing number $D(\tfrac{\epsilon^2}{4(C'_U)^2\left \lbrace 3C_0 + 2 \sqrt{C_0} \right \rbrace^2},\mcM,d_E)$, i.e., the maximum number of $\tfrac{\epsilon^2}{4(C'_U)^2\left \lbrace 3C_0 + 2 \sqrt{C_0} \right \rbrace^2}$-separate elements in $\mcM$, equal to a constant times $\epsilon^{-2}$, %} 
which implies that the class $\mathcal{H}$ is Euclidean when conditioning on $Y_1, \dots, Y_m$. 
Note that an envelope function for $\mathcal{H}$ is given by the function $F = 2C_w \sup_{x \in \Omega} \sup_{u \in \mcM} |\hat{F}^Y_x(u)-F^Y_x(u)|$. 
Following the proof of Corollary 4 in \cite{sher:94} for some $0 < \alpha < 1$,
\begin{equation*}
\E \left( {\sup_{u\in\mcM}} m | \IIIA | \mid Y_1, \dots, Y_m \right) \leq {\mathrm{const.}} \left\lbrace 2C_w \sup_{x \in \Omega} \sup_{u \in \mcM} |\hat{F}^Y_x(u)-F^Y_x(u)| \right \rbrace^{\alpha}.
\end{equation*}
Hence $$\E\left( {\sup_{u\in\mcM}} \frac{nm}{n+m} | \IIIA | \right) \leq {\mathrm{const.}} \{2C_w\}^\alpha \E\left\lbrace \sup_{x \in \Omega} \sup_{u \in \mcM} |\hat{F}^Y_x(u)-F^Y_x(u)| \right \rbrace^{\alpha} = o(1)$$ by Theorem~\ref{thm:fhat} and the continuous mapping theorem in conjunction with the fact that $\sup_{m}\left\lbrace \sup_{x \in \Omega} \sup_{u \in \mcM} |\hat{F}^Y_x(u)-F^Y_x(u)| \right \rbrace^{\alpha}\le 1$, which implies the uniform integrability of $\left\lbrace \sup_{x \in \Omega} \sup_{u \in \mcM} |\hat{F}^Y_x(u)-F^Y_x(u)| \right \rbrace^{\alpha}$. By Markov's inequality, ${\sup_{u\in\mcM}}\frac{nm}{n+m} | \IIIA | = o_{\prob}(1)$.
\vspace{0.1in}

\noindent \underline{\textit{Asymptotic Limit}}: 
To arrive at the asympotic distribution, we  combine the terms $\IC$, $\IIB$ and $2\IIIB$ as derived in the previous steps. Observe that $\IC + \IIB -2\IIIB$ is asymptotically equivalent to 
\begin{align*}
& T_1(u)  = \\ & \E_{X'} \left\lbrace w_{X'}(u) \left( \denom[n]\add[jn]\ind[d(X',X_j)\leq u]-F_{X'}^X(u) - \denom[m]\add[km]\ind[d(X',Y_k)\leq u] + F_{X'}^Y(u) \right)^2 \right\rbrace
\end{align*}
with $X' \sim P_1$. By repeating the same arguments for $ \left \lbrace \frac{nm}{n+m} G^{YX}(u): u \in \mcM \right \rbrace$, one has a similar term $T_2(u)$ featuring in the asymptotic distribution, 
\begin{align*}
& T_2(u)  = \\ & \E_{Y'} \left\lbrace w_{Y'}(u) \left( \denom[n]\add[jn]\ind[d(Y',X_j)\leq u]-F_{Y'}^X(u) - \denom[m]\add[km]\ind[d(Y',Y_k)\leq u] + F_{Y'}^Y(u) \right)^2 \right\rbrace,
\end{align*}
where $Y' \sim P_2$. Define 
\begin{align*}
& W^{n,m}_x(u)  = \\ & \sqrt{\frac{nm w_x(u)}{n+m}} \left\lbrace \denom[n]\add[jn]\ind[d(x,X_j)\leq u]-F_{x}^X(u) - \denom[m]\add[km]\ind[d(x,Y_k)\leq u]+F_{x}^Y(u) \right \rbrace.
\end{align*}
Following similar arguments to the proof of Theorem~\ref{thm:fhat}, for any fixed $x \in \Omega$, one has that
$u\mapsto\sqrt{w_x(u)\frac{nm}{n+m}} \left\lbrace \denom[n]\add[jn]\ind[d(x,X_j)\leq u]-F_{x}^X(u) \right \rbrace $ converges weakly to a zero mean Gaussian process with covariance given by $C^X_x(u,v)$, where 
\begin{equation*}
C^X_x(u,v) = (1-c) \ \sqrt{w_x(u)w_x(v)} \ \mathrm{Cov}(\ind[d(x,X)\leq u],\ind[d(x,X)\leq v]), 
\end{equation*}
since $\frac{n}{n+m} \rightarrow c$ as $n,m \rightarrow \infty$ by Assumption \ref{ass:samp_ratio}. 
Similarly for any fixed $x \in \Omega$,  the term 
$u\mapsto\sqrt{w_x(u)\frac{nm}{n+m}} \left\lbrace \denom[m]\add[km]\ind[d(x,Y_k)\leq u]-F_{x}^Y(u) \right \rbrace $ converges weakly to a zero mean Gaussian process with  covariance  $C^Y_x(u,v)$, where 
\begin{equation*}
C^Y_x(u,v) = c \ \sqrt{w_x(u)w_x(v)} \ \mathrm{Cov}(\ind[d(x,Y)\leq u],\ind[d(x,Y)\leq v]).
\end{equation*}
Since $\left\lbrace \denom[n]\add[jn]\ind[d(x,X_j)\leq u]-F_{x}^X(u) \right \rbrace $ and $ \left\lbrace \denom[m]\add[km]\ind[d(x,Y_k)\leq u]-F_{x}^Y(u) \right \rbrace $ are independent,
by combining the above arguments,  $W^{n,m}_x(\cdot)$ converges weakly to a zero mean Gaussian process $G_x(\cdot)$ with covariance given by $C_x(u,v)= C^X_x(u,v)+ C^Y_x(u,v)$. Since $C_x(u,v)$ is symmetric, non-negative definite and continuous, by Mercer's Theorem 
\begin{equation*}
C_x(u,v) = \add[j\infty] \lambda_j^x \phi_j^x(u) \phi_j^x(v),
\end{equation*}
where $\lambda_1^x \geq \lambda_2^x \geq\dots$ are the eigenvalues of $C_x(u,v)$ and $\phi_1^x(\cdot), \phi_2^x(\cdot), \dots$ are the corresponding eigenfunctions,  which form an orthonormal basis of $L^2(\mcM)$. By the Karhunen--Lo\`eve expansion,
\begin{equation*}
G_x(u) = \add[j\infty] Z_j \sqrt{\lambda^x_j} \phi^x_j(u),
\end{equation*}
where $Z_1, Z_2, \dots$ are independent  $N(0,1)$ random variables. Hence
\begin{equation*}
(G_x(u))^2 = \add[j\infty] \add[k\infty]Z_j Z_k \sqrt{\lambda^x_j}\sqrt{\lambda^x_k} \phi^x_j(u)\phi^x_k(u).
\end{equation*}
Observe that 
\begin{equation*}
\frac{nm}{n+m}  T_1(u)  = \E_{X'}\left \lbrace \left( W^{n,m}_{X'}(u) \right)^2 \right \rbrace 
\end{equation*}
and that $(x,u)\mapsto W^{n,m}_{x}(u)$ converges weakly to $(x,u)\mapsto G_{x}(u)$ by Theorem~\ref{thm:fhat}.
In conjunction with the continuous mapping theorem, $\E_{X'}\left \lbrace \int \left( W^{n,m}_{X'}(u) \right)^2 \df u \right \rbrace$ converges weakly to $\E_{X'}\left \lbrace \int (G_{X'}(u))^2 \df u \right \rbrace$, which is given by
\begin{align*}
\E_{X'}\left[\int (G_{X'}(u))^2 \df u\right] = & \add[j\infty] \add[k\infty]Z_j Z_k \E_{X'} \left\lbrace \sqrt{\lambda^{X'}_j}\sqrt{\lambda^{X'}_k} \int \phi^{X'}_j(u)\phi^{X'}_k(u) \df u \right\rbrace \\ = & \add[j\infty] Z_j^2 \E_{X'} \left\lbrace \lambda_j^{X'} \right \rbrace
\end{align*}
as by orthogonality $\int \phi^{X'}_j(u)\phi^{X'}_k(u) \df u = \delta_{jk}$, where $\delta_{jk}=1$ when $j = k$ and $\delta_{jk}=0$ otherwise. Similarly $\int\frac{nm}{n+m}  T_2(u)\diffop u = \E_{Y'}\left \lbrace \int \left( W^{n,m}_{Y'}(u) \right)^2 \diffop u \right \rbrace$ converges weakly to the law of $\add[j\infty] Z_j^2 \E_{Y'} \left\lbrace \lambda_j^{Y'} \right \rbrace$. 

Under $H_0: P_1 = P_2$, $\sup_{u \in \mcM}|T_1(u)-T_2(u)|=0$ almost surely. Therefore the final asymptotic distribution of $T^w_{n,m}$ in \eqref{eq:teststat_weighted} is derived using the limiting distribution of $\frac{2nm}{n+m} T_1(u)$, which concludes the proof. 

\subsection{Proof of Theorem \ref{thm: power}} 
Observe that
\begin{align*}
T^w_{n,m} = & \int \frac{nm}{n+m} \left( G^{XY}(u) + G^{YX}(u) \right) \df{u} \\ =& \frac{nm}{n+m} \left( U_1 + U_2 + U_3 \right),
\end{align*}
where 
\begin{align*}
& U_1 =  \int \left[ \denom[n]\add[in] \wght[Xiu] \left \lbrace \left( \hf[XXiu]-\f[XXiu] \right) - \left( \hf[YXiu] -\f[YXiu] \right)\right \rbrace^2 \right. \\ + & \left. \denom[m]\add[km] \wght[Yku] \left \lbrace \left( \hf[YYku]-\f[YYku] \right) - \left( \hf[XYku] -\f[XYku] \right)\right \rbrace^2 \right] \df{u},
\end{align*}
\begin{align*}
& U_2 = \hspace{14cm}  \\ & \int \left[ \frac{2}{n} \add[in] \wght[Xiu] \left \lbrace \left( \hf[XXiu]-\f[XXiu] \right) - \left( \hf[YXiu] -\f[YXiu] \right)\right \rbrace \left \lbrace \f[XXiu]-\f[YXiu] \right \rbrace \right. \\ + & \left. \frac{2}{m} \add[km] \wght[Yku] \left \lbrace \left( \hf[YYku]-\f[YYku] \right) - \left( \hf[XYku] -\f[XYku] \right)\right \rbrace \left \lbrace \f[YYku]-\f[XYku] \right \rbrace \right] \df{u}
\end{align*}
and
\begin{align*}
U_3 = & \int \left[ \denom[n]\add[in] \wght[Xiu] \left \lbrace \f[XXiu] -\f[YXiu] \right \rbrace^2 + \denom[m]\add[km] \wght[Yku] \left \lbrace \f[YYku] -\f[XYku] \right \rbrace^2 \right] \df{u}.
\end{align*}

\noindent Observe that under $H_{n,m}$, generally $P_1 = P_2$ does not hold.  By following the proof of Theorem \ref{thm: null_dist}, $\frac{nm}{n+m} U_{1}$ has the same asymptotic rate as $\frac{nm}{n+m} \int \left\lbrace T_{1}(u)+T_2(u) \right\rbrace \df{u}$ with $T_1(u)$ and $T_2(u)$ as in the proof of Theorem \ref{thm: null_dist}. Since the arguments in the derivation of this rate do not require the assumption that $P_1=P_2$, it turns out that $\frac{nm}{n+m} \left| \int \left\lbrace T_{1}(u)+T_2(u) \right\rbrace \df{u}\right|=O_{\prob}(1)$, which implies that $\frac{nm}{n+m}U_1 =O_{\prob}(1)$.

\noindent By  Assumption \ref{ass:assumption_weights}, $\sup_{x \in \Omega} \sup_{u \in \mcM} |\hat{w}_x(u)-w_x(u)|=o_{\prob}(1)$, which leads to $$\sup_{x \in \Omega} \sup_{u \in \mcM} \left \lvert \hat{w}_x(u) \left\lbrace F^X_x(u)-F^Y_x(u)\right \rbrace^2 - {w}_x(u) \left\lbrace F^X_x(u)-F^Y_x(u)\right \rbrace^2 \right \rvert=o_{\prob}(1).$$ With the continuous mapping theorem this implies  $\lvert U_3 - \tilde{U}_3 \rvert =o_{\prob}(1)$, where 
\begin{align*}
\tilde{U}_3 = & \int \left[ \denom[n]\add[in] w_{X_i}(u) \left \lbrace \f[XXiu] -\f[YXiu] \right \rbrace^2 + \denom[m]\add[km] w_{Y_k}(u) \left \lbrace \f[YYku] -\f[XYku] \right \rbrace^2 \right] \df{u}.
\end{align*}
Observe that $\E\left(\tilde{U}_3\right) = D^w_{XY}$ and therefore by the weak law of large numbers $\lvert \tilde{U}_3 -D^w_{XY}\rvert = o_{\prob}(1)$. Combining the above arguments yields $\lvert U_3 - D^w_{XY}\rvert = o_{\prob}(1)$. 
Defining events $A^{(1)}_{nm} = \left\{ \frac{D^w_{XY}}{2}\le U_3\le 2D^w_{XY} \right\}$, one has $\prob(A^{(1)}_{nm})\rightarrow 1$ as $n,m\rightarrow\infty$.

\noindent By Assumption \ref{ass:assumption_weights}, $\sup_{x \in \Omega} \sup_{u \in \mcM} \lvert \hat{w}_x(u) \rvert \leq 2 C_w$ almost surely. Define events $$A_{nm,\eta}=\left \lbrace \sqrt{\frac{nm}{n+m}} \sup_{x \in \Omega} \sup_{u \in \mcM} \left( \lvert \hat{F}^X_x(u)-F^X_x(u)\rvert + \lvert \hat{F}^Y_x(u)-F^Y_x(u)\rvert \right) \leq \eta\right \rbrace.$$
Using Theorem~\ref{thm:fhat} and Assumption \ref{ass:samp_ratio} leads to  $\sqrt{\frac{nm}{n+m}} \sup_{x \in \Omega} \sup_{u \in \mcM} \lvert \hat{F}^X_x(u)-F^X_x(u)\rvert = O_{\prob}(1)$ and $\sqrt{\frac{nm}{n+m}} \sup_{x \in \Omega} \sup_{u \in \mcM} \lvert \hat{F}^Y_x(u)-F^Y_x(u)\rvert = O_{\prob}(1)$ whence 
$\prob \left( A_{nm,\eta}\right) \rightarrow 1$ as $m,n \rightarrow \infty$ and $\eta\rightarrow\infty$.  
Next 
defining events $A^{(2)}_{nm} = \left\{ \sup_{x\in\Omega} \sup_{u\in\mcM}\sqrt{\hat{w}_x(u)} \le 2\sqrt{C_w} \right\}$, observe that under Assumption~\ref{ass:assumption_weights}, $\prob(A^{(2)}_{nm}) \rightarrow 1$ as $n,m\rightarrow\infty$. Furthermore, note that when $A^{(2)}_{nm}$ and $A_{nm,\eta}$ are true, 
\begin{align*}
\lvert U_2 \rvert \leq & \Delta_{nm}\left\{4 \sqrt{C_w} \eta\right\}/\sqrt{\frac{nm}{n+m}}, 
\end{align*}
where 
\begin{align*}
\Delta_{nm} = & \int \left[ \frac{1}{n} \add[in] \sqrt{\wght[Xiu]} \left \lvert \f[XXiu]-\f[YXiu] \right \rvert + \frac{1}{m} \add[km] \sqrt{\wght[Yku]} \left \lvert \f[XYku]-\f[YYku] \right \rvert \right] \df{u}.
\end{align*}
Moreover
\begin{align*}
& \left( \frac{1}{n} \add[in] \sqrt{\wght[Xiu]} \left \lvert \f[XXiu]-\f[YXiu] \right \rvert + \frac{1}{m} \add[km] \sqrt{\wght[Yku]} \left \lvert \f[XYku]-\f[YYku] \right \rvert \right)^2 \\ & \leq 2 \left[ \left( \frac{1}{n} \add[in] \sqrt{\wght[Xiu]} \left \lvert \f[XXiu]-\f[YXiu] \right \rvert \right)^2 + \left(\frac{1}{m} \add[km] \sqrt{\wght[Yku]} \left \lvert \f[XYku]-\f[YYku] \right \rvert \right)^2\right] \\ & \leq 2 \left[ \denom[n]\add[in] \wght[Xiu] \left \lbrace \f[XXiu] -\f[YXiu] \right \rbrace^2 + \denom[m]\add[km] \wght[Yku] \left \lbrace \f[YYku] -\f[XYku] \right \rbrace^2\right],
\end{align*}
where the last step follows by applying the Cauchy-Schwartz inequality. This implies that
\begin{align}
\label{eq:U2_to_U3}
\int \left( \frac{1}{n} \add[in] \sqrt{\wght[Xiu]} \left \lvert \f[XXiu]-\f[YXiu] \right \rvert + \frac{1}{m} \add[km] \sqrt{\wght[Yku]} \left \lvert \f[XYku]-\f[YYku] \right \rvert \right)^2 \df{u} \leq 2\ U_3. 
\end{align}
With $\mathcal{M}=\mathrm{diam}(\O)$ and by using the Cauchy-Schwartz inequality in conjunction with the inequality in \eqref{eq:U2_to_U3} one has
\begin{align*}
\Delta_{nm} \leq \sqrt{2\mathcal{M}U_3}. 
\end{align*}
Hence when $A_{nm,\eta}$ and $A^{(1)}_{nm}$ are true, one has $\Delta_{nm} \leq \sqrt{4\mathcal{M} D^w_{XY}}$. This leads to 
\begin{align*}
\frac{nm}{n+m} \lvert U_2 \rvert \leq 8\eta\sqrt{C_w \mathcal{M}\frac{nm}{n+m}D^w_{XY}}
\end{align*}
when $A_{nm,\eta}$, $A^{(1)}_{nm}$ and $A^{(2)}_{nm}$ are true. When $D^w_{XY}=a_{nm}$ and $\frac{nm}{n+m}a_{nm} \rightarrow \infty$ as $m,n \rightarrow \infty$, $\frac{nm}{n+m} a_{nm}$ dominates $\sqrt{\frac{nm}{n+m} a_{nm}}$. Therefore under $H_{n,m}$, when $A_{nm,\eta}$, $A^{(1)}_{nm}$ and $A^{(2)}_{nm}$ are true, one has for sufficiently large $n,m$ that $\frac{nm}{n+m}\lvert U_2 \rvert \leq \frac{nm}{n+m} \frac{a_{nm}}{4}$, which leads to
\begin{align}
\label{eq:pow}
\frac{nm}{n+m} \left( U_2 + U_3 \right) \geq \frac{nm}{n+m} \frac{a_{nm}}{4}.
\end{align}
Finally, observe that 
\begin{align*}
\beta^w_{n,m} = & \prob_{H_{n,m}} \left( T^w_{n,m} > q_\alpha \right) \\ \geq & \prob_{H_{n,m}} \left( T^w_{n,m} > q_\alpha,\, A_{nm,\eta} \cap A^{(1)}_{nm}\cap A^{(2)}_{nm}\right) \\ \geq & \prob_{H_{n,m}} \left( \frac{nm}{n+m} U_1 > q_\alpha - \frac{nm}{n+m} \frac{a_{nm}}{4},\, A_{nm,\eta} \cap A^{(1)}_{nm}\cap A^{(2)}_{nm}\right),
\end{align*}
where the last step follows using \eqref{eq:pow}. Since $\frac{nm}{n+m} U_1 = O_{\prob}(1)$ and $\frac{nm}{n+m}a_{nm} \rightarrow \infty$ as $m,n \rightarrow \infty$, and the fact that $\prob(A_{nm,\eta} \cap A^{(1)}_{nm}\cap A^{(2)}_{nm})\rightarrow 1$ as $n,m \rightarrow \infty$ and $\eta\rightarrow\infty$, 
one has $\beta^w_{n,m} \rightarrow 1$ as $n,m \rightarrow \infty$. 

\vspace{0.2 em}
\noindent
\subsection{Proof of Theorem \ref{thm:perm}} For any permutation $p \in \Pi_{n,m}$, denote the pooled sample permuted according to $p$ by 
$V_p= \left(V_{p(1)}, \dots, V_{p(n+m)}\right)$  and  let $X_p= \left(V_{p(1)}, \dots, V_{p(n)}\right)$ and $Y_p=\left(V_{p(n+1)}, \dots, V_{p(n+m)}\right)$ denote the splitting of the permuted pooled observations into two samples of sizes $n$ and $m$ respectively. Let $\tilde{\Gamma}^w_{n,m}(\cdot)$ represent the randomization distribution of $T^w_{n,m}$ as defined by
\begin{equation*}
\tilde{\Gamma}^w_{n,m}(t) = \frac{1}{(n+m)!} \sum_{\pi \in \Pi_{n,m}} \ind[T^w_{n,m}(X_{\pi},Y_{\pi}) \leq t].
\end{equation*}
Let $\pi$ and $\pi'$ be independent and uniformly distributed on $\Pi_{n,m}$ and independent of $V_1, V_2, \dots, V_{n+m}$. Suppose that under the probability measure on $\Omega^{n+m}$ induced by $V_1,V_2,\dots, V_{n+m}$, \begin{equation}
\label{eq:condition_perm}
\left( T^{w}_{n,m}(X_{\pi},Y_{\pi}), T^{w}_{n,m}(X_{\pi'},Y_{\pi'}) \right) \xrightarrow{D} (T,T'),
\end{equation}
as $n,m \rightarrow \infty$ such that $T$ and $T'$ are independent, each with a common c.d.f. $\Gamma(\cdot)$. Then by Theorem 5.1 in \cite{chun:13}, for all continuity points $t$ of $\Gamma(\cdot)$, $| \tilde{\Gamma}^w_{n,m}(t) - \Gamma(t)|=o_{\prob}(1)$ as $n,m \rightarrow \infty$. Observe that conditional on the data $V_1, V_2, \dots, V_{n+m}$, $\sup_t \lvert \hat{\Gamma}^w_{n,m}(t) - \tilde{\Gamma}^w_{n,m}(t) \rvert = o_{\prob}(1) $ for $\hat{\Gamma}^w_{n,m}(t)$ as in equation \eqref{eq:perm_cdf} as $K \rightarrow \infty$ using Theorem 11.2.18 in \cite{lehm:86}, regardless of the data distribution. Specifically, $\prob(\sup_t \lvert \hat{\Gamma}^w_{n,m}(t) - \tilde{\Gamma}^w_{n,m}(t) \rvert > C \mid V_1,\dots,V_{n+m}) \le 2\exp(-2KC^2)$ a.s.
Hence unconditionally as $n,m, K \rightarrow \infty$, one has for all continuity points $t$ of $\Gamma(\cdot)$, $| \hat{\Gamma}^w_{n,m}(t) - \Gamma(t)|=o_{\prob}(1)$ provided the condition in \eqref{eq:condition_perm} holds. In what follows we first outline the steps of the proof and then establish the details in each step. 

\begin{enumerate}[label=\Roman*.]
	\item Under $H_0$, we will show that condition \eqref{eq:condition_perm} holds with $\Gamma(\cdot) \equiv \Gamma_L(\cdot)$. Then for all $t$ such that  $\Gamma_L(t)$ is continuous, $| \hat{\Gamma}^w_{n,m}(t) - \Gamma_L(t)|=o_{\prob}(1)$ as $n,m,K \rightarrow \infty$. If in addition $\Gamma_L(t)$ is continuous and strictly increasing at $q_\alpha$ as per the statement of the theorem, then $\lvert \hat{q}_\alpha-q_\alpha \rvert = o_{\prob}(1)$ as $n,m,K \rightarrow \infty$.
	\item Let $\bar{V}=(\bar{V}_1, \dots , \bar{V}_{n+m})$ be an i.i.d. sample from $\bar{P}$ generated through the coupling construction described in Section~5.3 of \cite{chun:13} by using the pooled observations $V=(V_1, \dots, V_{n+m})$ such that for some permutation $\pi_0$, $V_i$ and $\bar{V}_{\pi_0(i)}$ agree for all $i=1, \dots, n+m$ except for $D$ entries, where $\E(D/(n+m)) \leq (n+m)^{-1/2}$. For a permutation $p \in \Pi_{n,m}$, let $\bar{V}_p= \left(\bar{V}_{p(1)}, \dots, \bar{V}_{p(n+m)}\right)$, and $\bar{X}_p= \left(\bar{V}_{p(1)}, \dots, \bar{V}_{p(n)}\right)$ and $\bar{Y}_p=\left(\bar{V}_{p(n+1)}, \dots, \bar{V}_{p(n+m)}\right)$ be the division of $\bar{V}$ into two samples of sizes $n$ and $m$ respectively. 
	Due to condition \eqref{eq:condition_perm} to be established in Step $\mathrm{I}$ under $H_0$, 
	condition \eqref{eq:condition_perm} also holds for $\left(T^w_{n,m}(\bar{X}_\pi,\bar{Y}_\pi), T^w_{n,m}(\bar{X}_{\pi'},\bar{Y}_{\pi'}) \right)$ under the data distribution of $\bar{V}$ with $\Gamma(\cdot)\equiv \bar{\Gamma}_L(\cdot)$, where $\pi$ and $\pi'$ are independent permutations uniformly distributed on $\Pi_{n,m}$ and independent of $V$ and $\bar{V}$. This establishes condition (5.9) in \cite{chun:13}. 
	Next we will show that $\lvert T^w_{n,m}(\bar{X}_{\pi \pi_0},\bar{Y}_{\pi\pi_0})-T^w_{n,m}(X_{\pi},Y_{\pi})\rvert = o_{\prob}(1)$, where $\pi\pi_0$ is the composition of the permutations $\pi$ and $\pi_0$ with $\pi_0$ applied first. 
	This will establish condition (5.10) of \cite{chun:13}. By Lemma~5.1 in \cite{chun:13} for all continuity points $t$ of $\bar{\Gamma}_L(t)$, $\lvert \tilde{\Gamma}^w_{n,m}(t)-\bar{\Gamma}_L(t)\rvert=o_{\prob}(1)$, which in conjunction with previous arguments implies that $\lvert \hat{\Gamma}^w_{n,m}(t)-\bar{\Gamma}_L(t)\rvert=o_{\prob}(1)$ as $n,m,K \rightarrow \infty$,  whenever $t$ is a continuity point of $\bar{\Gamma}_L(t)$. Since $\bar{\Gamma}_L(\cdot)$ is continuous and strictly increasing at $\bar{q}_\alpha$, one obtains $\lvert \hat{q}_\alpha - \bar{q}_\alpha \rvert = o_{\prob}(1)$. 
	\item Let $B_{nm}$ denote the event $B_{nm}=\{\hat{q}_\alpha \leq 2\bar{q}_\alpha\}$. Assuming that we can establish the missing details in Step $\mathrm{II}$, one has $\lvert \hat{q}_\alpha - \bar{q}_\alpha \rvert = o_{\prob}(1)$, which implies that $\prob(B_{nm}) \rightarrow 1$ as $n \rightarrow \infty$. With $A_{nm,\eta}$, $A^{(1)}_{nm}$ and $A^{(2)}_{nm}$ as in 
	the proof of Theorem \ref{thm: power},  observe 
	\begin{align*}
	\tilde{\beta}^w_{n,m} = & \prob_{H_{n,m}} \left( T^w_{n,m} > \hat{q}_\alpha \right) \\ \geq & \prob_{H_{n,m}} \left( T^w_{n,m} > \hat{q}_\alpha, A_{nm,\eta} \cap A^{(1)}_{nm}\cap A^{(2)}_{nm}, B_{nm} \right) \\ \geq & \prob_{H_{n,m}} \left( \frac{nm}{n+m} U_1 > 2 \bar{q}_\alpha - \frac{nm}{n+m} \frac{a_{nm}}{4}, A_{nm,\eta} \cap A^{(1)}_{nm}\cap A^{(2)}_{nm}, B_{nm}\right),
	\end{align*}
	where the last step follows using \eqref{eq:pow} and since $\hat{q}_\alpha \leq 2\bar{q}_\alpha$ under $B_{nm}$. Since $\frac{nm}{n+m} U_1 = O_{\prob}(1)$ and $\frac{nm}{n+m}a_{nm} \rightarrow \infty$ as $m,n \rightarrow \infty$, and due to the fact that $\prob(A_{nm,\eta} \cap A^{(1)}_{nm}\cap A^{(2)}_{nm})$ and $\prob(B_{nm})$ converges to $1$ as $n,m \rightarrow \infty$
	and $\eta\rightarrow\infty$, 
	one has $\tilde{\beta}^w_{n,m} \rightarrow 1$ as $n,m \rightarrow \infty$. 
\end{enumerate}
\noindent Hence it remains to establish the open ends in Steps $\mathrm{I}$ and $\mathrm{II}$ to complete the proof. 

\vspace{1em}

\noindent \underline{Step $\mathrm{I}$}: 
What remains to be shown is that condition \eqref{eq:condition_perm} holds with $\Gamma(\cdot) \equiv \Gamma_L(\cdot)$ under $H_0$.
For any permutation $\pi$ and $V \sim P_1$, where $P_1=P_2$ under $H_0$, define $\tilde{T}^w_{n,m,\pi}$ as
\begin{equation*}
\tilde{T}^w_{n,m,\pi} = 2\, \E_V\left( \int \left\lbrace W^{n,m}_{V,\pi} \right \rbrace^2 \df{u} \right),
\end{equation*}
where 
\begin{align*}
W^{n,m}_{x,\pi}(u) & = \sqrt{\frac{nm w_x(u)}{n+m}} \left\lbrace \denom[n]\add[in]\ind[d(x,V_{\pi(i)})\leq u]-F_{x}^V(u) \right. \\ 
& \quad -\left.\denom[m]\add[km]\ind[d(x,V_{\pi(k+n)})\leq u]+F_{x}^V(u) \right \rbrace
\end{align*}
and $F_x^V(u)=\prob\left( d(x,V) \leq u \right)$. Likewise for permutation $\pi'$, one obtains $\tilde{T}^w_{n,m,\pi'}$ from $W^{n,m}_{x,\pi'}(u)$ analogously. Using arguments in the proof of Theorem \ref{thm: null_dist} conditional on $\pi$ one has that under $H_0$,
\begin{align}
\left \lvert T^{w}_{n,m}(X_{\pi},Y_{\pi}) - \tilde{T}^w_{n,m,\pi} \right \rvert \xrightarrow{P|\pi} 0
\end{align}
as $n,m \rightarrow \infty$, where $\xrightarrow{P|\pi}$ is the notation for convergence in probability conditional on $\pi$. Hence $\left \lvert T^{w}_{n,m}(X_{\pi},Y_{\pi}) - \tilde{T}^w_{n,m,\pi} \right \rvert=o_{\prob}(1)$ and $\left \lvert T^{w}_{n,m}(X_{\pi'},Y_{\pi'}) - \tilde{T}^w_{n,m,\pi'} \right \rvert=o_{\prob}(1)$ as $n,m \rightarrow \infty$ unconditionally as well. By Slutsky's Theorem, the asymptotic distribution of $ \left( T^{w}_{n,m}(X_{\pi},Y_{\pi}), T^{w}_{n,m}(X_{\pi'},Y_{\pi'}) \right)$ is determined by that of $ \left( \tilde{T}^w_{n,m,\pi}, \tilde{T}^w_{n,m,\pi'} \right)$ and hence it is sufficient to show that under $H_0$, as $n,m \rightarrow \infty$,
\begin{equation}
\label{eq:condition_perm_relaxed}
\left( \tilde{T}^w_{n,m,\pi}, \tilde{T}^w_{n,m,\pi'} \right) \xrightarrow{D} (T,T')
\end{equation}
such that $T$ and $T'$ are independent, each with a common c.d.f. $\Gamma_L(\cdot)$. Observe that by the continuous mapping theorem and the arguments in the proof of Theorem \ref{thm: null_dist}, the condition in \eqref{eq:condition_perm_relaxed} holds if one can establish the following fact:  
The two-dimensional stochastic process $(x,x',u)\mapsto \left( W^{n,m}_{x,\pi}(u), W^{n,m}_{x',\pi'}(u) \right) $ with $x,x'\in\Omega$ and $u\in[0,\mathcal{M}]$,  $\mathcal{M}=\mathrm{diam}(\Omega)$, converges weakly to $(x,x',u)\mapsto\left(\mathcal{G}_x(u),\mathcal{G}'_{x'}(u)\right)$, 
where $(x,u) \mapsto \mathcal{G}_x(u)$ and $(x',u) \mapsto \mathcal{G}'_{x'}(u)$ are independent zero mean Gaussian processes on $\Omega\times[0,\mathcal{M}]$ with covariances  $C_{x_1,x_2}(u_1,u_2)$ and   $C_{x'_1,x'_2}(u'_1,u'_2)$, respectively, where $$C_{x_1,x_2}(u_1,u_2)=\sqrt{w_{x_1}(u_1)w_{x_2}(u_2)} \ \mathrm{Cov}(\ind[d(x_1,V)\leq u_1],\ind[d(x_2,V)\leq u_2])$$ and $$C_{x'_1,x'_2}(u'_1,u'_2)=\sqrt{w_{x_1'}(u'_1)w_{x_2'}(u'_2)} \ \mathrm{Cov}(\ind[d(x_1',V)\leq u'_1],\ind[d(x_2',V)\leq u'_2]).$$ By Lemma 1.5.4 in \cite{well:96}, it is sufficient to show that 
\\ \subitem i) the two-dimensional stochastic process $(x,x',u)\mapsto \left( W^{n,m}_{x,\pi}(u), W^{n,m}_{x',\pi'}(u) \right) $ on $\Omega\times\Omega\times[0,\mathcal{M}]$ is asymptotically tight. By Lemma 1.4.3 in \cite{well:96} this is implied if  the stochastic processes $(x,u)\mapsto W^{n,m}_{x,\pi}(u)$ and $(x',u)\mapsto W^{n,m}_{x',\pi'}(u)$ are each asymptotically tight.
\\ \subitem ii) for any choice of $(x_1,u_1),(x_2,u_2), (x'_1,u'_1), (x'_2,u'_2) \in \Omega \times [0,\mathcal{M}]$ one has that $\left( W^{n,m}_{x_1,\pi}(u_1), W^{n,m}_{x_2,\pi}(u_2), W^{n,m}_{x_1',\pi'}(u'_1) , W^{n,m}_{x'_2,\pi'}(u'_2) \right) $ converges in distribution as $n,m \rightarrow \infty$ to a multivariate Gaussian distribution  with zero mean and covariance   \medskip  

\begin{center} $\begin{pmatrix}
	C_{x_1,x_1}(u_1,u_1) & C_{x_1,x_2}(u_1,u_2) & 0 & 0\\
	C_{x_1,x_2}(u_1,u_2) & C_{x_2,x_2}(u_2,u_2) & 0 & 0\\
	0 & 0 & C_{x'_1,x'_1}(u'_1,u'_1) & C_{x'_1,x'_2}(u'_1,u'_2) \\
	0 & 0 & C_{x'_1,x'_2}(u'_1,u'_2) & C_{x'_2,x'_2}(u'_2,u'_2)
	\end{pmatrix}.$
\end{center} 
\medskip

\noindent Observe that by using arguments in the proof of Theorem \ref{thm: null_dist} conditionally on $\pi$ one has that $(x,u) \mapsto W^{n,m}_{x,\pi}(u)$ converges weakly to $(x,u) \mapsto \mathcal{G}_x(u)$ as $n,m \rightarrow \infty$, where the stochastic process $(x,u) \mapsto \mathcal{G}_x(u) $ is independent of $\pi$. Therefore,  unconditionally $(x,u) \mapsto W^{n,m}_{x,\pi}(u)$ converges weakly to the process $(x,u) \mapsto \mathcal{G}_x(u)$ and $(x',u) \mapsto W^{n,m}_{x',\pi'}(u)$ converges weakly to the process $(x',u) \mapsto \mathcal{G}'_{x'}(u)$ as $n,m \rightarrow \infty$. By Theorem 1.5.4 in \cite{well:96} this implies that the processes $(x,u)\mapsto W^{n,m}_{x,\pi}(u)$ and $(x',u)\mapsto W^{n,m}_{x',\pi'}(u)$ are asymptotically tight, thereby implying condition i) and condition ii), provided that $\left( W^{n,m}_{x_1,\pi}(u_1), W^{n,m}_{x_2,\pi}(u_2), W^{n,m}_{x_1',\pi'}(u'_1) , W^{n,m}_{x'_2,\pi'}(u'_2) \right)$ jointly converge to a multivariate Gaussian distribution and $\mathrm{Cov}\left(W^{n,m}_{x,\pi}(u),W^{n,m}_{x',\pi'}(u')\right)=0$ for any $x,x' \in \Omega$ and any $u,u' \in [0,\mathcal{M}]$. 
In what follows we show that $\mathrm{Cov}\left(W^{n,m}_{x,\pi}(u),W^{n,m}_{x',\pi'}(u')\right)=0$ for any $x,x' \in \Omega$ and any $u,u' \in [0,\mathcal{M}]$
and then  use the Cram\'er--Wold theorem to show that for any $l=(l_1,l_2,l_3,l_4)\tps \in \mathbb{R}^4$, $l_1 W^{n,m}_{x_1,\pi}(u_1)+ l_2 W^{n,m}_{x_2,\pi}(u_2)+ l_3 W^{n,m}_{x'_1,\pi'}(u'_1) + l_4 W^{n,m}_{x'_2,\pi'}(u'_2)$ converges to a Gaussian  distribution, excluding the trivial case when $l=(0,0,0,0)\tps$, which then concludes the proof for  Step $\mathrm{I}$. \\

\noindent For $x \in \Omega$ and $u \in [0,\mathcal{M}]$, $W^{n,m}_{x,\pi}(u)$ can be expressed as
\begin{equation*}
W^{n,m}_{x,\pi}(u) = \sqrt{\frac{nm w_x(u)}{n+m}} \left[\sum_{i=1}^{n+m} \xi_{\pi^{-1}(i)} \left \lbrace \ind[d(x,V_i)\leq u] - F_x^V(u) \right \rbrace \right],
\end{equation*}
where $\xi_{i} = \frac{1}{n}$ if $i \leq n$ and $\xi_{i} = -\frac{1}{m}$ if $i \geq n+1$. For a random $\pi$ uniformly distributed in $\Pi_{n,m}$, $\E\left( \xi_{\pi^{-1}(i)} \right) = 0$. By the independence of $\pi$, $\pi'$ and the data $V_1, \dots, V_n$, one has
\begin{equation*}
\E \left( \xi_{\pi^{-1}(i)} \left \lbrace \ind[d(x,V_i)\leq u] - F_x^V(u) \right \rbrace \right) \\ =  \E \left( \xi_{\pi^{-1}(i)} \right) \E \left \lbrace \ind[d(x,V_i)\leq u] - F_x^V(u) \right \rbrace  =  0 
\end{equation*}
and similarly for $x' \in \Omega$ and $u' \in [0,\mathcal{M}]$, $\E \left( \xi_{\pi'^{-1}(j)} \left \lbrace \ind[d(x',V_j)\leq u'] - F_{x'}^V(u') \right \rbrace \right) =0$. Moreover,
\begin{align*}
&\ \E \left( \xi_{\pi^{-1}(i)} \left \lbrace \ind[d(x,V_i)\leq u] - F_x^V(u) \right \rbrace \xi_{\pi'^{-1}(j)} \left \lbrace \ind[d(x',V_j)\leq u'] - F_{x'}^V(u') \right \rbrace \right) \\ 
= &\ \E \left( \xi_{\pi^{-1}(i)} \right) \E \left( \xi_{\pi'^{-1}(j)} \right) \E \left \lbrace \left(\ind[d(x,V_i)\leq u] - F_x^V(u)\right) 
\left(\ind[d(x',V_j)\leq u'] - F_{x'}^V(u')\right)\right \rbrace \\ = &\ 0.
\end{align*}
Hence for any $i,j \in \{1,2,\dots, n+m\}$, 
\begin{equation*}
\mathrm{Cov}\left(\xi_{\pi^{-1}(i)} \left \lbrace \ind[d(x,V_i)\leq u] - F_x^V(u) \right \rbrace , \xi_{\pi'^{-1}(j)} \left \lbrace \ind[d(x',V_j)\leq u'] - F_{x'}^V(u') \right \rbrace \right) = 0,
\end{equation*}
which leads to  $\mathrm{Cov}\left(W^{n,m}_{x,\pi}(u),W^{n,m}_{x',\pi'}(u')\right)=0$ for any $x,x' \in \Omega$ and any $u,u' \in [0,\mathcal{M}]$.

Finally observe that
\begin{align*}
& l_1 W^{n,m}_{x_1,\pi}(u_1)+ l_2 W^{n,m}_{x_2,\pi}(u_2)+ l_3 W^{n,m}_{x'_1,\pi'}(u'_1) + l_4 W^{n,m}_{x'_2,\pi'}(u'_2) = \sum_{i=1}^{n+m} \Xi_{i}^{\pi,\pi',l},
\end{align*}
where we define $ \Xi_i^{\pi,\pi',l} = \sqrt{\frac{nm}{n+m}} [ l_1 \xi_{\pi^{-1}(i)} \sqrt{w_{x_1}(u_1)} \lbrace \ind[d(x_1,V_i)\leq u_1] - F_{x_1}^V(u_1) \rbrace + l_2 \xi_{\pi^{-1}(i)} \sqrt{w_{x_2}(u_2)} \lbrace \ind[d(x_2,V_i)\leq u_2] - F_{x_2}^V(u_2) \rbrace + l_3 \xi_{\pi'^{-1}(i)} \sqrt{w_{x'_1}(u'_1)} \lbrace \ind[d(x'_1,V_i)\leq u'_1] - F_{x'_1}^V(u'_1) \rbrace + l_4 \xi_{\pi'^{-1}(i)} \sqrt{w_{x'_2}(u'_2)} \lbrace \ind[d(x'_2,V_i)\leq u'_2] - F_{x'_2}^V(u'_2) \rbrace ] $. 
To simplify notations, let 
\begin{align*}
&c^{x,x}_{ij} =\\ &\Cov \left(\sqrt{w_{x_i}(u_i)} \lbrace \ind[d(x_i,V_i)\leq u_i] - F_{x_i}^V(u_i) \rbrace, \sqrt{w_{x_j}(u_j)} \lbrace \ind[d(x_j,V_i)\leq u_j] - F_{x_j}^V(u_j) \rbrace\right),\\
&c^{x',x'}_{ij} =\\ &\Cov \left(\sqrt{w_{x'_i}(u'_i)} \lbrace \ind[d({x'_i},V_i)\leq u'_i] - F_{x'_i}^V(u'_i) \rbrace, \sqrt{w_{x'_j}(u'_j)} \lbrace \ind[d({x'_j},V_i)\leq u'_j] - F_{x'_j}^V(u'_j) \rbrace \right),\\
&c^{x,x'}_{ij}=\\ &\Cov\left(\sqrt{w_{x_i}(u_i)} \lbrace \ind[d({x_i},V_i)\leq u_i] - F_{x_i}^V(u_i) \rbrace, \sqrt{w_{x'_j}(u'_j)} \lbrace \ind[d({x'_j},V_i)\leq u'_j] - F_{x'_j}^V(u'_j) \rbrace\right).
\end{align*}
By symmetry, one has $c^{x,x}_{ij}=c^{x,x}_{ji}, c^{x',x'}_{ij}=c^{x',x'}_{ji}$ and $c^{x,x'}_{ij}=c^{x',x}_{ji}$. Conditioning on $\pi$ and $\pi'$, $\{\Xi_i^{\pi,\pi',l}\}_{i=1}^{n+m}$ are independent with conditional moments $\E( \Xi_i^{\pi,\pi',l} \mid \pi, \pi')=0$ and $\Var( \Xi_i^{\pi,\pi',l} \mid \pi, \pi') = \sigma^2_{\pi,\pi',i}$,  where
\begin{align*}
\sigma^2_{\pi,\pi',i} = \frac{nm}{n+m} [ & l_1^2 \xi^2_{\pi^{-1}(i)} c^{x,x}_{11}+l_2^2 \xi^2_{\pi^{-1}(i)} c^{x,x}_{22} + l_3^2 \xi^2_{\pi'^{-1}(i)} c^{x',x'}_{33}+ l_4^2 \xi^2_{\pi'^{-1}(i)} c^{x',x'}_{44} \\ + & 2l_1l_2\xi^2_{\pi^{-1}(i)}c^{x,x}_{12}+2l_1l_3\xi_{\pi^{-1}(i)}\xi_{\pi'^{-1}(i)}c^{x,x'}_{13}+2l_1l_4\xi_{\pi^{-1}(i)}\xi_{\pi'^{-1}(i)}c^{x,x'}_{14} \\ + & 2l_2l_3\xi_{\pi^{-1}(i)}\xi_{\pi'^{-1}(i)}c^{x,x'}_{23} + 2l_2l_4\xi_{\pi^{-1}(i)}\xi_{\pi'^{-1}(i)}c^{x,x'}_{24} + 2l_3l_4\xi^2_{\pi'^{-1}(i)}c^{x',x'}_{34}] \\ = \frac{nm}{n+m} [ & \xi^2_{\pi^{-1}(i)} \lbrace l_1^2 c^{x,x}_{11}+l_2^2 c^{x,x}_{22}+2l_1l_2 c^{x,x}_{12}\rbrace + \xi^2_{\pi'^{-1}(i)} \lbrace l_3^2 c^{x',x'}_{33}+ l_4^2 c^{x',x'}_{44}+ 2l_3l_4 c^{x',x'}_{34}\rbrace \\ + & \xi_{\pi^{-1}(i)}\xi_{\pi'^{-1}(i)} \lbrace 2l_1l_3 c^{x,x'}_{13}+2l_1l_4 c^{x,x'}_{14}+2l_2l_3 c^{x,x'}_{23} + 2l_2l_4 c^{x,x'}_{24} \rbrace ].
\end{align*}
Therefore one has
\begin{align*}
\sum_{i=1}^{n+m} \sigma^2_{\pi,\pi',i} &= \frac{nm}{n+m} \left[ \sum_{i=1}^{n+m}\xi^2_{\pi^{-1}(i)} \lbrace l_1^2 c^{x,x}_{11}+l_2^2 c^{x,x}_{22}+2l_1l_2 c^{x,x}_{12}\rbrace\right. \\ 
&\quad + \sum_{i=1}^{n+m} \xi^2_{\pi'^{-1}(i)} \lbrace l_3^2 c^{x',x'}_{33}+ l_4^2 c^{x',x'}_{44}+ 2l_3l_4 c^{x',x'}_{34}\rbrace \\ 
&\quad + \left.\sum_{i=1}^{n+m} \xi_{\pi^{-1}(i)}\xi_{\pi'^{-1}(i)} \lbrace 2l_1l_3 c^{x,x'}_{13}+2l_1l_4 c^{x,x'}_{14}+2l_2l_3 c^{x,x'}_{23} + 2l_2l_4 c^{x,x'}_{24} \rbrace \right]. 
\end{align*}
Before proceeding note that $\sum_{i=1}^{n+m}\xi^2_{\pi^{-1}(i)} = \sum_{i=1}^{n+m} \xi^2_{\pi'^{-1}(i)} = \frac{n+m}{nm}$. 
Second, $l_1^2 c^{x,x}_{11}+l_2^2 c^{x,x}_{22}+2l_1l_2 c^{x,x}_{12} = l_{12}\tps \begin{pmatrix}
c^{x,x}_{11} & c^{x,x}_{12} \\ c^{x,x}_{12} & c^{x,x}_{22}
\end{pmatrix}l_{12}$, where $l_{12}=(l_1,l_2)\tps$ and $l_3^2 c^{x',x'}_{33}+ l_4^2 c^{x',x'}_{44}+ 2l_3l_4 c^{x',x'}_{34} = l_{34}\tps \begin{pmatrix}
c^{x',x'}_{33} & c^{x',x'}_{34} \\ c^{x',x'}_{34} & c^{x',x'}_{44}
\end{pmatrix}l_{34}$, where $l_{34}=(l_3,l_4)\tps$. Observe that $\begin{pmatrix}
c^{x,x}_{11} & c^{x,x}_{12} \\ c^{x,x}_{12} & c^{x,x}_{22}
\end{pmatrix}$ and $\begin{pmatrix}
c^{x',x'}_{33} & c^{x',x'}_{34} \\ c^{x',x'}_{34} & c^{x',x'}_{44}
\end{pmatrix}$ are two-dimensional snapshots of the covariance surfaces of the Gaussian processes $(x,u) \mapsto \mathcal{G}_x(u)$ and $(x',u) \mapsto \mathcal{G}_{x'}(u)$, respectively, therefore  the matrices $\begin{pmatrix}
c^x_{11} & c^x_{12} \\ c^x_{12} & c^x_{22}
\end{pmatrix}$ and $\begin{pmatrix}
c^{x'}_{33} & c^{x'}_{34} \\ c^{x'}_{34} & c^{x'}_{44}
\end{pmatrix}$ are symmetric positive definite. Hence for $l \neq 0$ one has $ l_1^2 c^{x,x}_{11}+l_2^2 c^{x,x}_{22}+2l_1l_2 c^{x,x}_{12} > 0$ and $l_3^2 c^{x',x'}_{33}+ l_4^2 c^{x',x'}_{44}+ 2l_3l_4 c^{x',x'}_{34} > 0$, implying 
\begin{equation}
\label{eq:sigma_limit}
C^{x,x'}_{l,1}=l_1^2 c^{x,x}_{11}+l_2^2 c^{x,x}_{22}+2l_1l_2 c^{x,x}_{12}+l_3^2 c^{x',x'}_{33}+ l_4^2 c^{x',x'}_{44}+ 2l_3l_4 c^{x',x'}_{34} > 0.
\end{equation}

Further,  with 
$C^{x,x'}_{l,1}$ as in \eqref{eq:sigma_limit} and $C^{x,x'}_{l,2}=  2l_1l_3 c^{x,x'}_{13}+2l_1l_4 c^{x,x'}_{14}+2l_2l_3 c^{x,x'}_{23} + 2l_2l_4 c^{x,x'}_{24} $ one has
\begin{equation*}
\sum_{i=1}^{n+m} \sigma^2_{\pi,\pi',i} = C^{x,x'}_{l,1} + C^{x,x'}_{l,2} \frac{nm}{n+m} \sum_{i=1}^{n+m} \xi_{\pi^{-1}(i)}\xi_{\pi'^{-1}(i)}.
\end{equation*}
By the independence of $\pi$ and $\pi'$, and since $\E\left( \xi_{\pi^{-1}(i)} \right) = 0$ and $\E\left( \xi_{\pi'^{-1}(i)} \right) = 0$, we have $\E( \sum_{i=1}^{n+m} \sigma^2_{\pi,\pi',i})=C^{x,x'}_{l,1}$ and $\mathrm{Var}(\sum_{i=1}^{n+m} \xi_{\pi^{-1}(i)}\xi_{\pi'^{-1}(i)})=\E\left[\left(\sum_{i=1}^{n+m} \xi_{\pi^{-1}(i)}\xi_{\pi'^{-1}(i)}\right)^2\right]$. Hence $\mathrm{Var}(\sum_{i=1}^{n+m} \sigma^2_{\pi,\pi',i})= (C^{x,x'}_{l,2})^2 \left( \frac{nm}{n+m} \right)^2 \E\left[\left(\sum_{i=1}^{n+m} \xi_{\pi^{-1}(i)}\xi_{\pi'^{-1}(i)}\right)^2\right]$, where by the independence of $\pi$ and $\pi'$,
\begin{align*}
& \E\left[\left(\sum_{i=1}^{n+m} \xi_{\pi^{-1}(i)}\xi_{\pi'^{-1}(i)}\right)^2\right] \\ = & \sum_{i=1}^{n+m} \E\left(\xi^2_{\pi^{-1}(i)}\right) \E\left(\xi^2_{\pi'^{-1}(i)}\right) + \sum_{i=1}^{n+m}\sum_{j \neq i} \E\left(\xi_{\pi^{-1}(i)}\xi_{\pi^{-1}(j)}\right) \E\left(\xi_{\pi'^{-1}(i)} \xi_{\pi'^{-1}(j)}\right).
\end{align*}
When $\pi$ is uniform on $\Pi_{nm}$, $\E\left(\xi^2_{\pi^{-1}(i)}\right) = \frac{1}{nm}$ and \begin{align*}
&\ \E\left(\xi_{\pi^{-1}(i)}\xi_{\pi^{-1}(j)}\right) \\ 
= &\ \frac{1}{n^2} \frac{n(n-1)}{(n+m)(n+m-1)} + \frac{1}{m^2} \frac{m(m-1)}{(n+m)(n+m-1)} - \frac{1}{nm} \frac{2mn}{(n+m)(n+m-1)} \\ 
= &\ -\frac{1}{nm(n+m-1)},
\end{align*}
whence 
\begin{align*}
\Var\left(\sum_{i=1}^{n+m} \sigma^2_{\pi,\pi',i}\right)
&= (C^{x,x'}_{l,2})^2 \left( \frac{nm}{n+m} \right)^2 \left\lbrace \frac{n+m}{n^2 m^2} + \frac{(n+m)(n+m-1)}{n^2 m^2 (n+m-1)^2}\right \rbrace \\ 
&= (C^{x,x'}_{l,2})^2 \frac{1}{n+m-1}.
\end{align*}
By  Chebychev's inequality, 
\begin{equation}
\label{eq:sigma_consistency}
\sum_{i=1}^{n+m} \sigma^2_{\pi,\pi',i} \xrightarrow{P} C^{x,x'}_{l,1} 
\end{equation}
as $n+m \rightarrow \infty$   
and we observe that under Assumption~\ref{ass:samp_ratio} there exists a constant $C > 0$ such that for each $i=1,\dots,n$,
\begin{align*}
\left \lvert \Xi_i^{\pi,\pi',l} \right \rvert \leq C \frac{1}{\sqrt{n+m}}.
\end{align*}
Hence for any $\epsilon > 0$,
\begin{equation*}
\frac{1}{ \sum_{i=1}^{n+m} \sigma^2_{\pi,\pi',i} } \sum_{i=1}^{n+m} \E \left[ \left\lvert \Xi_i^{\pi,\pi',l}\right\rvert^2 \ind[\left \lvert \Xi_i^{\pi,\pi',l} \right \rvert > \epsilon \sqrt{\sum_{i=1}^{n+m} \sigma^2_{\pi,\pi',i}} ] \right]\rightarrow 0 
\end{equation*}
as $n+m \rightarrow \infty$. By the Lindeberg's CLT, 
$$\frac{l_1 W^{n,m}_{x_1,\pi}(u_1)+ l_2 W^{n,m}_{x_2,\pi}(u_2)+ l_3 W^{n,m}_{x'_1,\pi'}(u'_1) + l_4 W^{n,m}_{x'_2,\pi'}(u'_2)}{\sqrt{\sum_{i=1}^{n+m} \sigma^2_{\pi,\pi',i}} } \xrightarrow{D} N(0,1)$$ 
as $n+m \rightarrow \infty$ conditionally on $\pi$ and $\pi'$. Since the asymptotic distribution does not depend on $\pi$ or $\pi'$,  it follows that 
$$\frac{l_1 W^{n,m}_{x_1,\pi}(u_1)+ l_2 W^{n,m}_{x_2,\pi}(u_2)+ l_3 W^{n,m}_{x'_1,\pi'}(u'_1) + l_4 W^{n,m}_{x'_2,\pi'}(u'_2)}{\sqrt{\sum_{i=1}^{n+m} \sigma^2_{\pi,\pi',i}} } \xrightarrow{D} N(0,1)$$ unconditionally as well. Slutsky's Theorem,  \eqref{eq:sigma_limit} and \eqref{eq:sigma_consistency}
then imply $$\frac{l_1 W^{n,m}_{x_1,\pi}(u_1)+ l_2 W^{n,m}_{x_1,\pi}(u_2)+ l_3 W^{n,m}_{x'_1,\pi'}(u'_1) + l_4 W^{n,m}_{x'_2,\pi'}(u'_2)}{\sqrt{C^{x,x'}_{l,1}} } \xrightarrow{D} N(0,1),$$ which concludes the proof for Step $\mathrm{I}$.

\vspace{0.2in}

\noindent \underline{Step $\mathrm{II}$}: 
What remains to be shown is that $\lvert T^w_{n,m}(\bar{X}_{\pi \pi_0},\bar{Y}_{\pi\pi_0})-T^w_{n,m}(X_{\pi},Y_{\pi})\rvert = o_{\prob}(1)$, where $\pi\pi_0$ is the composition of the permutations $\pi$ and $\pi_0$ with $\pi_0$ applied first. 
Observe that
$T^w_{n,m}(\bar{X}_{\pi \pi_0},\bar{Y}_{\pi\pi_0})= \int \bar{S}_{n,m}(u) \df{u}$ and $T^w_{n,m}({X}_{\pi},{Y}_{\pi})= \int {S}_{n,m}(u) \df{u}$, where with $r_{i,n,m}$ defined as $r_{i,n,m} =\frac{1}{n}$ for $i \leq n$ and $r_{i,n,m} = -\frac{1}{m} $ for $n+1 \leq i \leq n+m$, we have
\begin{equation*}
\bar{S}_{n,m}(u)=\frac{nm}{n+m} \left[ \sum_{i=1}^{n+m} r_{i,n,m} \hat{w}_{\bar{V}_{\pi\pi_0(i)}}(u) \left\lbrace \hat{F}^{\bar{X}_{\pi\pi_0}}_{\bar{V}_{\pi\pi_0(i)}}(u)-\hat{F}^{\bar{Y}_{\pi\pi_0}}_{\bar{V}_{\pi\pi_0(i)}}(u)\right \rbrace^2\right],
\end{equation*}
where $\hat{F}^{\bar{X}_{\pi\pi_0}}_{x}(u)=\frac{1}{n} \sum_{i=1}^n \ind[d(x,\bar{V}_{\pi\pi_0(i)}) \leq u]$ and \newline $\hat{F}^{\bar{Y}_{\pi\pi_0}}_{x}(u)=\frac{1}{m} \sum_{j=1}^m \ind[d(x,\bar{V}_{\pi\pi_0(n+j)}) \leq u]$. Furthermore, 
\begin{equation*}
{S}_{n,m}(u)=\frac{nm}{n+m} \left[ \sum_{i=1}^{n+m} r_{i,n,m} \hat{w}_{{V}_{\pi(i)}}(u) \left\lbrace \hat{F}^{{X}_\pi}_{{V}_{\pi(i)}}(u)-\hat{F}^{{Y}_\pi}_{{V}_{\pi(i)}}(u)\right \rbrace^2\right],
\end{equation*}
where $\hat{F}^{{X}_\pi}_{x}(u)=\frac{1}{n} \sum_{i=1}^n \ind[d(x,{V}_{\pi(i)}) \leq u]$ and $\hat{F}^{{Y}_\pi}_{x}(u)=\frac{1}{m} \sum_{j=1}^m \ind[d(x,{V}_{\pi(n+j)}) \leq u]$. If we can show that $\sup_u \lvert \bar{S}_{n,m}(u) - {S}_{n,m}(u)\rvert = o_{\prob}(1)$, then by the continuous mapping theorem we have $\lvert T^w_{n,m}(\bar{X}_{\pi \pi_0},\bar{Y}_{\pi\pi_0})-T^w_{n,m}(X_{\pi},Y_{\pi})\rvert = o_{\prob}(1)$, which completes the proof. \\

\noindent Let  $\tilde{I}=\{i : \bar{V}_{\pi_0(i)} \neq V_{i}\} \subset \{1,2,\dots, n+m\}$ be the collection of (random) indices, where $\bar{V}_{\pi_0}$ does not match ${V}$ such that the cardinality of $\tilde{I}$ is  $\lvert \tilde{I}\rvert = D$. 
Let {$\tilde{I}^C = \{1,\dots,n+m\}\backslash \tilde{I}$}. 
Observe that one may write $\bar{S}_{n,m}(u)$ and ${S}_{n,m}(u)$ as 
\begin{equation*}
\bar{S}_{n,m}(u) = \frac{nm}{n+m} \left[ \sum_{i=1}^{n+m} r_{\pi^{-1}(i),n,m} \hat{w}_{\bar{V}_{\pi_0(i)}}(u) \left\lbrace \hat{F}^{\bar{X}_{\pi\pi_0}}_{\bar{V}_{\pi_0(i)}}(u)-\hat{F}^{\bar{Y}_{\pi\pi_0}}_{\bar{V}_{\pi_0(i)}}(u)\right \rbrace^2 \right]
\end{equation*}
and 
\begin{equation*}
{S}_{n,m}(u)=\frac{nm}{n+m} \left[ \sum_{i=1}^{n+m} r_{\pi^{-1}(i),n,m} \hat{w}_{{V}_{i}}(u) \left\lbrace \hat{F}^{{X}_\pi}_{{V}_{i}}(u)-\hat{F}^{{Y}_\pi}_{{V}_{i}}(u)\right \rbrace^2\right].
\end{equation*}
Setting  $r^{\pi}_i = r_{\pi^{-1}(i),n,m}$, we can write 
\begin{equation*}
\bar{S}_{n,m}(u) - {S}_{n,m}(u) = (\mathrm{S.I})+(\mathrm{S.II}),
\end{equation*}
where, using that for $i \in \tilde{I}^C$, $\bar{V}_{\pi_0(i)}=V_i$,
\begin{align*}
(\mathrm{S.I}) =\\ & \hspace{-1cm} \frac{nm}{n+m} \left[ \sum_{i\in \tilde{I}^C} r^{\pi}_i\hat{w}_{\bar{V}_{\pi_0(i)}}(u) \left\lbrace \left(\hat{F}^{\bar{X}_{\pi\pi_0}}_{\bar{V}_{\pi_0(i)}}(u)-\hat{F}^{\bar{Y}_{\pi\pi_0}}_{\bar{V}_{\pi_0(i)}}(u)\right)^2-\left( \hat{F}^{{X}_\pi}_{\bar{V}_{\pi_0(i)}}(u)-\hat{F}^{{Y}_{\pi}}_{\bar{V}_{\pi_0(i)}}(u)\right)^2\right \rbrace \right], \\ 
(\mathrm{S.II}) = \\ & \hspace{-1.2cm} \frac{nm}{n+m} \left[ \sum_{i\in \tilde{I}} r^{\pi}_i \left\lbrace \hat{w}_{\bar{V}_{\pi_0(i)}}(u) \left(\hat{F}^{\bar{X}_{\pi\pi_0}}_{\bar{V}_{\pi_0(i)}}(u)-\hat{F}^{\bar{Y}_{\pi\pi_0}}_{\bar{V}_{\pi_0(i)}}(u)\right)^2-\hat{w}_{{V}_{i}}(u) \left( \hat{F}^{{X}_\pi}_{{V}_{i}}(u)-\hat{F}^{{Y}_\pi}_{{V}_{i}}(u)\right)^2\right \rbrace \right].
\end{align*}
Hence it is sufficient to establish that $\sup_u \lvert (\mathrm{S.I}) \rvert = o_{\prob}(1)$ and $\sup_u \lvert (\mathrm{S.II}) \rvert = o_{\prob}(1)$. \\

\noindent Observe that for any $x \in \Omega$,
\begin{align*}
& \left \lvert \left( \hat{F}^{\bar{X}_{\pi\pi_0}}_x(u)-\hat{F}^{\bar{Y}_{\pi\pi_0}}_x(u) \right) - \left( \hat{F}^{{X}_{\pi}}_x(u)-\hat{F}^{{Y}_{\pi}}_x(u) \right) \right \rvert \\ 
=\ & \left \lvert \frac{1}{n} \sum_{i=1}^n \left\lbrace \ind[d(x,\bar{V}_{\pi\pi_0(i)}) \leq u] - \ind[d(x,{V}_{\pi(i)}) \leq u] \right\rbrace \right. \\ 
& + \left. \frac{1}{m} \sum_{j=1}^m \left\lbrace \ind[d(x,\bar{V}_{\pi\pi_0(n+j)}) \leq u] - \ind[d(x,{V}_{\pi(n+j)}) \leq u] \right\rbrace \right\rvert \\ 
=\ & \left \lvert \frac{1}{n} \sum_{i \in \tilde{I},\, i\in [1,n]} \left\lbrace \ind[d(x,\bar{V}_{\pi\pi_0(i)}) \leq u] - \ind[d(x,{V}_{\pi(i)}) \leq u] \right\rbrace \right. \\ 
& + \left. \frac{1}{m} \sum_{j \in \tilde{I},\, j \in [n+1,n+m]} \left\lbrace \ind[d(x,\bar{V}_{\pi\pi_0(n+j)}) \leq u] - \ind[d(x,{V}_{\pi(n+j)}) \leq u] \right\rbrace \right\rvert \\ 
\leq\ & \frac{D}{n} \left \lvert \frac{1}{D} \sum_{i \in \tilde{I},\, i\in[1,n]} \left\lbrace \ind[d(x,\bar{V}_{\pi\pi_0(i)}) \leq u] - \ind[d(x,{V}_{\pi(i)}) \leq u] \right\rbrace \right\rvert \\ & + \frac{D}{m} \left\lvert \frac{1}{D} \sum_{j \in \tilde{I},\, j \in [n+1,n+m]} \left\lbrace \ind[d(x,\bar{V}_{\pi\pi_0(n+j)}) \leq u] - \ind[d(x,{V}_{\pi(n+j)}) \leq u] \right\rbrace \right\rvert.
\end{align*}
We will establish that $$\sqrt{n+m} \sup_{x,u}\left \lvert \left( \hat{F}^{\bar{X}_{\pi\pi_0}}_x(u)-\hat{F}^{\bar{Y}_{\pi\pi_0}}_x(u) \right) - \left( \hat{F}^{{X}_{\pi}}_x(u)-\hat{F}^{{Y}_{\pi}}_x(u) \right) \right \rvert = o_{\prob}(1)$$ as $n,m \rightarrow \infty$. For $\gamma < 1/2$, under Assumption~\ref{ass:samp_ratio}, writing  $\vartheta(n,m)=\mathbb{I} \left\lbrace D > (n+m)^\gamma \right \rbrace,$
\begin{equation}
\label{eq:part1}
\sup_{x,u} \left \lvert \hat{F}^{\bar{X}_{\pi\pi_0}}_x(u)-\hat{F}^{\bar{Y}_{\pi\pi_0}}_x(u) - \hat{F}^{{X}_{\pi}}_x(u)+\hat{F}^{{Y}_{\pi}}_x(u) \right \rvert \vartheta(n,m)  = O \left( \frac{1}{(n+m)^{1-\gamma}} \right).
\end{equation}
Let $\tilde{V}_1,\dots,\tilde{V}_{n+m}$ be i.i.d. observations from the mixture distribution $\bar{P}$, which are independent of $\bar{V}_1, \dots, \bar{V}_{n+m}$. Then conditional on $D$ and $\pi$, using Theorem~\ref{thm:fhat}, \begin{align*}
\sup_{x,u} \left \lvert \frac{1}{D} \sum_{i \in \tilde{I}, \leq n} \left\lbrace \ind[d(x,\bar{V}_{\pi\pi_0(i)}) \leq u] - \ind[d(x,\tilde{V}_{\pi(i)}) \leq u] \right\rbrace \right\rvert \vartheta(n,m)  =o_{\prob}(1)
\end{align*}
and 
\begin{align*}
& \hspace{-1cm}  \sup_{x,u} \left\lvert \frac{1}{D} \sum_{j \in \tilde{I}, j \in [n+1,n+m]} \left\lbrace \ind[d(x,\bar{V}_{\pi\pi_0(n+j)}) \leq u] - \ind[d(x,\tilde{V}_{\pi(n+j)}) \leq u] \right\rbrace \right\rvert \vartheta(n,m)\\ & = o_{\prob}(1)
\end{align*}
as $n,m \rightarrow \infty$,  which implies the above results hold unconditionally as well. This in conjunction with Lemma 5.3 in \cite{chun:13} gives 
\begin{equation}
\label{eq:part2}
\sup_{x,u} \left \lvert \frac{1}{D} \sum_{i \in \tilde{I}, \leq n} \left\lbrace \ind[d(x,\bar{V}_{\pi\pi_0(i)}) \leq u] - \ind[d(x,{V}_{\pi(i)}) \leq u] \right\rbrace \right\rvert \vartheta(n,m) =o_{\prob}(1)
\end{equation}
and
\begin{equation}
\label{eq:part3}
\sup_{x,u} \left\lvert \frac{1}{D} \sum_{j \in \tilde{I}, j \in [n+1,n+m]} \left\lbrace \ind[d(x,\bar{V}_{\pi\pi_0(n+j)}) \leq u] - \ind[d(x,{V}_{\pi(n+j)}) \leq u] \right\rbrace \right\rvert \vartheta(n,m) = o_{\prob}(1)
\end{equation}
as $n,m \rightarrow \infty$. 
Next observe that by inequality (5.8) in \cite{chun:13} and Assumption~\ref{ass:samp_ratio}, we have $\sqrt{n+m} \frac{D}{n}= O_{\prob}(1)$ and $\sqrt{n+m} \frac{D}{m}= O_{\prob}(1)$ as $n,m \rightarrow \infty$, which together with \eqref{eq:part2} and \eqref{eq:part3} implies
\begin{equation}
\label{eq:part4}
\sqrt{n+m} \sup_{x,u} \left \lvert \hat{F}^{\bar{X}_{\pi\pi_0}}_x(u)-\hat{F}^{\bar{Y}_{\pi\pi_0}}_x(u) - \hat{F}^{{X}_{\pi}}_x(u)+\hat{F}^{{Y}_{\pi}}_x(u) \right \rvert \vartheta(n,m) = o_{\prob}(1).
\end{equation}
Combining \eqref{eq:part4} and \eqref{eq:part1} as $n,m \rightarrow \infty$,
\begin{equation}
\label{eq:part5}
\sqrt{n+m} \sup_{x,u} \left \lvert \hat{F}^{\bar{X}_{\pi\pi_0}}_x(u)-\hat{F}^{\bar{Y}_{\pi\pi_0}}_x(u) - \hat{F}^{{X}_{\pi}}_x(u)+\hat{F}^{{Y}_{\pi}}_x(u) \right \rvert = o_{\prob}(1).
\end{equation}
Furthermore, observe that
\begin{align*}
& \sqrt{n+m} \sup_{x,u} \left \lvert \hat{F}^{\bar{X}_{\pi\pi_0}}_x(u)-\hat{F}^{\bar{Y}_{\pi\pi_0}}_x(u) + \hat{F}^{{X}_{\pi}}_x(u)- \hat{F}^{{Y}_{\pi}}_x(u) \right \rvert \\ \leq\ & 2 \sqrt{n+m} \sup_{x,u} \left \lvert \hat{F}^{\bar{X}_{\pi\pi_0}}_x(u)-\hat{F}^{\bar{Y}_{\pi\pi_0}}_x(u)  \right \rvert + o_{\prob}(1),
\end{align*}
using \eqref{eq:part5}. Conditional on $\pi$,  using Theorem~\ref{thm:fhat} and the continuous mapping theorem,  $\sqrt{n+m} \sup_{x,u} \left \lvert \hat{F}^{\bar{X}_{\pi\pi_0}}_x(u)-\hat{F}^{\bar{Y}_{\pi\pi_0}}_x(u)  \right \rvert$ converges in distribution to a random variable that does not depend on $\pi$. Hence as $n,m \rightarrow \infty$, $\sqrt{n+m} \sup_{x,u} \left \lvert \hat{F}^{\bar{X}_{\pi\pi_0}}_x(u)-\hat{F}^{\bar{Y}_{\pi\pi_0}}_x(u)  \right \rvert = O_{\prob}(1)$, implying
\begin{equation}
\label{eq:part6}
\sqrt{n+m} \sup_{x,u} \left \lvert \hat{F}^{\bar{X}_{\pi\pi_0}}_x(u)-\hat{F}^{\bar{Y}_{\pi\pi_0}}_x(u) + \hat{F}^{{X}_{\pi}}_x(u)- \hat{F}^{{Y}_{\pi}}_x(u) \right \rvert = O_{\prob}(1)
\end{equation}
as $n,m \rightarrow \infty$. Combining \eqref{eq:part5} and \eqref{eq:part6},
\begin{equation}
\label{eq:part7}
(n+m) \sup_{x,u} \left\lvert \left( \hat{F}^{\bar{X}_{\pi\pi_0}}_x(u)-\hat{F}^{\bar{Y}_{\pi\pi_0}}_x(u) \right)^2 - \left( \hat{F}^{{X}_{\pi}}_x(u)- \hat{F}^{{Y}_{\pi}}_x(u) \right)^2 \right \rvert = o_{\prob}(1)
\end{equation}
as $n,m \rightarrow \infty$. By Assumption \ref{ass:samp_ratio}, $r^{\pi}_i \leq \frac{C_r}{n+m}$ for some constant $C_r > 0$ almost surely for each $i=1, 2, \dots, n+m$ and by Assumption \ref{ass:assumption_weights}, $\sup_{x \in \Omega} \sup_{u \in \mcM} \lvert \hat{w}_x(u) \rvert \leq 2 C_w$ almost surely. This implies that
\begin{align*}
& \sup_u \lvert (\mathrm{S.I}) \rvert \\ \leq\ & 2 C_r C_w \frac{nm}{n+m} \frac{n+m-D}{n+m} \sup_{x,u} \left\lvert \left( \hat{F}^{\bar{X}_{\pi\pi_0}}_x(u)-\hat{F}^{\bar{Y}_{\pi\pi_0}}_x(u) \right)^2 - \left( \hat{F}^{{X}_{\pi}}_x(u)- \hat{F}^{{Y}_{\pi}}_x(u) \right)^2 \right \rvert.
\end{align*}
By \eqref{eq:part7}, 
one concludes $\sup_u \lvert (\mathrm{S.I}) \rvert = o_{\prob}(1)$ as $n,m \rightarrow \infty$. For term $(\mathrm{S.II})$ observe that
\begin{align*}
& \sup_u \lvert (\mathrm{S.II}) \rvert \\ \leq\ & 2 C_r C_w \frac{nm}{n+m} \frac{D}{n+m} \left\lbrace \sup_{x,u} \lvert \hat{F}^{\bar{X}_{\pi\pi_0}}_x(u)-\hat{F}^{\bar{Y}_{\pi\pi_0}}_x(u) \rvert^2 + \sup_{x,u} \lvert \hat{F}^{{X}_{\pi}}_x(u)- \hat{F}^{{Y}_{\pi}}_x(u) \rvert^2 \right \rbrace \\ \leq\ & 2 C_r C_w \frac{nm}{n+m} \frac{D}{n+m} \left\lbrace 2 \sup_{x,u} \lvert \hat{F}^{\bar{X}_{\pi\pi_0}}_x(u)-\hat{F}^{\bar{Y}_{\pi\pi_0}}_x(u) \rvert^2 + o_{\prob}\left(\frac{1}{n+m}\right) \right \rbrace,
\end{align*}
where the last line follows using \eqref{eq:part7}. As shown earlier,  $\sup_{x,u} \lvert \hat{F}^{\bar{X}_{\pi\pi_0}}_x(u)-\hat{F}^{\bar{Y}_{\pi\pi_0}}_x(u) \rvert^2 = O_{\prob} \left( \frac{1}{n+m} \right)$ as $n,m \rightarrow \infty$. {Together with (5.8) in \cite{chun:13}} 
this implies that $\sup_u \lvert (\mathrm{S.II}) \rvert = o_{\prob}(1)$ as $n,m \rightarrow \infty$, which completes the proof.

\section{Additional Simulation Results for Two-Sample Tests}\label{sec:simu_twosam_supp}

For the first scenario of multivariate data with location shifts, we generated multivariate Gaussian data with another shift direction, with the population covariance matrix $\Sigma = I_{\rdim}$ and the dimension $\rdim\in\{3,10,30\}$. 
The population mean vector is $\mu = \mathbf{0}_{\rdim}= (0,0,\dots,0)\tps$ for the first samples $\{X_{i}\}_{i=1}^{n}$, and $\mu =\pone' \mathbf{1}_{\rdim}/\sqrt{\rdim}$ for the second samples $\{Y_{i}\}_{i=1}^{m}$, where $\pone'$ ranges from $0$ to $3$ and $\mathbf{1}_{\rdim}$ is a vector of length $\rdim$ with all entries being $1$. The results are shown in Figure~\ref{fig:mvnorm-mean-shift2}, where it can be seen that Hotelling's $T^2$ test and the energy test outperform all the other tests, the former in particular in lower dimensional cases.

\begin{figure}[!hbt]
	\centering
	\includegraphics[width=\textwidth]{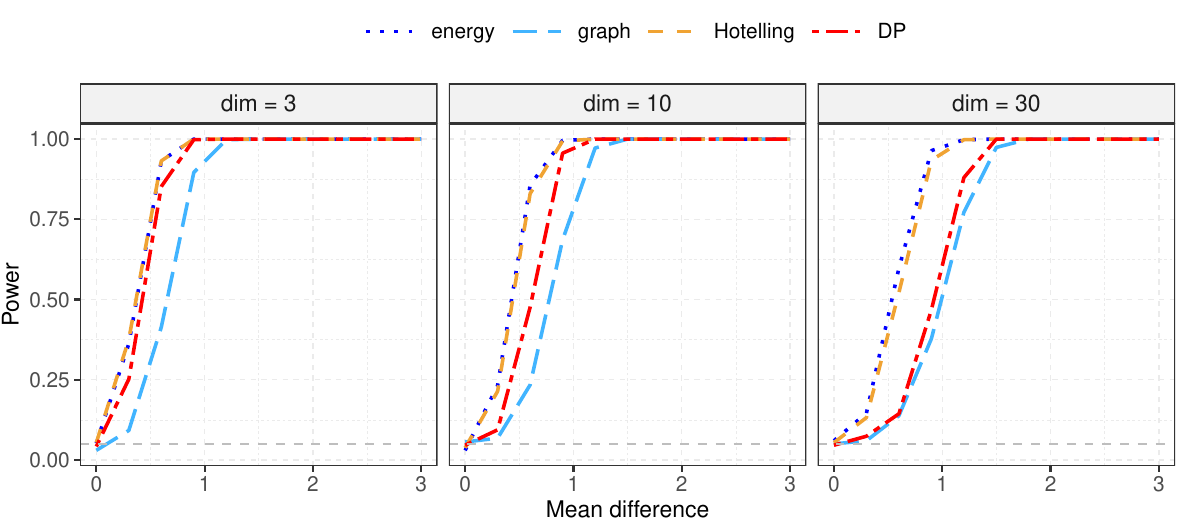}
	\caption{Power comparison for increasing values of mean difference $\pone'$ for two samples of $\rdim$-dimensional random vectors sampled from $N(\mu,\Sigma)$. Here, $\mu =\mathbf{0}_{\rdim}= (0,0,\dots,0)\tps$ for the first samples and $\mu =\pone' \mathbf{1}_{\rdim}/\sqrt{\rdim}$ for the second samples; $\Sigma = I_{\rdim}$ for both samples. The dashed grey line denotes the significance level $0.05$.}
	\label{fig:mvnorm-mean-shift2}
\end{figure}

For the last simulation scenario for comparing two samples of random networks in Section~\ref{sec:test_simu}, in addition
to the previous results we conducted simulations with larger sample sizes $n=m=200$. The results are shown in Figure~\ref{fig:network200}.

\begin{figure}[hbt!]
	\centering
	\includegraphics[width=.5\textwidth]{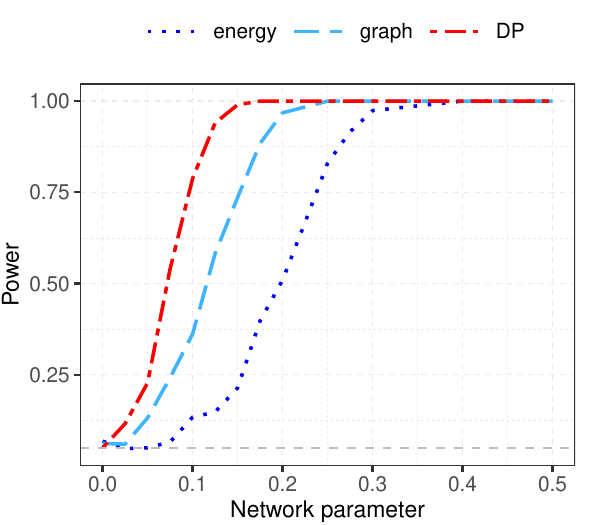}
	\caption{Power comparison for increasing values of $\psev$ for two samples of random networks with 200 nodes from the preferential attachment model, with larger sample sizes $n=m=200$. The attachment function is  proportional to $k^{\psev}$ with  $\psev=0$ for the first samples, 
		and $\psev$ increasing from $0$ to $0.5$ for the second samples. The dashed grey line denotes the significance level $0.05$.}
	\label{fig:network200}
\end{figure}
% \FloatBarrier 

\section{Additional Simulations for Multimodal Multivariate Data} \label{sec:simu_mv}  

We performed additional simulations illustrating how the proposed transport ranks behave for multivariate data that follow multi-modal distributions. 
Specifically, samples $\{\obj_{\subidx}\}_{\subidx=1}^n$ of $n=500$ independent observations were drawn from $2$-dimensional Gaussian mixture distributions: $\obj_{\subidx}\sim 0.2 N(\bm\mu_1,\bm\Sigma_1) + 0.3 N(\bm\mu_2,\bm\Sigma_2) + 0.5 N(\bm\mu_3,\bm\Sigma_3)$. 
We considered three scenarios, with mean vectors 
\begin{gather}\label{eq:2dGaussMix-mean}
\bm\mu_1 = (-8,6)^\top, \bm\mu_2 = (-3,-2)^\top, \bm\mu_3 = (4,0)^\top
\end{gather} 
and covariance matrices with increasing eigenvalues, 
\begin{equation}\label{eq:2dGaussMix-cov}
\aligned
\bm\Sigma_1 &= \bm{R}(\pi/6) \diag(a^2, 0.3 a^2) \bm{R}(\pi/6)^\top,\\
\bm\Sigma_2 &= \bm{R}(-\pi/6) \diag(a^2, 0.3 a^2) \bm{R}(-\pi/6)^\top\\ 
\text{ and } \bm\Sigma_3 &= \diag(0.8 a^2, 0.8 a^2)
\endaligned
\end{equation} 
for $a=1,2,3$ in the three scenarios, respectively, where $\bm{R}(\theta) = \begin{pmatrix} \cos\theta & -\sin\theta\\ \sin\theta & \cos\theta \end{pmatrix}$. 
The distance profiles $\hfi$ \eqref{eq:FhatO} and transport ranks $\hrank_{\obj_{\subidx}}$ \eqref{eq:hrank} were computed for each observation for  the Euclidean metric $d$.

Results for the three scenarios are displayed in the top (a=1), middle (a=2) and bottom (a=3) panels in Figure~\ref{fig:2dGaussMix}, respectively. 
In the first scenario ($a=1$), points from the three Gaussian distributions are well separated and the transport ranks are highest  in the subset drawn from the Gaussian distribution with the highest weight 0.5 in the Gaussian mixture distribution. 
In all three scenarios, points drawn from the Gaussian distribution with the lowest weight 0.2 always have lower transport ranks,  indicating that they are somewhat outlying. 
As the eigenvalues get larger, the points from the three Gaussian distributions become closer, and the  most central points with highest transport ranks move towards a  mixture of points drawn from the two Gaussian distributions with higher weights 0.3 and 0.5. 

\begin{figure}[!htb]
	\centering
	\begin{subfigure}[b]{0.32\linewidth}
		\centering
		\includegraphics[width=\linewidth]{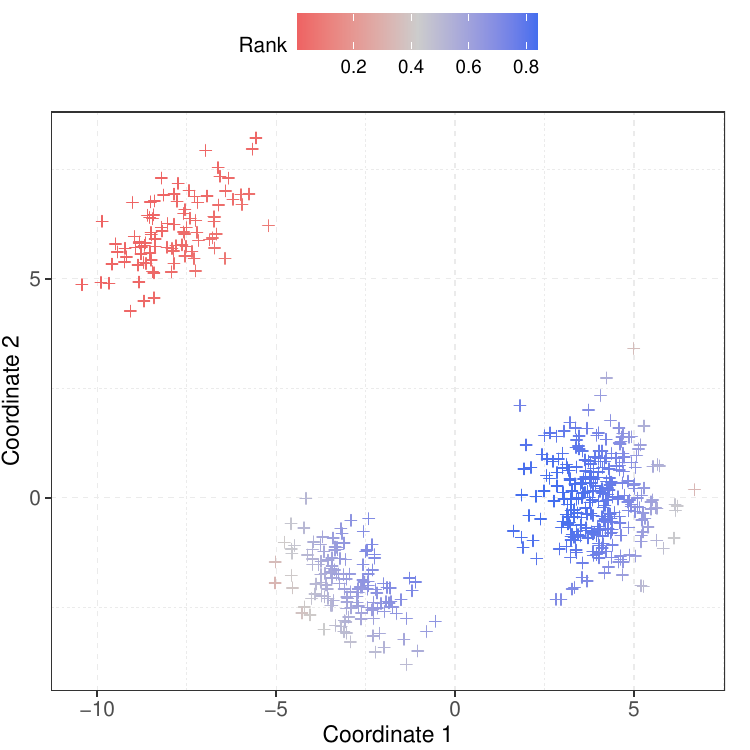}
	\end{subfigure}
	\begin{subfigure}[b]{0.32\linewidth}
		\centering
		\includegraphics[width=\linewidth]{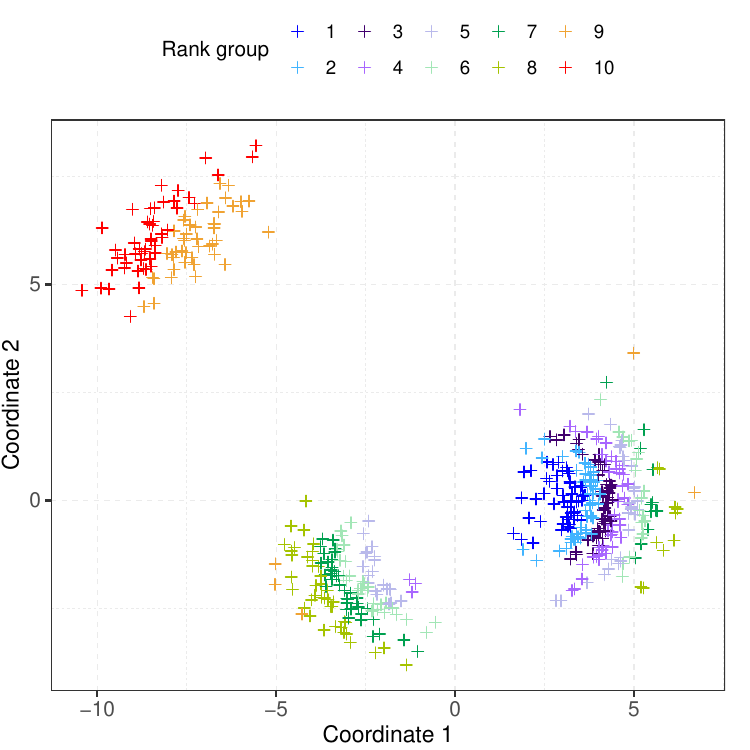}
	\end{subfigure}
	\begin{subfigure}[b]{0.32\linewidth}
		\centering
		\includegraphics[width=\linewidth]{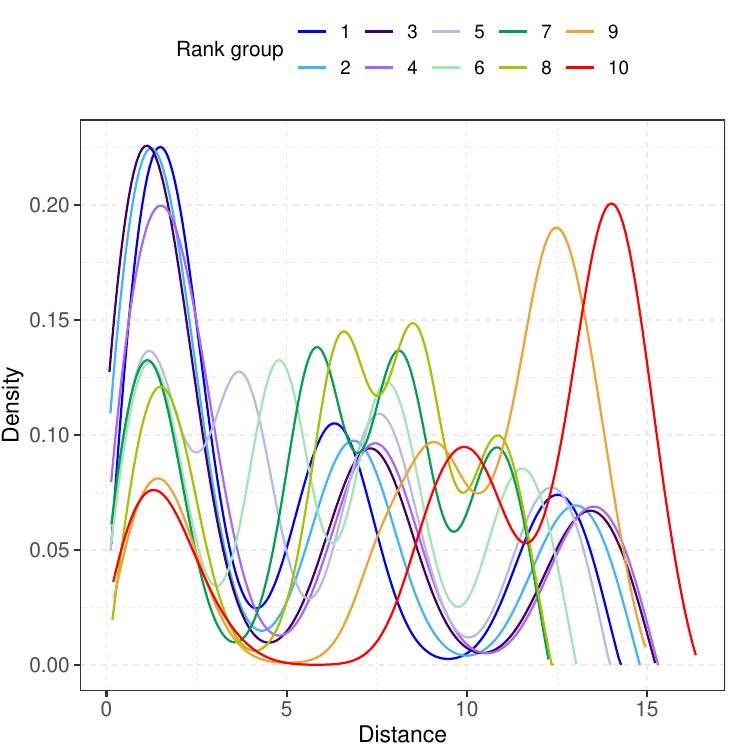}
	\end{subfigure}
	\begin{subfigure}[b]{0.32\linewidth}
		\centering
		\includegraphics[width=\linewidth]{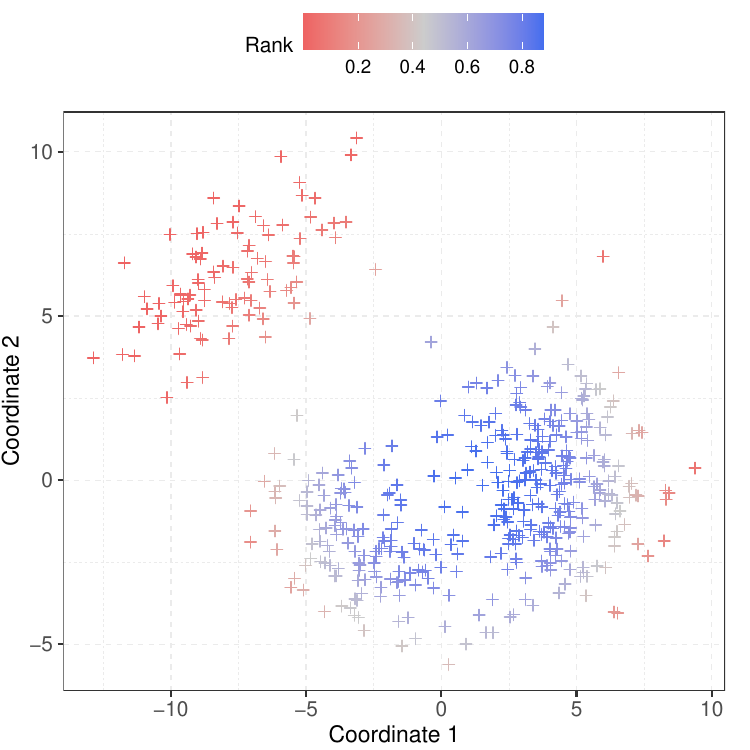}
	\end{subfigure}
	\begin{subfigure}[b]{0.32\linewidth}
		\centering
		\includegraphics[width=\linewidth]{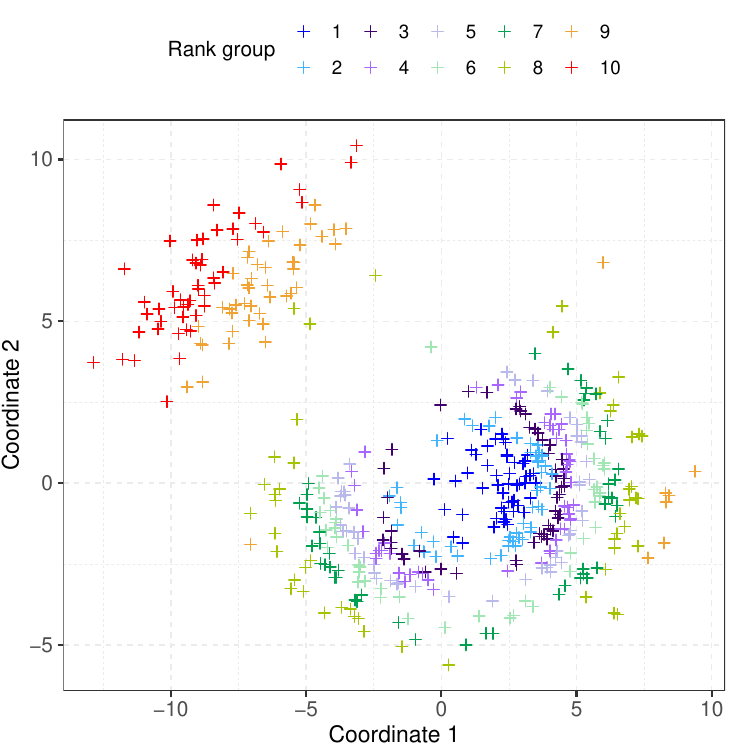}
	\end{subfigure}
	\begin{subfigure}[b]{0.32\linewidth}
		\centering
		\includegraphics[width=\linewidth]{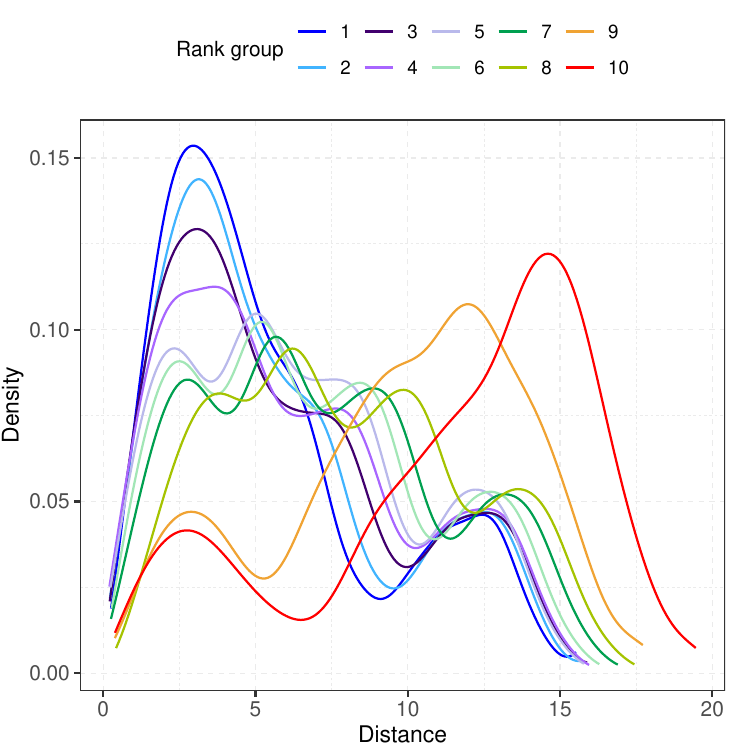}
	\end{subfigure}
	\begin{subfigure}[b]{0.32\linewidth}
		\centering
		\includegraphics[width=\linewidth]{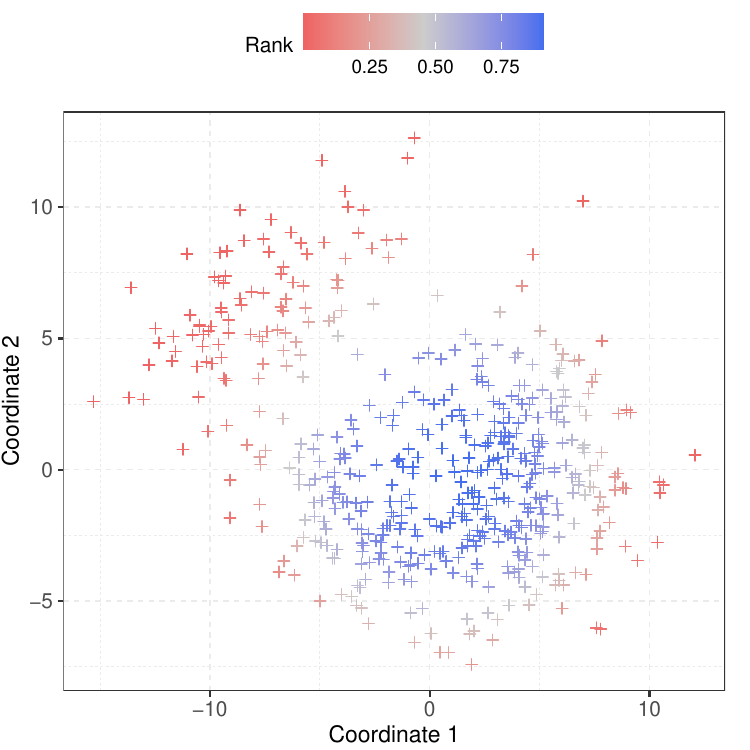}
	\end{subfigure}
	\begin{subfigure}[b]{0.32\linewidth}
		\centering
		\includegraphics[width=\linewidth]{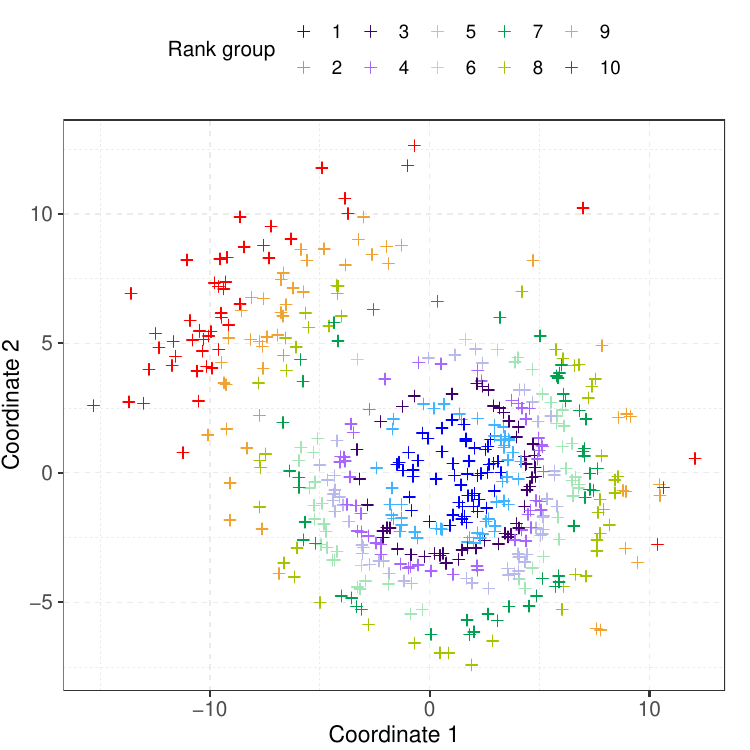}
	\end{subfigure}
	\begin{subfigure}[b]{0.32\linewidth}
		\centering
		\includegraphics[width=\linewidth]{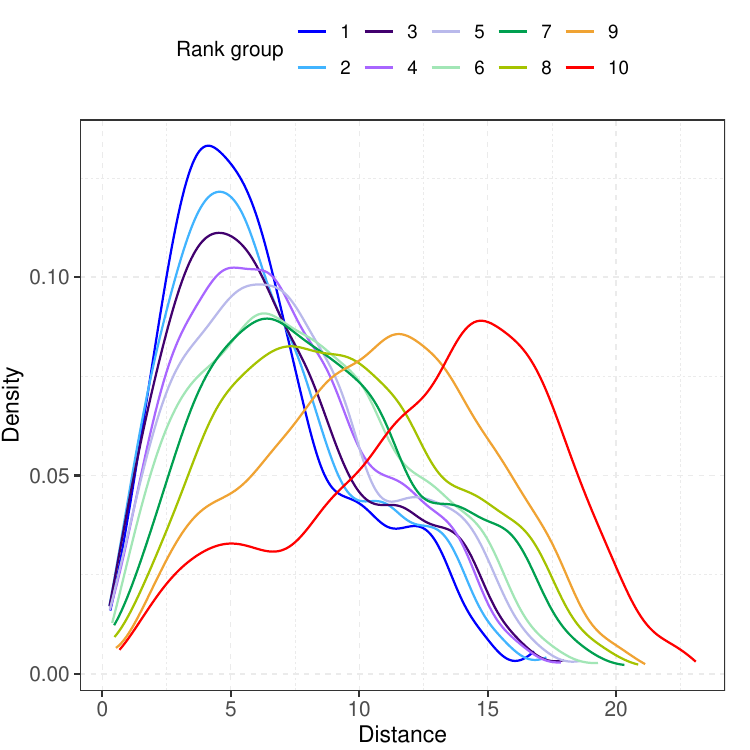}
	\end{subfigure}
	\caption{Analysis of samples of $n=500$ observations generated from $2$-dimensional Gaussian mixture distributions $0.2 N(\bm\mu_1,\bm\Sigma_1) + 0.3 N(\bm\mu_2,\bm\Sigma_2) + 0.5 N(\bm\mu_3,\bm\Sigma_3)$, where the mean vectors $\{\bm\mu_j\}_{j=1}^{3}$ are the same across the three rows and the covariance matrices have increasing eigenvalues from top to bottom. 
		Columns: Scatterplots of simulated samples, where the points are colored according to their transport ranks \eqref{eq:hrank} (left); grouped according to quantiles of transport ranks as described below \eqref{eq:dwass} (middle); 
		Wasserstein barycenters of the distance profiles within each group represented by density functions (right). } \label{fig:2dGaussMix}
\end{figure}

\section{Additional Simulations for Distributional Data}\label{sec:simu_distn}

For these additional simulations, we considered a sample of $n=500$ one-dimensional Gaussian distributions $\{N(\mu_{\subidx},\sigma_{\subidx}^2)\}_{\subidx=1}^{n}$ that were generated from a mixture of two distributions of distributions, i.e., there are two groups of distributions. 
Specifically, we first generated $Z_{\subidx}\sim\text{Bernoulli}(p)$ and then sampled $\mu_{\subidx}\sim N(-2,0.5^2)$ if $Z_{\subidx}=1$ and $\mu_{\subidx}\sim N(2,0.5^2)$ if $Z_{\subidx}=0$, choosing $\sigma_{\subidx}\sim\mathrm{Gamma}(2,4)$. 
Here, $\mathrm{Gamma}(\alpha,\beta)$ denotes a gamma distribution with shape $\alpha$ and rate $\beta$. In addition, we considered balanced and unbalanced designs with $p=0.5$ and $0.2$, respectively. 

Our findings were as follows: When the design is balanced, as shown in the bottom three panels of Figure~\ref{fig:distnSimu}, distributions sampled from the two groups have similar transport ranks, and distributions lying closer to the empirical barycenter of all the $500$ distributions have the highest  transport ranks and thus are classified as more central.  
In contrast, when $p=0.2$ and the design is unbalanced, as shown in the top three panels of Figure~\ref{fig:distnSimu}, the  observations in the sample with the highest transport ranks  all lie in the larger subsample (represented by triangles in Figure~\ref{fig:distnSimu}) and the distributions in the smaller subsample (represented by circles in Figure~\ref{fig:distnSimu}) all have lower transport ranks as compared to those in the larger subsample and are thus considered to be somewhat outlying.

\begin{figure}[htbp]
	\centering
	\includegraphics[width=0.325\linewidth]{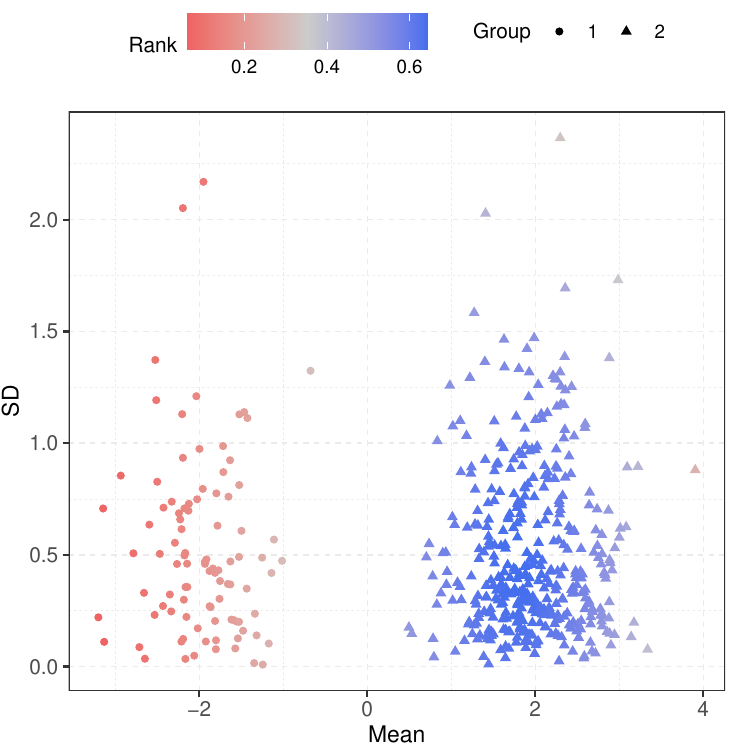}
	\includegraphics[width=0.325\linewidth]{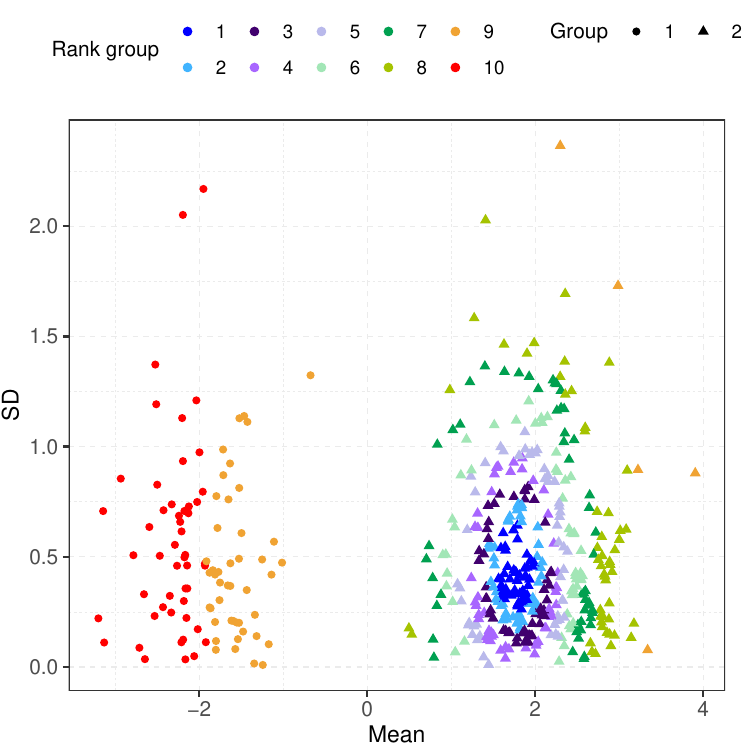}
	\includegraphics[width = 0.325\linewidth]{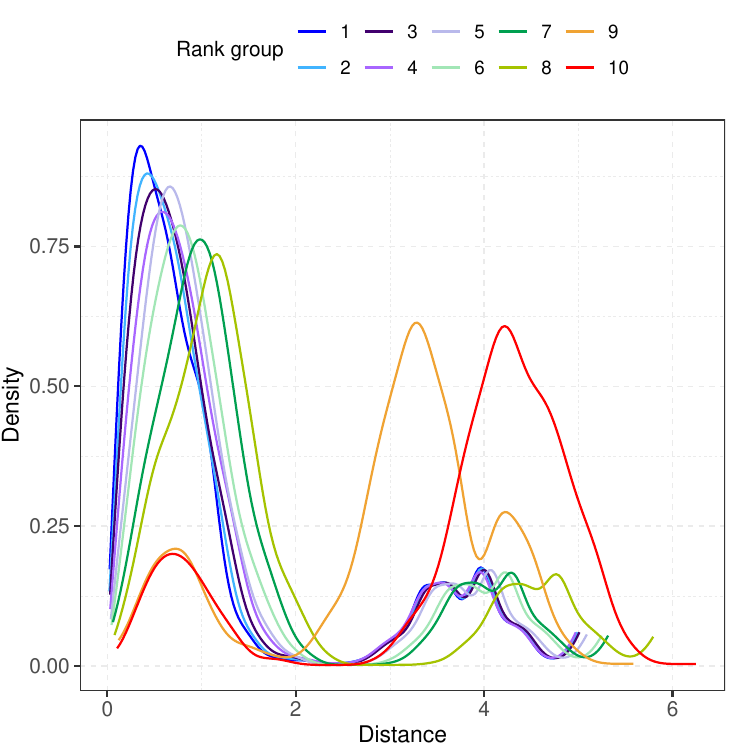}\\
	\includegraphics[width = 0.325\linewidth]{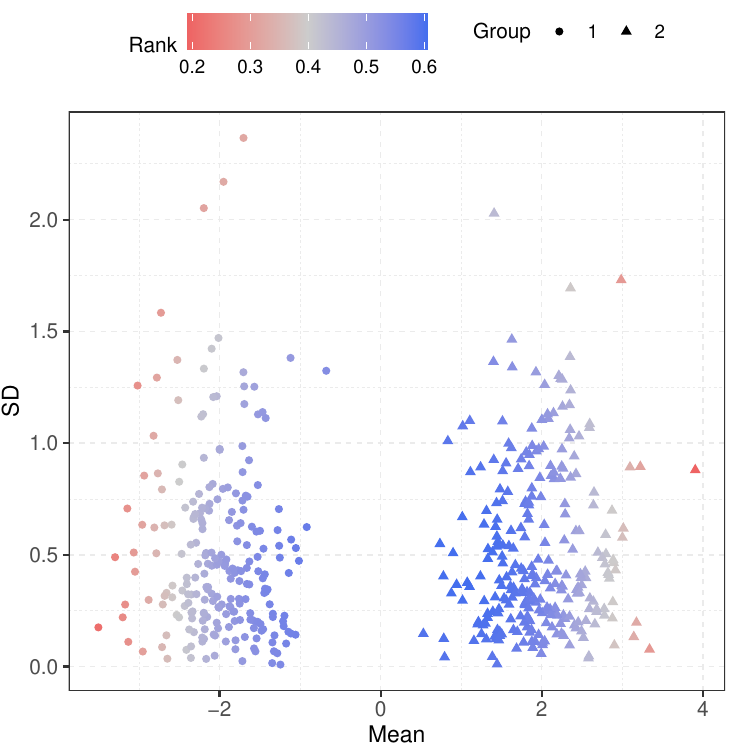}
	\includegraphics[width = 0.325\linewidth]{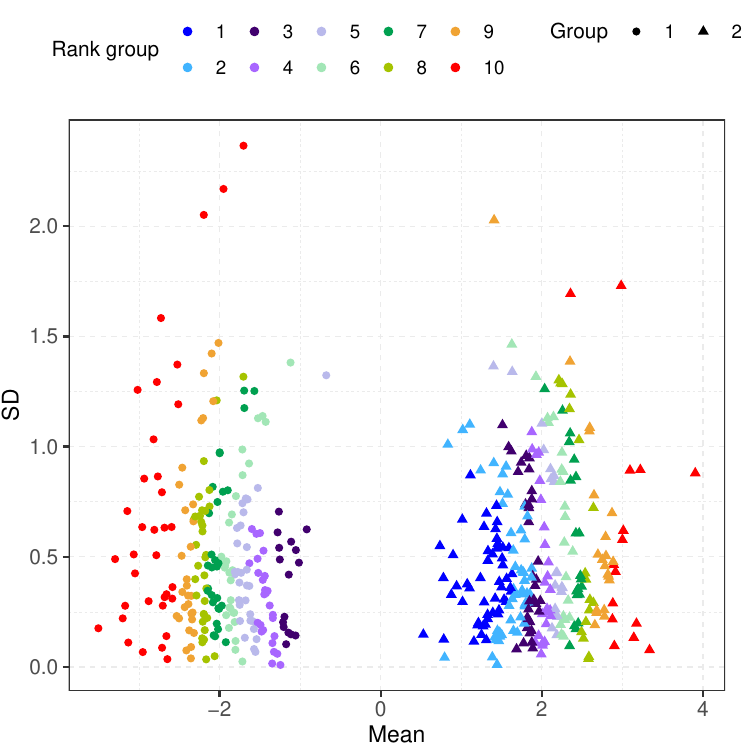}
	\includegraphics[width = 0.325\linewidth]{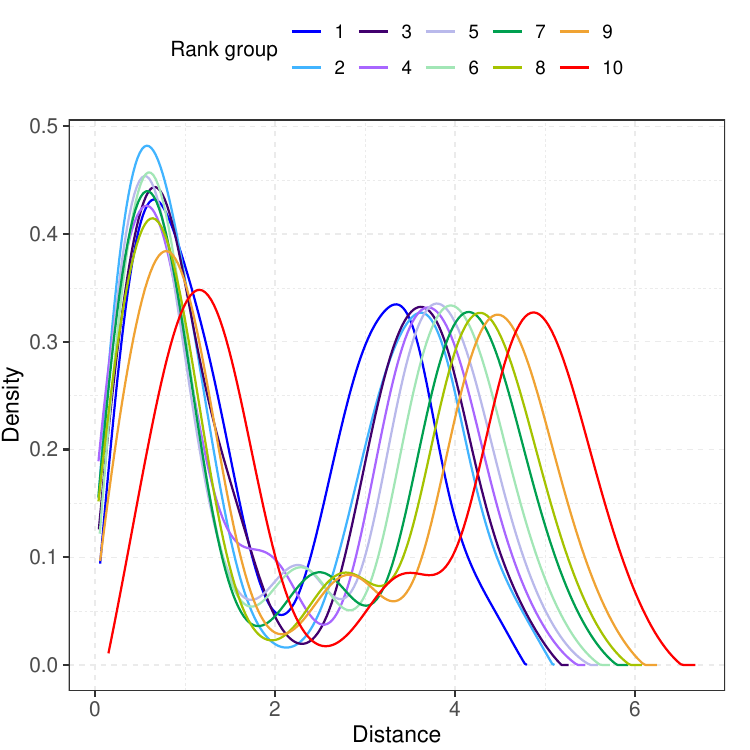} 
	\caption{
		Analysis of samples of $n=500$ one-dimensional Gaussian distributions $\{N(\mu_{\subidx},\sigma_{\subidx}^2)\}_{\subidx=1}^{n}$ where $\mu_{\subidx}\sim N(-2,0.5^2)$ if $Z_{\subidx}=1$ (Group 1) and $\mu_{\subidx}\sim N(2,0.5^2)$ if $Z_{\subidx}=0$ (Group 2) and $Z_{\subidx}\sim\text{Bernoulli}(p)$ with $p=0.2$ (top) and $0.5$ (bottom), respectively. 
		Columns: Scatterplots of mean $\mu_{\subidx}$ and standard deviation (SD) $\sigma_{\subidx}$, where the points are colored according to the transport ranks \eqref{eq:hrank} (left) and grouped according to quantiles of transport ranks as per \eqref{eq:dwass} (middle); Wasserstein barycenters of the distance profiles $\hfi$ for the distributions within each rank group (right).
	}
	\label{fig:distnSimu}
\end{figure}

\section{Additional Details for Human Mortality Data}\label{sec:mort_supp}

Age-at-death distributions in 2000 for males and females are shown in Figure~\ref{fig:mort2000_data}, represented by density functions.

\section{Distance Profiles for Networks: Manhattan Yellow Taxi Data}\label{sec:taxi}

Yellow taxi trip records in New York City (NYC) including pick-up and drop-off dates/times, pick-up and drop-off locations, and driver-reported passenger counts are available at 
\url{http://www1.nyc.gov/site/tlc/about/tlc-trip-record-data.page}. 
We focus on the data pertaining to Manhattan and, excluding Governor's Island, Ellis Island, and Liberty Island, divide the remaining 66 zones of Manhattan into 13 regions (Table~\ref{tab:taxiRgns}). 

Of interest are networks that represent how many people traveled between these areas during a day. 
To this end, we constructed networks for yellow taxi trips between the 13 regions for each day in the year 2019, obtaining a 13-dimensional graph adjacency matrix indexed by day, where each entry corresponds to the edge weight that reflects  the total number of passengers traveling between the two corresponding regions within the given day. The edge weights are then normalized by the maximum edge weight for each day so that they lie in $[0,1]$. We choose the Frobenius metric $\dfrob$ as metric between the resulting weighted graph adjacency matrices,
\bal\label{eq:dfrob}
\dfrob(\mathbf{R}_1,\mathbf{R}_2) = \left\{\trace\left[(\mathbf{R}_1-\mathbf{R}_2)(\mathbf{R}_1-\mathbf{R}_2)^\top\right]\right\}\half, \text{ for } \mathbf{R}_1, \mathbf{R}_2 \in\real^{13\times 13}. \eal

Weekdays are found to have higher transport ranks and are more central, while weekends have lower transport ranks and are more outlying (Figure~\ref{fig:taxi_dpMds}). 
Some ``abnormal'' days are highlighted, which are weekdays yet with distance profiles more similar to weekends and higher transport ranks than average weekdays.  Among these, Memorial Day, Independence Day, July 5, Veterans Day, and New Year's Eve and New Year's Day are weekdays but also holidays or Fridays after a holiday.
Every weekday between September 23 and September 30, including September 23--26, was designated as a ``gridlock alert day'' by the NYC Department of Transportation, due to the UN General Assembly meetings held from September 24 through 30, i.e., these are the days likely to feature the heaviest traffic of the year. 
Hence, it is not surprising that these days are closer to weekends in terms of both distance profiles and transport ranks.

Interestingly, the two-dimensional profile MDS plot of the networks in Figure~\ref{fig:taxi_dpMds} exhibits a similar horseshoe shape as seen before 
for the human mortality distance profiles in Figure~\ref{fig:mort2000_mds}, emphasizing a nearly one-dimensional strict ordering of the distance profiles, indicating that distance profiles lie on a lower-dimensional manifold.

\begin{figure}[H]
	\centering
	\includegraphics[width=.49\linewidth]{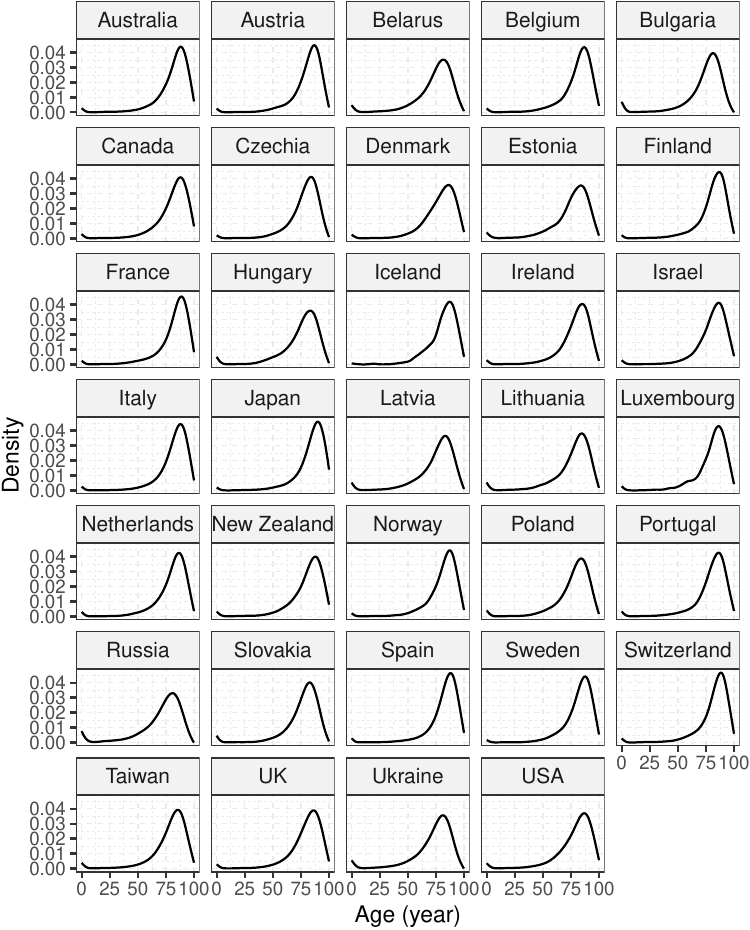}
	\includegraphics[width=.49\linewidth]{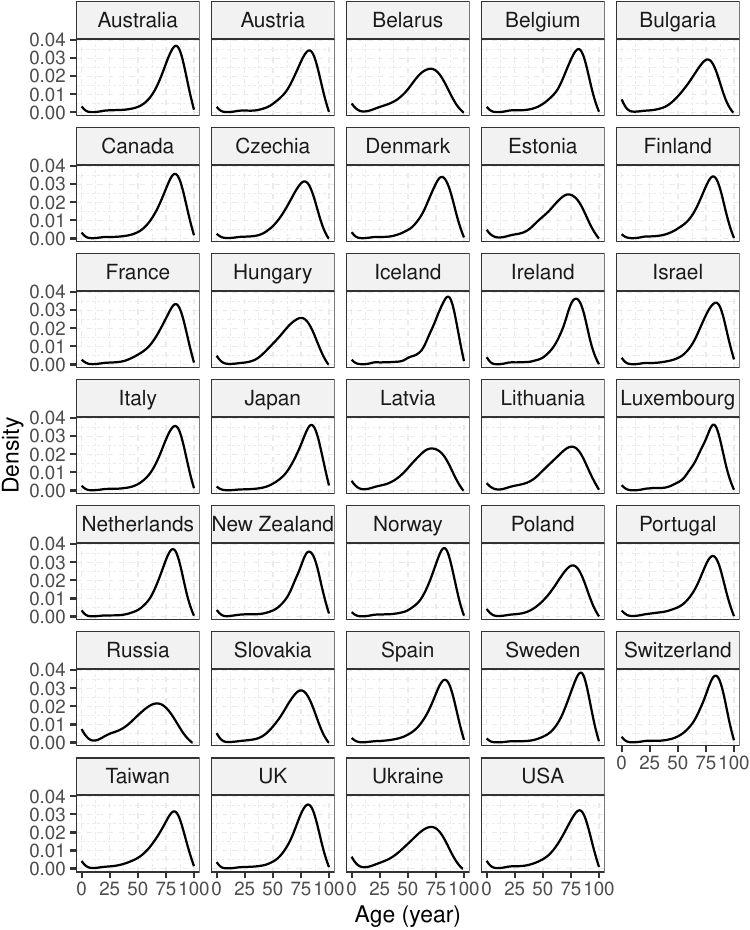}
	\caption{Age-at-death distributions in 2000 for females (left) and for males (right), represented by density functions.}
	\label{fig:mort2000_data}
\end{figure}

\begin{table}[H]
	\caption{\label{tab:taxiRgns}Thirteen regions of Manhattan Borough.} 
	\centering
	\begin{tabular}{c p{0.85\linewidth}}
		\toprule
		Region & Zones included \\
		\midrule
		1 & Battery Park, Battery Park City, Financial District North, Financial District South, Seaport, TriBeCa/Civic Center, World Trade Center \\ 
		2 & Greenwich Village North, Greenwich Village South, Hudson Sq, Little Italy/NoLiTa, Meatpacking/West Village West, SoHo, West Village \\ 
		3 & Alphabet City, Chinatown, East Village, Lower East Side, Two Bridges/Seward Park \\ 
		4 & Clinton East, Clinton West, East Chelsea, Flatiron, West Chelsea/Hudson Yards \\ 
		5 & Garment District, Midtown Center, Midtown North, Midtown South, Penn Station/Madison Sq West, Times Sq/Theatre District, Union Sq \\ 
		6 & Gramercy, Kips Bay, Midtown East, Murray Hill, Stuy Town/Peter Cooper Village, Sutton Place/Turtle Bay North, UN/Turtle Bay South \\ 
		7 & Bloomingdale, Lincoln Square East, Lincoln Square West, Manhattan Valley, Upper West Side North, Upper West Side South \\ 
		8 & Lenox Hill East, Lenox Hill West, Roosevelt Island, Upper East Side North, Upper East Side South, Yorkville East, Yorkville West \\ 
		9 & Hamilton Heights, Manhattanville, Morningside Heights \\ 
		10 & Central Harlem, Central Harlem North \\ 
		11 & East Harlem North, East Harlem South, Randalls Island \\ 
		12 & Highbridge Park, Inwood, Inwood Hill Park, Marble Hill, Washington Heights North, Washington Heights South \\ 
		13 & Central Park \\ 
		\bottomrule
	\end{tabular}
\end{table}

\begin{figure}[H]
	\centering
	\begin{subfigure}[t]{0.325\linewidth}
		\centering
		\includegraphics[width=\linewidth]{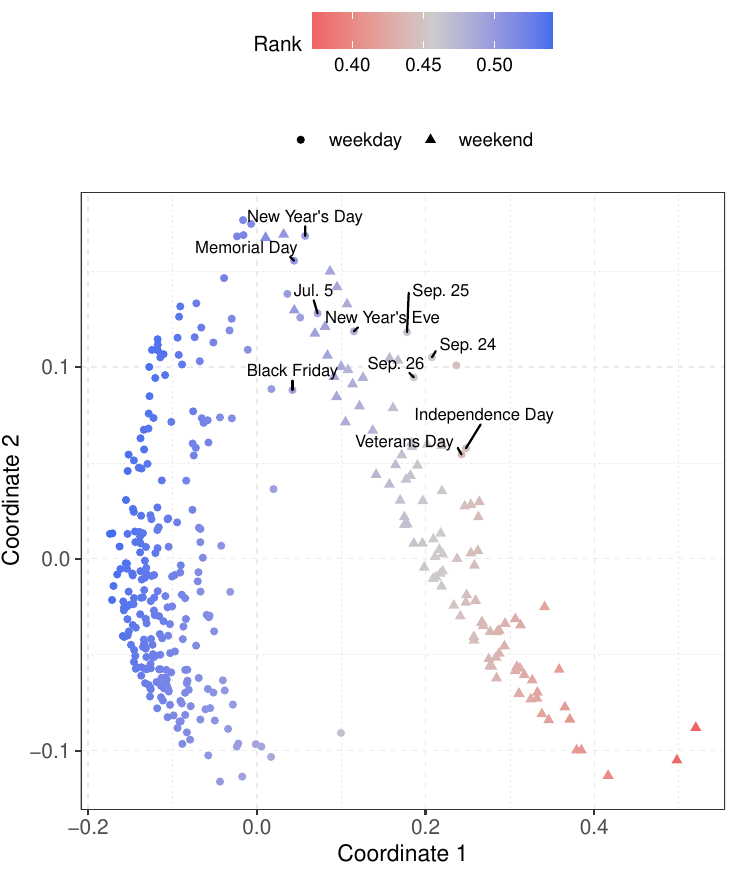}
		% \caption{Ranks.}
	\end{subfigure}
	\begin{subfigure}[t]{0.325\linewidth}
		\centering
		\includegraphics[width=\linewidth]{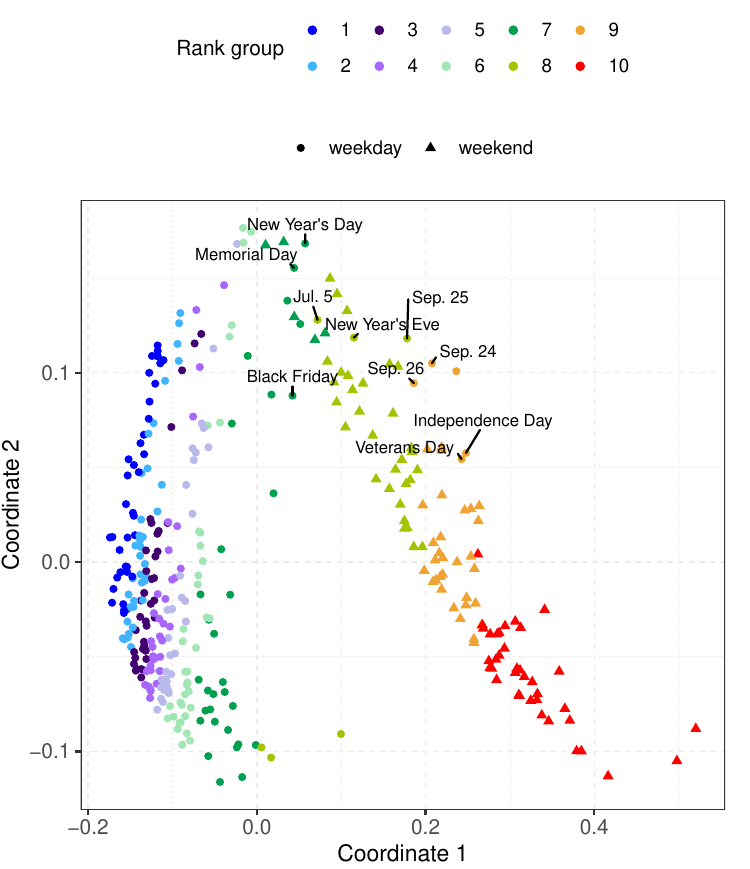}
		% \caption{Clustering.}
	\end{subfigure}
	\begin{subfigure}[t]{0.325\linewidth}
		\centering
		\includegraphics[width=\linewidth]{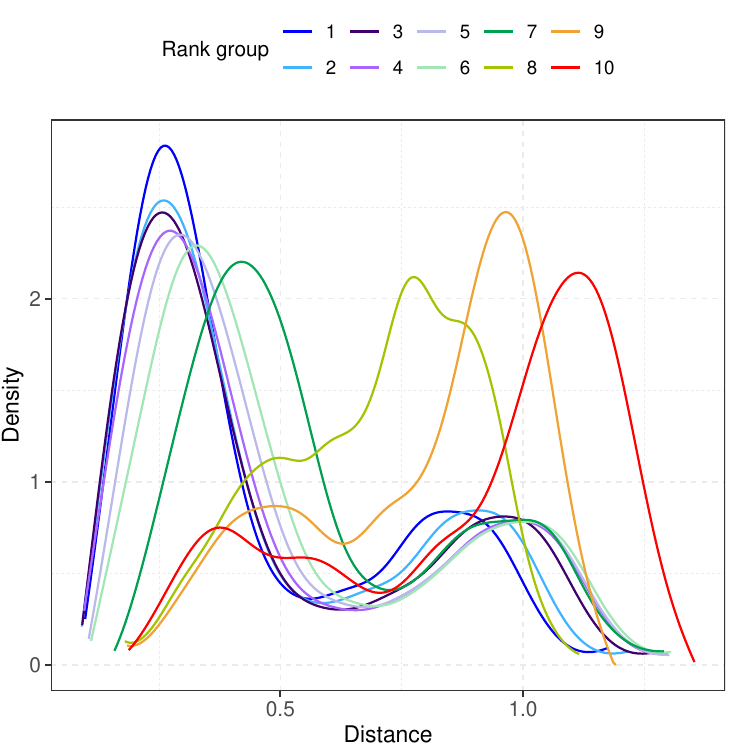}
	\end{subfigure}
	\caption{Two-dimensional MDS of the distance profiles of the daily Manhattan Yellow Taxi transport networks in 2019 with normalized edge weights, where the points are colored according to their transport ranks \eqref{eq:hrank} (left) and grouped according to the quantiles of transport ranks as described in \eqref{eq:dwass} (middle); Wasserstein barycenters of the distance profiles within each group (right).} \label{fig:taxi_dpMds}
\end{figure}

\clearpage
\references

\end{document}